\newcount\Comments  

\newif\iflong
\longtrue

\newif\iffinal
\finaltrue

\iffinal
\fi

\newif\ifold

\documentclass[a4paper,UKenglish,cleveref,autoref, thm-restate]{lipics-v2021}

\nolinenumbers
\hideLIPIcs

\usepackage{graphicx} \usepackage{amsmath}
\usepackage{mathtools}
\usepackage{amsthm}
\usepackage{subcaption}
\usepackage{amssymb}
\usepackage{tikz}
\usepackage{xspace}
\usepackage{color}
\usepackage{colortbl}

\newcommand{\mg}{\mathcal{G}}
\newcommand{\Oh}{\mathcal{O}}
\newcommand{\fes}{\mathrm{fes}}
\newcommand{\mc}{\mathcal{C}}

\newcommand{\minlab}{frugal\xspace}

\newcommand{\prob}[3]{\begin{quote}  \textsc{#1}\\  \textbf{Input:} #2\\  \textbf{Question:} #3\end{quote}}

\newcommand{\RGR}{\textsc{RGR}\xspace}

\newcommand{\RGDlong}{\textsc{Reachability Graph Realizability}\xspace}
\newcommand{\RGD}{\textsc{RGR}\xspace}

\newcommand{\URGDlong}{\textsc{Undirected} \RGDlong}
\newcommand{\URGD}{\textsc{U}\RGD}

\newcommand{\DRGDlong}{\textsc{Directed} \RGDlong}
\newcommand{\DRGD}{\textsc{D}\RGD}

\newcommand{\str}{\textsc{Strict}\xspace}
\newcommand{\strict}{\str}
\newcommand{\nstr}{\textsc{Non-strict}\xspace}

\newcommand{\simp}{\textsc{Simple}\xspace}
\newcommand{\pro}{\textsc{Proper}\xspace}
\newcommand{\happy}{\textsc{Happy}\xspace}
\newcommand{\any}{\textsc{Any}\xspace}

\newcommand{\SC}{\textsc{Set Cover}\xspace}
\newcommand{\mf}{\mathcal{F}}

\newcommand{\NP}{\textrm{NP}}
\newcommand{\bth}{$\textrm{NP} \subseteq \textrm{coNP/poly}$}

\newcommand{\pa}{\Gamma}

\newcommand{\ol}[1][x]{{\overline{#1}}}

\usepackage{cleveref}

\usepackage{algpseudocode}
\usepackage[linesnumbered,lined,commentsnumbered,ruled]{algorithm2e}

\usetikzlibrary{calc,shapes,through,arrows,matrix,backgrounds,fit,decorations.pathreplacing}
\usetikzlibrary{decorations.pathreplacing, decorations.pathmorphing, calc, shapes.arrows, positioning, matrix,fit}

\iffinal
\newcommand{\tnew}[1]{{\color{black}#1}}

\newcommand{\nnew}[1]{{\color{black}#1}}
\else
\newcommand{\tnew}[1]{{\color{blue}#1}}

\newcommand{\nnew}[1]{{\color{gray}#1}}
\fi

\iflong
\title{Recognizing and Realizing Temporal Reachability Graphs}\else
\title{Recognizing and Realizing Temporal Reachability Graphs}
\fi

\author{Thomas Erlebach}{Department of Computer Science, Durham University, UK}{thomas.erlebach@durham.ac.uk}{https://orcid.org/0000-0002-4470-5868}{}
\author{Othon Michail}{Department of Computer Science, University of Liverpool, UK}{othon.michail@liverpool.ac.uk}{https://orcid.org/0000-0002-6234-3960}{}

\author{Nils Morawietz}{Friedrich Schiller University Jena, Institute of Computer Science, Germany\\ LaBRI, Université de Bordeaux, France}{nils.morawietz@uni-jena.de}{https://orcid.org/0000-0002-7283-4982}{}

\authorrunning{Thomas Erlebach, Othon Michail, Nils Morawietz}

\begin{document}

\maketitle

\begin{abstract}
A temporal graph $\mathcal{G}=(G,\lambda)$ can be represented by an underlying
graph $G=(V,E)$ together with a function $\lambda$ that
assigns to each edge $e\in E$ the set of time steps
during which $e$ is present. The reachability graph
of $\mathcal{G}$ is the directed graph $D=(V,A)$ with
$(u,v)\in A$ if only if there is a temporal path from $u$ to~$v$.
We study
the Reachability Graph
Realizability (RGR) problem that asks whether a given directed graph
$D=(V,A)$ is the reachability graph of some
temporal graph. The question can be asked for undirected or
directed temporal graphs, for reachability
defined via strict or non-strict temporal paths, and
with or without restrictions on $\lambda$ (proper, simple,
or happy). 
Answering an open question
posed by Casteigts et al.~(Theoretical Computer Science 991 (2024)), 
we show that all variants
of the problem are $\NP$-complete, except for two variants that
become trivial in the directed case.
For undirected temporal graphs, we consider the complexity of the problem with respect to the solid graph, that is, \tnew{the} graph containing all edges that
could potentially receive a label in any realization.
We show that the RGR problem is polynomial-time solvable if the solid
graph is a tree and fixed-parameter tractable with respect
to the feedback edge set number of the solid graph.
As we show, the latter parameter can presumably not be replaced by smaller parameters like feedback vertex set or treedepth, since  the problem is W[2]-hard with respect
to these parameters. 
\end{abstract}
\section{Introduction}
Temporal graphs are graphs whose edge set can change over time. The vertex
set is often assumed to be fixed, and the edge set can differ from one
time step to the next. The study of temporal graphs has attracted significant
attention in recent years~\cite{casteigts2012time,michail2016introduction,michail2018elements}.
A temporal graph $\mathcal{G}$ with vertex set $V$ can be
represented as a sequence $(G_i)_{i\in [L]}$, where $G_i=(V,E_i)$ is the
graph containing the edges that are present in time step~$i$.
An alternative way to represent
$\mathcal{G}$ is to specify
a graph $G=(V,E)$ together
with a labeling function $\lambda:E\to 2^{[L]}$ that maps each edge
$e\in E$ to the (possibly empty) set of time steps in which $e$ is present.
We write $\mathcal{G}=(G,\lambda)$ in this case.
If an edge $e$ is present in time step~$i$, we refer to $(e,i)$
as a time edge. A strict \emph{temporal path} from $u$ to $v$ in
$\mathcal{G}$ is a sequence of time edges such that the edges
form a $u$-$v$-path in the underlying graph and the time steps
are strictly increasing. A non-strict temporal path is defined
analogously, except that the time steps only need to be non-decreasing.
Given a temporal graph $\mathcal{G}$, the reachability relation (with
respect to strict or non-strict temporal paths) can
be represented as a directed graph $D=(V,A)$, called the
\emph{reachability graph}, with $(u,v)\in A$ for $u\neq v$ if
and only if there exists a (strict or non-strict, respectively)
temporal path from $u$ to $v$ in $\mathcal{G}$.
Note that self-loops, which would represent the trivial
reachability from a vertex to itself, are omitted from~$D$.

A natural question
is then which directed graphs $D=(V,A)$ can arise as reachability
graphs of temporal graphs, and how difficult it is to determine
for a given directed graph $D=(V,A)$ whether there exists
a temporal graph $\mathcal{G}$ with reachability graph~$D$.
This question was posed as an open problem by
Casteigts et al.~\cite[Open question~5]{DBLP:journals/tcs/CasteigtsCS24}.
In this paper, we answer this open question by showing that this
decision problem is NP-hard. Actually, there are a number
of variations of the question, as we may ask
for an undirected or a directed temporal graph~$\mathcal{G}$,
for a temporal graph with a restricted kind
of labeling (simple, proper, happy), and may consider reachability with
respect to strict or non-strict temporal paths. We show that all these
\tnew{variations are NP-hard if we ask for an undirected temporal graph,
and also if we ask for a directed temporal graph except for
two variations} (strict temporal
paths with arbitrary or simple labelings) that are known to become trivial in the directed case~\cite{D25}. See Table~\ref{table for complexity} for an overview
of these complexity results.

From the positive side, we present the following algorithmic
results. For a given digraph $D=(V,A)$, we refer to
$\{u,v\}$ for $u,v\in V$ as a \emph{solid} edge if both $(u,v)$
and $(v,u)$ are in~$A$. Let $G=(V,E)$ be the undirected
graph on $V$ whose edge set is the set of solid edges.
We refer to $G$ as the \emph{solid graph} of~$D$.
We show that all undirected problem variants can be solved in polynomial
time if the solid graph is a tree.
Furthermore, we give
an FPT algorithm with respect to the feedback edge set number
of the solid graph.
This parameter can presumably not be replaced by smaller parameters like feedback vertex set, treedepth, or pathwidth, since two undirected versions of our problem turn out to be W[2]-hard for these parameters. 

\begin{table}
\caption{The complexity results for all variants of \RGD.
Red cells indicate NP-hard cases and green cells indicate cases that are trivial (and thus polynomial-time solvable).}
\label{table for complexity}

\begin{subfigure}{0.48\textwidth}
\begin{tabular}{l|l|l}
 & \str & \nstr\\\hline
\any & \cellcolor{red!25}\Cref{hardness trianglefree} & 
\cellcolor{red!25}\Cref{hardness proper}\\\hline
\simp & \cellcolor{red!25}\Cref{hardness trianglefree}  & 
\cellcolor{red!25}\Cref{hardness proper}\\\hline
\pro & 
\multicolumn{2}{c}{\cellcolor{red!25}\Cref{hardness proper}}  \\\hline
\happy & 
\multicolumn{2}{c}{\cellcolor{red!25}\Cref{hardness proper}}\\\hline
\end{tabular}
\caption{undirected}
\end{subfigure}
\begin{subfigure}{0.48\textwidth}
\begin{tabular}{l|l|l}
 & \str & \nstr\\\hline
\any & \cellcolor{green!25}\Cref{easy parts} and \cite{D25} & 
 \cellcolor{red!25}\Cref{hardness directed fes}\\\hline
\simp & \cellcolor{green!25}\Cref{easy parts} and \cite{D25} & 
\cellcolor{red!25}\Cref{hardness directed fes}\\\hline
\pro & 
\multicolumn{2}{c}{\cellcolor{red!25}\Cref{hardness directed fes}}  \\\hline
\happy & 
\multicolumn{2}{c}{\cellcolor{red!25}\Cref{hardness directed fes}}\\\hline
\end{tabular}
\caption{directed}
\end{subfigure}

\end{table}

\subparagraph{Related work}
Casteigts et al.~\cite{DBLP:journals/tcs/CasteigtsCS24}
studied the relationships between the classes of reachability graphs
that arise from undirected temporal graphs if different restrictions
are placed on the graph (proper, simple, happy~\cite{DBLP:journals/tcs/CasteigtsCS24}) and depending on
whether strict or non-strict temporal paths are considered.
They showed that reachability with respect to strict
temporal paths in arbitrary temporal graphs yields the widest
class of reachability graphs while reachability in happy
temporal graphs yields the narrowest class.
The class of reachability graphs that arise from proper temporal
graphs is the same as for non-strict paths in arbitrary temporal
graphs, and this class is larger than the class of reachability
graphs arising from non-strict reachability in simple temporal
graphs. Strict reachability in simple temporal graphs
was also shown to lie between the happy case and the general
strict case. D\"oring~\cite{D25} showed that the class
``strict \& simple'' is incomparable to ``non-strict \& simple''
and to ``non-strict / proper'', completing the picture of a
two-stranded hierarchy for undirected temporal graphs.
She also extended the study to directed temporal graphs and
showed that their classes of reachability graphs form a single-stranded
hierarchy from happy to strict \& simple, the latter being
equivalent to the general strict case.

Casteigts et al.~\cite{DBLP:journals/tcs/CasteigtsCS24}
posed the open question whether there is a characterization
of the directed graphs that arise as reachability graphs
of temporal graphs (or of some restricted subclass of
temporal graphs), and how hard it is to decide whether
a given directed graph is the reachability graph of
some temporal graph. These questions were posed
again by D\"oring~\cite{D25}, also in relation to the
setting of directed temporal graphs. 
This is in particular of interest because Casteigts et al.~\cite{CMW24} showed that several temporal graph problems can be solved in FPT time with respect to temporal parameters defined over the reachability graph.
In this paper
we resolve these
open questions regarding the complexity of all directed
and undirected variants.

\nnew{
Besides that, our work falls into the field of temporal graph realization problems.
In these problems, one is given some data about the behavior of a temporal graph, and the goal is to detect whether there actually is a temporal graph with this behavior (and to compute such a temporal graph if one exists).}
Klobas et al.~\cite{Klobas/SAND24} introduced the
problem of deciding for a given matrix of fastest
travel durations and a period $\Delta$ whether there exists a
simple temporal graph $\mathcal{G}$ with period $\Delta$ with the property that the duration
of the fastest temporal path between any pair of
nodes in $\mathcal{G}$ is equal to the value specified in the input
matrix. 
Erlebach et al.~\cite{EMW/SAND24} extended
the problem to a multi-label version, assuming that
the input specifies a bound $\ell$ on the maximum number
of labels per edge.

\nnew{
Motivated by the constraints in the design of transportation networks, Mertzios et al.~\cite{MMS24} considered the modified version of the problem (with one label per edge) where the fastest temporal path between any pair of
nodes is
\tnew{only} upper-bounded by the value specified in the input
matrix.
}

The problem of generating a temporal graph realizing a given reachability graph is an instance of the more general class of temporal network design problem. In such problems, the objective is to construct a temporal network that satisfies specified constraints while optimizing certain network-quality measures. These problems have natural applications in transportation and logistics, communication networks, social networks, and epidemiology. In an early example studied by Kempe et al.~\cite{kempe2000connectivity}, the goal was to reconstruct a temporal labeling restricted to a single label per edge, so that a designated root reaches via temporal paths all vertices in a set $P$ while avoiding those in a set $N$. For multi-labeled temporal graphs, Mertzios et al.~\cite{mertzios2019temporal} studied the problem of designing a temporal graph that preserves all reachabilities or paths of an underlying static graph while minimizing either the \emph{temporality} (the maximum number of labels per edge) or the \emph{temporal cost} (the total number of labels used).
\tnew{G{\"{o}}bel et al.~\cite{GCV91} showed that it is NP-complete to decide whether
the edges of a given undirected graph can be labeled with a single label per edge
in such a way that each vertex can reach every other vertex via a strict temporal
path.}
Other studies have focused on variants of minimizing edge deletions~\cite{enright2021deleting}, vertex deletions~\cite{zschoche2020complexity}, or edge delays~\cite{deligkas2022optimizing} to restrict reachability, motivated, for instance, by epidemic containment strategies that limit interactions. Temporal network design is an active area of research, with further related questions explored for example in~\cite{akrida2017complexity,klobas2024complexity,christiann2024inefficiently}. \iflong For general introductory texts to the area of dynamic networks the reader is referred to \cite{casteigts2012time,michail2016introduction,michail2018elements}.\fi

\subparagraph{Organization of the paper.}

\nnew{
In~\Cref{sec:pre} we provide a formal problem definition and define the notions used in this work.
In~\Cref{sec:basic}, we show upper and lower bounds for the required number of labels per edge in any realization and provide a single exponential algorithm for all problem variants.
Afterwards, in~\Cref{sec:bridge}, we analyze properties and define splitting operations based on bridge edges in the solid graph of the undirected versions of our problem.
These structural insights will be mainly used in our FPT algorithm in~\Cref{sec:fes}, but also immediately let us describe a polynomial-time algorithm for instances where the solid graph is a tree in~\Cref{sec:tree}.
In~\Cref{sec:hardundir}, we then provide NP-hardness results for all undirected problem versions as well as parameterized intractability results for two of them with respect to feedback vertex set number and treedepth.
Afterwards, in~\Cref{sec:fes}, we provide our main algorithmic result: an FPT algorithm with respect to the feedback edge set number~$\fes$ of the solid graph.
This algorithm is achieved in three steps:
Firstly, we apply our splitting operations of~\Cref{sec:bridge} and provide a polynomial-time reduction rule to simplify the instance at hand, such that all we need to deal with is a subset~$X^*$ of vertices of size~$\Oh(\fes)$ for which the remainder of the graph decomposes into edge-disjoint trees that only interact with~$X^*$ via two leaves each.
We call such trees~\emph{connector trees}.
Secondly, we show that we can efficiently extend the set~$X^*$ to a set~$W^*$ of size~$\Oh(\fes)$ such that each respective connector tree is more or less independent from the remainder of the graph with respect to the interactions of temporal paths in any realization.
All these preprocessing steps run in polynomial time.
Afterwards, our algorithm enumerates all reasonable labelings on the edges incident with vertices of~$W^*$ and tries to extend each such labeling to a realization for~$D$.
As we show, there are only FPT many such reasonable labelings and for each such labeling, there are only FPT many possible extensions that need to be checked.
Finally, in~\Cref{sec:harddir} we briefly discuss the directed version of the problem and provide NP-hardness results for all but the two trivial cases.
}

\section{Preliminaries}\label{sec:pre}
For definitions on parameterized complexity, the Exponential Time Hypothesis (ETH), or parameters like treedepth, we refer to the textbooks~\cite{DF13,C+15}.

For natural numbers $i,j$ with $i\le j$ we write $[i]$ for the set $\{1,2,3,\ldots,i\}$
and $[i,j]$ for the set $\{i,i+1,i+2,\ldots,j\}$.
For an undirected graph $G=(V,E)$, we denote an edge between vertices $u$ and $v$
as $\{u,v\}$ or $uv$.  
By $N(v)=\{u\in V\mid vu \in E\}$ we denote the set of neighbors of~$u$.
The \emph{degree} of~$v$ is the size of~$N(v)$.
An edge $e\in E$ is a \emph{bridge} or \emph{bridge edge} in a graph
$G=(V,E)$ if deleting $e$ increases the number of connected components.
\tnew{A vertex of degree~$1$ is called a \emph{pendant} vertex, and
an edge incident with a pendant vertex is called a \emph{pendant} edge.}

For directed graphs, we denote an arc from
$u$ to $v$ by $(u,v)$.
A directed graph~$D$ is \emph{simple} if it has no parallel arcs
and no self-loops, and it is a \emph{directed acyclic graph} (DAG)
if it does not contain a directed cycle.
\iflong
\nnew{The \emph{in-degree} of a vertex~$v$ in~$D$ is the number of incoming arcs of~$v$, i.e., the number of arcs with head~$v$.
Similarly, the \emph{out-degree} of a vertex~$v$ is the number of arcs
with tail~$v$.
The~\emph{degree} of a vertex in~$D$ is then the sum of its in-degree and out-degree.}
\fi

A temporal graph $\mathcal{G}$ with vertex set $V$ and lifetime $L$
is given by a sequence $(G_t)_{t\in [L]}$ of $L$ static graphs
$G_t=(V,E_t)$ referred to as \emph{snapshots} or \emph{layers}.
The graph $\mathcal{G}_\downarrow=(V,E_\downarrow)$ with $E_\downarrow=\bigcup_{i\in [L]} E_i$ is called the
\emph{underlying graph} of~$\mathcal{G}$.
Alternatively, $\mathcal{G}$ can be represented
by an undirected graph $G=(V,E)$ with $E\supseteq \bigcup_{t\in[L]} E_t$ 
and a labeling function $\lambda:E\to 2^{[1,L]}$
that assigns to each edge $e$ the (possibly empty) \tnew{set of} time steps during which $e$ is
present, i.e., $\lambda(e)=\{t\in[L] \mid e\in E_t\}$.
We write $\mathcal{G}=(G,\lambda)$ in this case.
In this representation we allow $G$ to contain extra edges
in addition to the edges of the underlying graph; such edges $e$ satisfy
$\lambda(e)=\emptyset$. This is useful because in the problems
we consider in this paper, there is a natural choice of a
graph $G=(V,E)$, the so-called solid graph defined below,
that contains all edges of the underlying
graph of every realization but may contain additional edges.

We assume that each layer of a temporal graph is an undirected
graph unless we explicitly refer to directed temporal graphs.
If $t\in \lambda(e)$, we refer to $(e,t)$ as a \emph{time edge}.
A \emph{temporal path} from $u$ to $v$ in $\mathcal{G}$ is a sequence
of time edges $((e_j,t_j))_{j\in [\ell]}$ for some $\ell$
such that $(e_1,e_2,\ldots,e_\ell)$ is a $u$-$v$ path in $G$
and $t_1\le t_2\le \cdots\le t_\ell$ holds in the non-strict case and
$t_1<t_2<\cdots<t_\ell$ holds in the strict case.

A temporal graph is \emph{proper} if no two adjacent edges
share a label, \emph{simple} if every edge has a single label,
and \emph{happy} if it is both proper and simple~\cite{DBLP:journals/tcs/CasteigtsCS24}.
Note that for proper and happy temporal graphs
there is no distinction between strict and non-strict
temporal paths.

The reachability graph $\mathcal{R}(\mathcal{G})$ of a temporal graph $\mathcal{G}$ with vertex set~$V$
is the directed graph $(V,A)$ with the same vertex set $V$ and $(u,v)\in A$ if and only if
$u\neq v$ and $\mathcal{G}$ contains a temporal path from $u$ to~$v$.
Note that $\mathcal{R}(\mathcal{G})$ depends on whether we consider
strict or non-strict temporal paths and can be computed in polynomial time in both cases~\cite{XFJ03}.

We are interested in the following problem:
\prob{Reachability Graph Realizability (RGR)}{A simple directed graph $D=(V,A)$.}{Does there exist a temporal graph $\mathcal{G}$
with $\mathcal{R}(\mathcal{G})=D$?}

For yes-instances of \RGR, we are also interested in computing a
temporal graph $\mathcal{G}$ with $\mathcal{R}(\mathcal{G})=D$.
We refer to such a temporal graph as a \emph{solution} or
a \emph{realization} for~$D$, and we typically represent it by a
labeling function.
With the adjacency matrix representation of~$D$
in mind, we also write $D_{uv}=1$ for $(u,v)\in A$ and
$D_{uv}=0$ for $(u,v)\notin A$.

We can consider \RGR with respect to reachability via strict temporal
paths or with respect to non-strict temporal paths.
Furthermore, we can require the realization of~$D$ to be simple, proper, or happy.
For proper and happy temporal graphs, strict and non-strict reachability
coincide. Therefore, the distinct problem variants that we can consider
are \any\str\RGR, \any\nstr\RGR, \simp\str\RGR, \simp\nstr\RGR, \pro\RGR, and \happy\RGR. Finally, we write \DRGD instead of \RGR
if we are asking for a directed temporal graph that realizes~$D$.
We sometimes write \URGD if we want to make it explicit that we are asking
for an undirected realization.

\subparagraph{The solid graph.}
If $D_{uv}=1$ and $D_{vu}=1$ for some $u\neq v$, we say that
there is a \emph{solid edge} between $u$ and $v$.
If $D_{uv}=1$ and $D_{vu}=0$, we say that there is a
\emph{dashed arc} from $u$ to~$v$.
We use $G=(V,E)$ to denote the graph on $V$ whose
edge set is the set of solid edges, and we refer to
this graph as the \emph{solid graph} (of~$D$).
It is clear that only
solid edges can receive labels in a realization of~$D$,
as the two endpoints of an edge that is present in at least
one time step can reach each other.
Solid edges that are bridges of the solid graph must receive
labels, but a solid edge~$e$ that is not a bridge
need not necessarily receive labels, as the endpoints of~$e$ could reach each other
via temporal paths of length greater than one in a realization.
\begin{observation}\label{incident same label}
Let~$D$ be an instance of~\URGD.
Let~$\lambda$ be a realization of~$D$ and let~$\{u,v\}$ and~$\{v,w\}$ be solid edges that both receive at least one label under~$\lambda$.
If~$D$ contains neither the arc~$(u,w)$ nor the arc~$(w,u)$, then there is a label~$\alpha$ such that~$\lambda(\{u,v\}) = \lambda(\{v,w\}) = \{\alpha\}$.
\end{observation}
Note that \tnew{Observation~\ref{incident same label}} implies that for~\pro\URGD, \happy\URGD, and~\nstr\URGD, no realization for~$D$ can assign labels to \tnew{both of two adjacent
edges~$\{u,v\}$ and~$\{v,w\}$ if $D$ contains neither $(u,w)$ nor~$(w,u)$}.

A path $P=(u_0=u,u_1,u_2,\ldots,u_\ell=v)$ from $u$ to $v$
in $G$ is a \emph{dense} $u$-$v$-path if
there exist dashed arcs $(u_i,u_j)$ for
all $0\le i < j \le \ell$.
Note that in each undirected realization of~$D$, each temporal path is a dense path in~$G$.
Similarly, for directed realizations of~$D$, each temporal path is a dense path in~$D$, where we define dense paths in~$D$ analogously to dense paths in~$G$.

We say that a realization $\lambda$ is \emph{\minlab} if there is
no edge $e$ such that the set of labels assigned to $e$ can be
replaced by a smaller set while maintaining the property that
$\lambda$ is a realization.
We say that a realization $\lambda$
is \emph{minimal} on edge $e$ if it is impossible to obtain
another realization by replacing $\lambda(e)$ with a proper
subset. A realization is minimal if it is minimal on every
edge~$e$. Note that every frugal realization is also minimal.

\section{Basic Observations and an Exponential Algorithm}\label{sec:basic}
\nnew{
We first provide upper and lower bounds for the number of labels per edge in a minimal realization.
}

\iflong
\begin{lemma}
\else
\begin{lemma}
\fi
\label{lem:tree-atmost2}
    Let $D$ be an instance of \URGD, and let $e$ be a solid edge
    of $G$ that is not part of a triangle. In each minimal realization
    of $D$, $e$~receives at most two labels.
    Moreover, if~$D$ is an instance of any version of~\URGD besides \any\str\URGD, then in each minimal realization, $e$ receives at most one label.
    \end{lemma}
    \iflong
\begin{proof}
    First, consider \any\str\URGD.
    Consider a minimal realization and assume the edge $e$ has at least three labels. 
    Let $L_e$, $U_e$, and $M_e$ be
    its smallest label, its largest label, and an arbitrary label in between.
    Note that no solid edge $f\neq e$ incident with an endpoint of $e$ can have a label
    in $(L_e,U_e)$ as this would imply a solid edge forming a triangle containing~$e$.
    Hence, all other edges incident with endpoints of $e$ only have
    labels $\le L_e$ or $\ge U_e$. Replacing the labels of $e$ by
    the single label $M_e$ maintains exactly the temporal paths
    containing $e$ that exist before the change, a contradiction
    to the original realization being minimal.
    Furthermore, if the original
    labeling was proper, than the resulting labeling is also proper.

    For \simp\str\URGD, \simp\nstr\URGD and \happy\URGD, it is trivially true
    that $e$ has at most one label in any realization. Now,
    consider \any\nstr\URGD and \pro\URGD. Assume
    that $e$ has at least two labels in a minimal realization,
    and let $L_e$ and $U_e$ be the smallest and largest label
    of~$e$.
    Note that no solid edge $f\neq e$ incident with an endpoint
    of $e$ can have a label in $[L_e,U_e]$: A label in $(L_e,U_e)$
    would create a solid edge forming a triangle, $f$ cannot
    share a label with $e$ in a proper labeling, and $f$ sharing
    a label with $e$ would create a solid edge forming a triangle
    in the non-strict case. Thus, all other edges incident with
    endpoints of $e$ only have labels $<L_e$ or $>U_e$.
    Replacing the labels of $e$ by $L_e$ maintains exactly the
    temporal paths containing $e$ that exist before the change,
    a contradiction to the original realization being minimal.
    \end{proof}
\fi

\nnew{
For general edges, we show a linear upper bound with respect to the number of vertices.}

\iflong
\begin{lemma}
\else
\begin{lemma}
\fi\label{max number label}
If a graph~$D=(V,A)$ is realizable, then each minimal realization for~$D$ assigns at  most~$n=|V|$ labels per edge.
\end{lemma}
\iflong
\begin{proof}
Let~$e$ be an edge with more than~$n$ labels. 
We show that at least one can be removed.
For each vertex~$v$ define~$\alpha_v$ as the earliest time step in which vertex~$v$ can traverse the edge~$e$ (if any).
Let~$\beta$ be a label of~$e$ that is unequal to all~$\alpha_v$ labels.
Then, removing~$\beta$ from~$e$ preserves at least one temporal~$(x,y)$-path, if such a path existed before.
\end{proof}
\fi

\nnew{
Note that this implies that all versions of \URGD and \DRGD under consideration are in NP.}
Next, we show that this bound on the number of labels per edge is essentially tight for~\any\str\URGD.
\iflong
\begin{theorem}
\else
\begin{theorem}
\fi
\tnew{For \any\str\URGD,}
there is an infinite family of directed graphs $\mathcal{B}$, such that for every $D=(V,A)\in\mathcal{B}$, where $G=(V,E)$ is the solid graph of $D$, all of the following properties hold:
\begin{itemize}
\item $G$ is planar.
\item $G$ has a feedback vertex set of size 2 and a feedback edge set of size $\Theta(n)$, where $n=|V|$.
\item $D$ is realizable and there is an edge $e\in E$ such that every realization of $D$ uses $\Omega(n)$ labels on $e$.
\end{itemize}
\end{theorem}
\iflong
\begin{proof}
The graphs of solid edges of the reachability graphs in $\mathcal{B}$ resemble books with vertices around the sides of their pages. The spine of the book consists of an edge $u_0u'_0$ joining an upper spine-vertex $u_0$ and a lower spine-vertex $u'_0$. Every page has an index $1\leq i\leq L$, where $L\geq 5$ is the total number of pages, and consists of two upper vertices, $u_i$ and $v_i$, and two lower vertices, $u'_i$ and $v'_i$, located at the middle and end of the upper and lower boundaries, respectively. We call $u_i$ and $u'_i$ the inner and $v_i$ and $v'_i$ the outer vertices of page $i$. The set of edges of page $i$ consists of edges $u_0u_i$, $u_iv_i$, $v_iv'_i$, $u'_iv'_i$, and $u'_0u'_i$, around the sides of the page, as well as two additional diagonal edges, $u_0u'_i$ and $u'_0u_i$, between the spine and the inner vertices.
We then modify the constructed graph by making pages $1$, $2$, $L-1$ and $L$ incomplete: For $i=1,2$,
page~$i$ does not contain the lower-side vertices $u_i'$ and $v_i'$, and the
diagonal between $u_i$ and $u_0'$ is a dashed arc $(u_i,u_0')$ instead
of a solid edge. For $i=L-1,L$, page~$i$ does not contain the
upper-side vertices $u_i$ and $v_i$, and the diagonal between
$u_i'$ and $u_0$ is a dashed arc $(u_0,u_i')$ instead of a solid edge.
Furthermore, there is a dashed arc from every inner vertex of a page $i$ to every same-side (i.e., upper to upper and lower to lower) inner vertex of a page $j$ if and only if $i<j$ and both vertices exist, and there is a dashed arc to every opposite-side (i.e., upper to lower and lower to upper) inner vertex of a page $k$ if and only if $i+1<k$ and the target vertex exists. All remaining ordered pairs of vertices $(u,v)$ are specified as unreachable by $D$ (i.e., $(u,v)$ is not a dashed arc). See Figure \ref{fig:lower-bound-on-labels}(a). 
Note that $n=4(L-4)+2\cdot 4+ 2 =4L-6$ and the solid graph
has $m=7(L-4)+2\cdot 4+1 = 7L-19$ edges.

\begin{figure}[!hbtp]
\centering
\begin{subfigure}{.45\textwidth}
\centering
\includegraphics[width=0.9\linewidth]{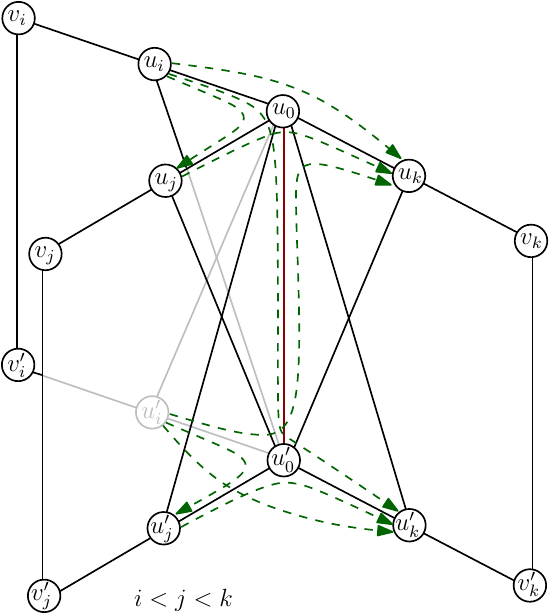}
\caption{}
\end{subfigure}
\hspace{1cm}
\begin{subfigure}{.45\textwidth}
\centering
\includegraphics[width=0.7\linewidth]{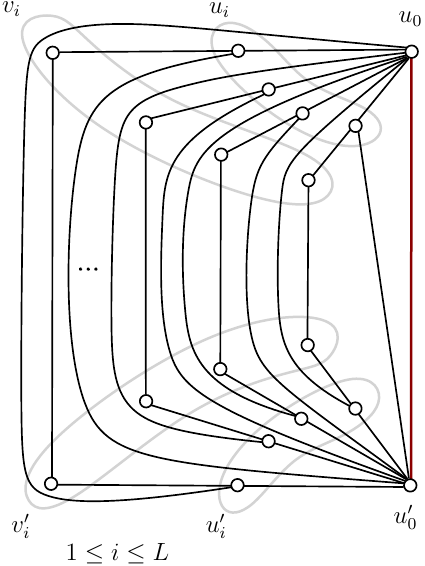}
\caption{}
\end{subfigure}
\caption{(a) An example of the graph of solid edges and dashed arcs for a $D\in\mathcal{B}$ restricted to three pages. (b) A planar drawing of the solid edges.}
\label{fig:lower-bound-on-labels}
\end{figure}

We start by proving the last property, fixing the edge $u_0u'_0$ of the spine as the edge required by it. We will show that $u_0u'_0$ must have at least $L-2=n/4-1/2=\Theta(n)$ labels. We claim that no two edges of a page $i$ can have distinct labels, so edges of page $i$ can be labeled by at most one label, which we denote by $l_i$. 
The two upper edges cannot have distinct labels, otherwise a missing arc between $u_0$ and $v_i$ would be realized. The same holds for the two lower edges. Similarly, $v_iv'_i$ and the diagonal edges cannot have a label different from that of the upper and lower edges, otherwise a missing arc would be realized. This implies that the only way to realize the arcs of $v_iv'_i$ is to label it by $l_i$. Edges $u_iv_i$ and $u'_iv'_i$ must also be labeled by $l_i$ as any other path to realize their arcs would need to use two consecutive edges of page $i$.
For $i=1,2$, both upper edges of page $i$ must be labeled
as they are bridges. For $i=L-1,L$, both lower edges of page $i$
must be labeled for the same reason. For
$3\le i\le L-2$, at least one of the edges $u_i u_0$ and
$u_i u_0'$ must be labeled, and at least one of the edges
$u_i'u_0'$ and $u_i'u_0$ must be labeled, as the arcs
between $u_i$ or $u_i'$ and $u_0$ and $u_0'$ cannot be
realized otherwise. We will show that $u_0 u_i$ and
$u_0' u_i'$ receive a label (which must be equal to $l_i$)
while the diagonals $u_0 u_i'$ and $u_0' u_i$ do not receive
a label. First, observe that $l_j>l_i$ must hold for any
$j>i$, otherwise the arc $(u_i,u_j)$ or $(u_i',u_j')$ cannot
be realized.
Now, assume for a contradiction that there is a page~$i$
with $3\le i \le L-2$ in which a diagonal receives a label.
Consider the smallest such~$i$.
Observe that $u_0 u_i'$ cannot
receive label $l_i$ as this would realize an
arc $(u_{i-1},u_i')$. Thus, $u_i' u_0'$ must
receive label~$l_i$. Furthermore, $u_i u_0$
must also receive label~$l_i$ as the arc
$(u_{i-1},u_i)$ must be realized via the
path $u_{i-1}, u_0, u_i$; it cannot be
realized via a temporal path ending with the edge $u_0' u_i$
as the existence of such a path would imply a temporal
path from $u_{i-1}$ to $u_i'$, a contradiction.
Finally, edge $u_i u_0'$ cannot receive a label;
otherwise, $u_i$ can reach both $u_0$ and $u_0'$
at time $l_i$ and therefore reach $u_{i+1}'$
at time $l_{i+1}$, a contradiction. This shows
that it is impossible that a diagonal edge in page
$i$ receives a label.
For every two distinct pages $i$ and $j$ such that $i<j$, the arcs from the inner vertices of $i$ to the same-side inner vertices of $j$ imply that $u_0u_i$, $u'_0u'_i$ and $u_0u_j$, $u'_0u'_j$ must be labeled by $l_i$ and $l_j$, respectively, and that $l_i<l_j$ as observed above.
Thus, the pages must have increasing labels according to the order of their indices.
For every page $1\leq i\leq L-2$, in order to realize the arc $(u_i,u'_{i+2})$ and/or $(u'_i,u_{i+2})$ (for every $i$ at least one of these two arcs exists),
the spine edge $u_0u'_0$ 
must have a label $l$ such that $l_{i}<l<l_{i+2}$.
Furthermore, that label must be equal to $l_{i+1}$, as
$l_{i}<l<l_{i+1}$ would create a temporal
path from an inner vertex of page $i$ to an opposite-side inner vertex
of page $i+1$, and $l_{i+1}<l<l_{i+2}$ can be excluded analogously.
Thus, the spine edge must receive labels $l_2,l_3,\ldots,l_{L-1}$,
a total of $L-2$ distinct labels.

The remaining part of the property follows by observing that the labeling assigning label $i$ to all non-diagonal edges of page $i$, for all $1\leq i\leq L$, and labels $2,3,4,\ldots,L-2,L-1$ to the spine, realizes $D$. For all $1<i+1<j\leq L$, all arcs from the inner vertices of a side of page $i$ to the opposite-side inner vertices of page $j$ are realized by using label $i+1$ of the spine. These temporal paths also realize the arcs $(u_i,u_0')$
for all $i=1,2,\ldots L-2$, $(u_0',u_i)$ for $i=3,\ldots, L-2$,
$(u_0,u_i')$ for $i=3,\ldots,L$,
and the arcs $(u_i',u_0)$ for $i=3\ldots,L-2$. Thus, all solid
diagonals in pages $3,\ldots,L-2$ and the diagonal arcs in the
remaining levels are realized.
For same-side inner vertices, the temporal paths realizing them go directly through the respective spine-vertex. As required, no dashed arc between the outer vertices and the rest of the graph is realized as all these paths go through inner vertices, and no two edges incident to an inner vertex use different labels. For all \tnew{$1\leq i< j\leq L$}, no dashed arc from the inner vertices of page $j$ to those of page $i$ is realized because \tnew{$j>i$}. 

For the first property we draw the book pages as nested trapezoids with their area increasing by increasing index, and having the spine as their shared basis (see Figure \ref{fig:lower-bound-on-labels}(b)). This gives a planar drawing of all edges but the diagonals. For all $1\leq i\leq L-1$, the diagonals $u'_iu_0$ and $u_{i+1}u_0'$ can be drawn on the face between pages $i$ and $i+1$ without crossing each other. Diagonals $u_1u_0'$ and $u_L'u_0$ can be drawn on the face formed by page 1 and on the outer face, respectively. This completes a planar drawing of a supergraph of $G$ in which all $L$ pages are complete;
the graph $G$ (in which some vertices and edges have been removed from
pages $1,2,L-1,L$) thus also admits a planar drawing.

The spine vertices form a feedback vertex set of size 2. The 
union of the diagonal edges and the edges joining the outer vertices of pages $3$ to $L-2$ forms a feedback edge set of size $3(L-4)= 3n/4-15/2$ and every feedback edge set has size at least $(7L-19)-(n-1)=3n/4-15/2=\Theta(n)$. Thus, the second property holds.
\end{proof}
\fi

We now present single exponential algorithms for all directed and undirected version of~\RGD under consideration based on our upper bounds on the number of assigned labels per edge/arc.

\iflong
\begin{theorem}
\else
\begin{theorem}
\fi\label{exact algos}
Each version of \URGD and \DRGD under consideration can be solved in \nnew{$2^{\Oh(|A|)}\cdot n^{\Oh(1)}$~time, where~$A$ denotes the arc set of the input graph.}
\end{theorem}
\begin{proof}
We describe a dynamic program.
We state the program only explicitly for \any\str\URGD and afterwards argue how the dynamic program can be adapted to the other \tnew{problem variants considered}.

Let~$D=(V,A)$ be an instance of~\any\str\URGD and let~$n:= |V|$.
Moreover, let~$E$ denote the set of solid edges of~$D$ and let~$m:= |E|$.
Due to~\Cref{max number label}, there is a realization for~$D$ if and only if there is a realization for~$D$ with lifetime at most~$m\cdot n$.

For each subset~$A'\subseteq A$ and each~$i\in [0,m\cdot n]$, our dynamic programming table~$T$ has an entry~$T[A',i]$.
This entry stores whether there is a temporal graph~$\mg$ with lifetime at most~$i$, such that the strict reachability graph of~$\mg$ is~$D':=(V,A')$.

For~$i = 0$, we set~$T[\emptyset,0] = 1$ and~$T[A',0] = 0$ for each nonempty~$A'\subseteq A$.
For each~$i \in [1,m\cdot n]$ and each~$A' \subseteq A$, we set 

\begin{align*}
T[A', i] =~& \exists A'' \subseteq A'.  \exists S \subseteq E.  T[A' \setminus A'',i-1]  \\
&\land(\forall (u,v)\in A''. \exists w. (w = u \lor (u,w)\in A'\setminus A'') \land w \text{ reaches } v \text{ via an edge of }S)\\
&\land(\forall (u,v)\notin A'. \not\exists w. (w = u \lor (u,w)\in A'\setminus A'') \land w \text{ reaches } v \text{ via an edge of }S)
\end{align*}

Intuitively, the recurrence resembles splitting the reachability of arcs~$A'$ into two sets~$A''$ and~$A'\setminus A''$.
We recursively check for the existence of a temporal graph~$\mg''$ of lifetime at most~$i-1$ with strict reachability graph equal to~$(V,A'\setminus A'')$, and if such a temporal graph exists, we append a new snapshot to the end with edges~$S$.
The second and the third line of the recurrence then ensure that (i)~all arcs of~$A'\setminus A''$ are realized by introducing this new snapshot~$i$, and that (ii)~no arcs outside of~$A'$ are realized by introducing this new snapshot~$i$.

We omit the formal correctness proof and proceed with the evaluation of the dynamic program.
The graph~$D$ is realizable if and only if $T[A,n\cdot m] = 1$.
Since there are~$2^{|A|} \cdot n \cdot m$ table entries and each such entry can be computed in $2^{|A|} \cdot 2^{|E|} \cdot n^{\Oh(1)} \subseteq 2^{\Oh(|A|)} \cdot n^{\Oh(1)}$~time, the whole algorithm runs in~$2^{\Oh(|A|)} \cdot n^{\Oh(1)}$~time.
This completes the proof for~\any\str\URGD.
Moreover, note that a corresponding realization (if one exists) can be found in the same time via traceback.

\iflong
If we consider proper labelings, we can modify the recurrence to only allowed to take subsets of edges~$S\subseteq E$ that are matchings.

If we consider non-strict paths, we can modify the recurrence to ask whether~$w$ reaches~$v$ via a path using only edges of~$S$ instead of only a single edge of~$S$.

If we consider simple labelings, we can modify the dynamic program by a third dimension~$E' \subseteq E$ which stores the edges of~$E$ we may possibly assign labels to.
In the recurrence, we then have to ensure that~$S$ is only a subset of~$E'$ and also remove~$S$ from the allowed edges in the recursive call, that is, we have to replace~$T[A'\setminus A'', i-1]$ by~$T[A'\setminus A'',i-1, E' \setminus S]$.

For the versions of~\DRGD \tnew{under consideration}, we do essentially the same, except that~$S$ has to be a subset of~$A$.

If we consider proper labelings, instead of picking a matching for~$S$, we must pick a set of arcs that do not \tnew{induce} a path of length more than one.

If we consider simple labelings, we do essentially the same as in the undirected case.
We extend the table by a third dimension that stores the arcs~$A^*$ of~$A$ we may possibly assign labels to.
In the recurrence, we then have to ensure that~$S$ is only a subset of~$A^*$.
\fi
\end{proof}

\section{Solid Bridge Edges: Properties and Splitting Operations}\label{sec:bridge}
In this section, we show several structural results as well as three splitting operations for graphs with bridge edges in the solid graph.
These insights will be important for our FPT algorithm and will also simplify instances of~\URGD where the solid graph is a tree.

Consider a bridge edge $e=\{u,v\}$ whose deletion
splits the solid graph $G$ into connected components $G_u$ and $G_v$, where
$G_u$ contains $u$ and $G_v$ contains $v$.
A dashed arc $(a,b)$ \emph{spans} $e$ if $a\in V(G_u)$ and $b\in V(G_v)$
or vice versa. Two arcs $(a,b)$ and $(c,d)$ span $e$ in the same direction
if they span $e$ and either $a,c\in V(G_u)$ or $a,c\in V(G_v)$.
If $e$ has a single label in a realization, then it must be the case
that the dashed arcs spanning $e$ are ``transitive'' in the following sense:
If dashed arcs $(a,b)$ and $(c,d)$ both span $e$ in the same direction, then $(a,d)$
and $(c,b)$ must also be dashed arcs. This is because for any two temporal
paths that pass through $e$ in the same direction, the part of one path
up to edge $e$ can be combined with the part of the other path after $e$.

\begin{definition}
    \label{def:special}
Consider a bridge edge $e=\{u,v\}$ whose deletion
splits the solid graph $G$ into connected components $G_u$ and $G_v$.
The edge $e$ is a \emph{special} bridge
edge if there exist vertices $a$ in $G_u$ and $b$ in $G_v$
such that $D_{av}=1$, $D_{ub}=1$ and $D_{ab}=0$ (or if the
same condition holds with $u$ and $v$ exchanged). A bridge edge
that is not special is called \emph{non-special}.
\end{definition}

Intuitively,
a special bridge \tnew{edge} is a bridge edge $e$ such that there are dashed arcs
spanning $e$ in the same direction for which transitivity (as outlined
above) is violated.
We also refer to special bridge edges as special bridges, special solid edges, or just as special edges.

\iflong
\begin{lemma}
\else
\begin{lemma}
\fi
\label{lem:special}
    In each \minlab realization of $D$,
    every special bridge edge
    is assigned two     labels and every non-special bridge edge is
    assigned a single label.
\end{lemma}
\iflong
\begin{proof}
    Let $\lambda$ be a \minlab realization. Consider a bridge edge $e=\{u,v\}$.
    It is clear that $e$ must receive at least one label in $\lambda$, as otherwise
    the endpoints of $e$ cannot reach each other.
    As a bridge edge is not contained in a triangle,
    it receives at most two labels in $\lambda$ by Lemma~\ref{lem:tree-atmost2}.
    It is clear that every special edge must receive two labels
    as otherwise the violation of transitivity that happens
    on a special edge would be impossible.
    Assume that $e=\{u,v\}$ is a non-special edge that
    has received two labels $x$ and $y$ with $x<y$. We claim that replacing the double label
    of $e$ by the new label $(x+y)/2$ produces a new solution $\lambda'$ that also
    realizes~$D$, a contradiction to $\lambda$ being \minlab.
    (The possibly half-integral label $(x+y)/2$ can be made integral in the end by replacing
    all labels by integers while maintaining their order.)

    Consider strict temporal paths first. Note again that the edges $f\neq e$ incident
    with endpoints of $e$ can only have labels $\le x$ or $\ge y$.
    We consider the reachability from $G_u$ to $G_v$. (The arguments for the reachability from
    $G_v$ to $G_u$ are analogous.)
    Case 1: If every $a\in V(G_u)$ can reach $v$ in step $x$ in $(G,\lambda)$, then it can reach $v$
    in step $(x+y)/2$ in $(G,\lambda')$. The vertices in $V(G_v)$ that $u$ can reach starting
    in time step $x$ in $\lambda$ are precisely the vertices in $V(G_v)$ that can
    be reached starting in time step $(x+y)/2$ in $\lambda'$, and hence the reachability
    from \tnew{$G_u$ to $G_v$} is the same in $(G,\lambda')$ as in $(G,\lambda)$.
    Case 2: If there exists $a\in V(G_u)$ that can reach $v$ only in step $y$ in $(G,\lambda)$,
    then every vertex $b$ in $V(G_v)$ that can be reached from $u$ in $(G,\lambda)$ must be reachable
    starting from $u$ in time step $y$ in $(G,\lambda)$ (otherwise $a$ would not be able
    to reach $b$, making $e$ a special edge). All vertices from $V(G_u)$ that can
    reach $v$ in $(G,\lambda)$ can reach $v$ in $(G,\lambda')$ in time step $(x+y)/2$. Furthermore,
    the vertices in $V(G_v)$ that $u$ can reach starting
    in time step $x$ in $(G,\lambda)$ (which are the same as those that can be reached
    starting at $u$ in time step $y$ in $(G,\lambda)$) are precisely the vertices in $V(G_v)$ that can
    be reached starting at $u$ in time step $(x+y)/2$ in $(G,\lambda')$. Hence, the reachability
    from $V(G_u)$ to $V(G_v)$ is the same in $(G,\lambda')$ as in $(G,\lambda)$.

    Repeating such replacements then yields a valid solution in which all non-special
    edges have a single label.

    The same approach works for non-strict temporal graphs (where the labels
    of edges $f\neq e$ incident with endpoints of $e$ must be $<x$ or $>y$ in $(G,\lambda)$).
    Furthermore, if the original labeling is proper, then the final labeling
    is also proper.
\end{proof}
\fi

Note that this implies the following due to~\Cref{lem:tree-atmost2}.

\begin{corollary}
\label{cor:nospecial}
Let~$D$ be an instance of any version of \URGD under consideration besides~\any\str\URGD.
If the solid graph of~$D$ contains a special solid edge, then~$D$ is a no-instance.
\end{corollary}

\iflong
We are now interested in when a pair of bridge edges that share an endpoint
must, must not, or may share a label. For the non-strict problem
variants (and thus also the happy and proper cases for strict paths),
adjacent bridges obviously cannot share a label.
The following definition captures
the condition for when adjacent bridges must share a label in
the other cases.

\begin{definition}
\label{def:bundled}Consider \any\strict\URGD or \simp\strict\URGD.
Let~$D$ be an instance of~\URGD and let~$G$ be the solid graph of~$D$.
Let~$\{u,c\}$ and~$\{v,c\}$ be bridges in~$G$, and let~$V_u$ ($V_v$) denote the set of vertices in the same connected component as~$u$ ($v$) in~$G-\{\{u,c\}\}$ ($G-\{\{v,c\}\}$).
We say that~$\{u,c\}$ and~$\{v,c\}$ are~\emph{bundled}, if 
\begin{itemize}
\item $(u,v)\notin A$ and~$(v,u)\notin A$, or
\item at least one of the two edges is special with respect to~$D[V_u \cup V_v \cup \{c\}]$.
\end{itemize}
\end{definition}

Note that the second condition of Definition~\ref{def:bundled}
cannot occur in a yes-instance of \simp\strict\URGD.

\begin{lemma}\label{meaning of bundled}
In every realization, each pair of bundled bridges shares a label.
\end{lemma}

\begin{proof}
Consider bridges $\{u,c\}$ and~$\{v,c\}$ that are bundled.
If $(u,v)\notin A$ and~$(v,u)\notin A$, both edges must receive
the same single label.
Assume for the rest of the proof that one of the two arcs, say
the arc $(u,v)$ is in $A$ (and hence $(v,u)\notin A$) and
at least one of the two edges $\{u,c\}$ and $\{v,c\}$ is special.
This can only happen in the \any\strict\URGD model.
The labels on $\{c,v\}$ must be at least as large as the largest label on $\{c,u\}$,
as $v$ cannot reach~$u$.

As at least one of the two edges $\{u,c\}$ and $\{v,c\}$ is special
in~$G[V_u \cup V_v \cup \{c\}]$ by Definition~\ref{def:bundled},
there must exist
vertices $a\in V_u$ and $b\in V_v$ such that $D_{ac}=D_{ub}=1$ and $D_{ab}=0$
if $\{u,c\}$ is special or such that $D_{av}=D_{cb}=1$ and $D_{ab}=0$
if $\{v,c\}$ is special. (It cannot be that $\{u,c\}$ or $\{v,c\}$ is
special because there are vertices $x\in V_v$ and $y\in V_u$ with $D_{xc}=D_{vy}=1$ or $D_{xu}=D_{cy}=1$ and $D_{xy}=0$; this would
contradict the absence of a temporal path from $v$ to~$u$.)
If $\{u,c\}$ and $\{v,c\}$ did not share a label, any temporal path from
$a$ to $c$ could be concatenated with any temporal path from $c$ to $b$,
contradicting $D_{ab}=0$. Hence, the two edges must share a label.
\end{proof}

Next, we characterize pairs of adjacent bridges that must not share a label
in any realization.

\begin{definition}
Let~$D$ be an instance of~\URGD and let~$G$ be the solid graph of~$D$.
Let~$\{u,c\}$ and~$\{v,c\}$ be bridges in~$G$.
We say that~$\{u,c\}$ and~$\{v,c\}$ are~\emph{separated}, if 
\begin{itemize}
\item we consider \pro\URGD, \happy\URGD, or \nstr\URGD,
\item we consider \simp\str\URGD and~$A$ contains an arc of~$\{(u,v),(v,u)\}$,
\item one arc of~$\{(u,v),(v,u)\}$ is contained in~$A$ and neither of the bridges~$\{u,c\}$ and~$\{v,c\}$ is special,
\item there are~$w_1,w_2 \in N(c)$ such that $\{c,w_1\}$ and $\{c,w_2\}$ are also bridges and~$(u,w_1,w_2,v)$ or $(v,w_1,w_2,u)$ is a path in~$D$, or
\item there is~$w \in N(c)$ such that $\{w,c\}$ is also a bridge and~$(u,w,v)$ (or $(v,w,u)$) is a path in~$D$ such that (i)~$\{u,c\}$ is non-special, (ii)~$\{w,c\}$ is special, or (iii)~$\{v,c\}$ is non-special.
\end{itemize}
\end{definition}

\begin{lemma}\label{meaning of separated}
In every \minlab realization, each pair of separated bridges shares no label.
\end{lemma}

\begin{proof}
If we consider non-strict paths, adjacent bridges cannot share a label because a shared label would imply a triangle containing at least one of these bridges. 
If we consider proper or happy labeling, no two adjacent edges can share a label.

If we consider \simp\str\URGD and $A$ contains, say,
the arc $(u,v)$, then the label of $\{c,v\}$ must be larger than the
label of $\{c,u\}$ because each edge receives at most one label and
the two edges must form a temporal path from $u$ to $v$.

If neither of the bridges $\{u,c\}$ and~$\{v,c\}$ is special and
$A$ contains one of the arcs $(u,v)$ or $(v,u)$ (note that it cannot
contain both), the same argument as in the previous paragraph
shows that the single labels of the two bridges cannot be the same.

Now consider the case that there exist $w_1,w_2 \in N(c)$ such that $\{c,w_1\}$ and $\{c,w_2\}$ are also bridges and~$(u,w_1,w_2,v)$ is a path in~$D$. (The argument for the case that $(v,w_1,w_2,u)$ is a path in~$D$ is analogous.)
As the four edges from $c$ to $u,w_1,w_2,v$ are all bridges,
there cannot be any solid edge among the vertices $u,w_1,w_2,v$.
Thus, there are dashed arcs $(u,w_1)$, $(w_1,w_2)$ and $(w_2,v)$.
Let $x_1=u$, $x_2=w_1$, $x_3=w_2$, and $x_4=v$.
Let $\lambda$ be a \minlab realization.
For $i\in[4]$, let $\ell_{x_i}=\min \lambda(\{c,x_i\})$ and
$h_{x_i}=\max \lambda(\{c,x_i\})$.
We have $\ell_{x_i}< h_{x_{i+1}}$ for $i \in [3]$ as
$(x_i,x_{i+1})\in A$ and $h_{x_i}\le \ell_{x_{i+1}}$ as
$(x_{i+1},x_{i})\notin A$. Thus, we have
$h_{x_1}\le \ell_{x_2} < h_{x_3}\le \ell_{x_4}$,
showing that $h_{x_1}<\ell_{x_4}$ and hence
$\{c,x_1\}$ and $\{c,x_4\}$ cannot share a label.

Finally, consider the case that
there is~$w \in N(c)$ such that $\{w,c\}$ is also a bridge and~$(u,w,v)$ is a path in~$D$ such that (i)~$\{u,c\}$ is non-special, (ii)~$\{w,c\}$ is special, or (iii)~$\{v,c\}$ is non-special. (The argument for the
case that $(v,w,u)$ is a path in $D$ is analogous.)
Again, there cannot be any solid edges among the vertices $u,w,v$
and thus there are dashed arcs $(u,w)$ and $(w,v)$.
Let $x_1=u$, $x_2=w$ and $x_3=v$ and let $\lambda$ be a \minlab
realization, and define $\ell_{x_i}$ and $h_{x_i}$ for $i\in[3]$
as above. We have
$\ell_{x_1}\le h_{x_1}\le \ell_{x_2} \le h_{x_2} \le \ell_{x_3} \le h_{x_3}$
and $\ell_{x_i}<h_{x_{i+1}}$ for $i\in[2]$.
If $\{u,c\}=\{x_1,c\}$ is non-special, we have $h_{x_1}=\ell_{x_1}<h_{x_2}\le \ell_{x_3}$.
If $\{w,c\}=\{x_2,c\}$ is special, we have $h_{x_1}\le \ell_{x_2} < h_{x_2}\le \ell_{x_3}$.
If $\{v,c\}=\{x_3,c\}$ is non-special, we have $h_{x_3}=\ell_{x_3}>\ell_{x_2}\ge h_{x_1}$.
In all three cases, we have $h_{x_1}<\ell_{x_3}$, showing
that $\{c,x_1\}$ and $\{c,x_3\}$ cannot share a label.
\end{proof}

\nnew{

\Cref{meaning of bundled,meaning of separated}, thus imply the following.
\begin{corollary}\label{bundledseparated}
Let~$e$ and~$e'$ be adjacent bridges that are both bundled and separated in~$D$.
Then~$D$ is not realizable.
\end{corollary}
}
\fi

\subsection{Splitting Into Subinstances}
Next, we show that bridge edges in the solid graph can be used to split an instance
into smaller subinstances that can be solved independently and have the property
that the original instance is realizable if and only if the two subinstances
are realizable. First, consider non-special bridges.

\begin{lemma}[Splitting at a non-special bridge]
    \label{lem:split-nonspecial}
    Consider any variant of \URGD.
    Let $e=\{u,v\}$ be a non-special bridge edge, and let $G_u$ and $G_v$
    be the connected components of $G$ resulting from the deletion of~$e$.
    Then the instance $D$ is realizable if and only if the subinstances
    induced by $V(G_u)\cup\{v\}$ and $V(G_v)\cup\{u\}$ are realizable.
\end{lemma}

\begin{proof}
Let $D_u$ and $D_v$ denote the subinstances induced by
$V(G_u)\cup\{v\}$ and $V(G_v)\cup\{u\}$, respectively.
If $D$ is realizable, then the realization $\lambda$ induces realizations
of the two subinstances $D_u$ and $D_v$. If $D_u$ and $D_v$
are realizable, their realizations assign a single label to $e$,
and these realizations can be chosen so that
$e$ receives the same label in both realizations. Hence,
the union of the two realizations is a realization of $D$.
\end{proof}

By applying Lemma~\ref{lem:split-nonspecial} repeatedly to a non-special
bridge edge that is not pendant until no such edge exists, we
obtain subinstances in which all non-special bridge are pendant.

Now consider special bridges.
Special bridge edges cannot occur in yes-instances
of any variant of \URGD except \any\str\URGD (see \Cref{cor:nospecial}), so we only consider
\any\str\URGD in the following. Let edge $e=\{u,v\}$ be a special
bridge with $G_u$ and $G_v$ defined as above. 
If the instance is realizable, for every vertex
in $V(G_u)$ there are at most three possibilities for which vertices
in $V(G_v)$ it can reach, depending on whether it cannot reach $v$,
can reach $v$ at time $\alpha$, or can reach $v$ only at time $\beta$,
where $\alpha$ and $\beta$ with $\alpha<\beta$ are the labels assigned
to $e$ in the realization. The same holds with $G_u$ and $G_v$ exchanged.
If this condition is violated, then the instance cannot be
realized. The following definition captures this condition.
For any vertex $w\in V(G_u)$, let $R_{u\to v}(w)=\{x\in V(G_v)\mid (w,x)\in A\}$,
and define $R_{v\to u}(x)$ for $x\in V(G_v)$ analogously.

\begin{definition}
    Let $e=\{u,v\}$ be a special bridge edge with $G_u$ and $G_v$
defined as above. We say that $e$ is a special bridge edge with \emph{plausible reachability}
if the following conditions hold:
\begin{itemize}
\item  If there exist vertices $a\in V(G_u)$ and $b\in V(G_v)$
such that $D_{av}=D_{ub}=1$ and $D_{ab}=0$, then every vertex
$w\in V(G_u)$ satisfies $R_{u\to v}(w)\in\{\emptyset,R_{u\to v}(a),R_{u\to v}(u)\}$ with $\emptyset\subset R_{u\to v}(a)\subset R_{u\to v}(u)$.
Otherwise, every vertex satisfies $R_{u\to v}(w)\in\{\emptyset,R_{u\to v}(u)\}$.
\item The same condition holds with the roles of $G_u$ and $G_v$ exchanged.
\end{itemize}
\end{definition}

\begin{lemma}[Splitting at a special bridge]
    \label{lem:split-special}
    Consider \any\str\URGD.
    Let $e=\{u,v\}$ be a special bridge edge with plausible reachability,
    and let $G_u$ and $G_v$ be the connected components of $G$ resulting from the deletion of~$e$.
    Then the instance $D$ is realizable if and only if the two subinstances
    constructed as follows are realizable:
    \begin{itemize}
        \item To construct subinstance $D_u$, take the subgraph $D'$ of $D$
        induced by $V(G_u)\cup\{v\}$ and attach one or two leaves to $v$
        as follows:
        \begin{itemize}
            \item If there exist vertices $a\in V(G_u)$ and $b\in V(G_v)$
such that $D_{av}=D_{ub}=1$ and $D_{ab}=0$,
then attach a leaf (called \emph{out-leaf})
$z$ to $v$ with $D'_{zv}=1$ and $D'_{az}=1$ for all \tnew{$a\in\{v\}\cup\{w\in V(G_u)\mid 
R_{u\to v}(w)=R_{u\to v}(u)\}$}, and all other entries of $D'$ involving $z$ equal to~$0$.
\item If there exist vertices $c\in V(G_v)$ and $d\in V(G_u)$
such that $D_{cu}=D_{vd}=1$ and $D_{cd}=0$, then attach a leaf (called \emph{in-leaf})
$z'$ to $v$ with $D'_{vz'}=1$ and $D'_{z'a}=1$ for all \tnew{$a\in \{v\}\cup R_{v\to u}(c)$} and $D'_{z'z}=1$ if an out-leaf $z$ has been added, and
all other entries of $D'$ involving $z'$ equal to~$0$.
        \end{itemize}
        Note that at least one of the two conditions above must be satisfied
        because $e$ is a special bridge edge.
        The resulting instance $D'$ is the desired subinstance~$D_u$.
        \item Subinstance $D_v$ is constructed analogously, with the roles
        of $u$ and $v$ exchanged.
    \end{itemize}
\end{lemma}

\iflong
\begin{proof}
    Observe that $e$ is a special bridge edge also in both subinstances
    $D_u$ and $D_v$. Furthermore $D_u$ has an out-leaf if and only if
    $D_v$ has an in-leaf, and vice versa.
    If $D$ is realizable, then a realization of $D_u$ can be obtained as follows:
    Let $\lambda$ be a \minlab realization of $D$. Let $\lambda_u$ be the restriction
    of $\lambda$ to edges in the subgraph of $G$ induced by
    $V(G_u)\cup\{v\}$, and let $\alpha,\beta$ with $\alpha<\beta$
    be the two labels assigned to $e$ by $\lambda_u$.
    If $D_u$ has an in-leaf $z$, set $\lambda_u(vz)=\alpha$.
    If $D_u$ has an out-leaf $z'$, set $\lambda(vz')=\beta$.
    We claim that $\lambda_u$ is a realization of $D_u$:
    It definitely realizes the reachability between all
    vertices in $V(G_u)\cup\{v\}$, so we only need to consider
    the in- and/or out-leaf attached to $v$. If $v$ has an
    out-leaf~$z$, then the \tnew{vertices} in $V(G_u)$ that can reach
    $z$ are exactly those that can reach $v$ at time $\alpha$, as desired.
    If $v$ has an in-leaf $z'$, then the vertices in $V(G_u)$
    that $z'$ can reach are exactly the vertices that can be
    reached from $v$ if traversing the edge $vu$ at time $\beta$,
    as desired. Hence, $\lambda_u$ is a realization of $D_u$.
    A realization $\lambda_v$ of $D_v$ can be obtained analogously.

    Now assume that $D_u$ and $D_v$ have realizations $\lambda_u$
    and $\lambda_v$, respectively. By shifting labels and inserting
    empty time steps if necessary, we can assume that both labelings
    assign the same two labels $\alpha,\beta$ with $\alpha<\beta$ to
    $uv$ (by possibly introducing empty time steps in both realizations). Consider $D_u$ and $\lambda_u$.
        If $D_u$ has an in-leaf $z$, the label of the non-special bridge
    edge $vz$ must equal $\alpha$, because $z$ must reach some vertices
    in $V(G_u)$ but not all vertices that $v$ can reach.
    If $D_u$ has an out-leaf $z'$, the label of the non-special bridge
    edge $vz'$ must equal $\beta$ as not all vertices from $V(G_u)$
    that can reach $v$ can reach~$z'$.
    
    We make the following modification to $\lambda_u$:
    If $D_u$ has no out-leaf and at least one edge $f\neq e$ incident with $u$
    receives label $\alpha$, we reduce all labels $\le \alpha$ on edges
    of $G_u$ by one. This ensures that all vertices in $V(G_u)$ that
    can reach $v$ can reach $v$ at time $\alpha$, without changing
    any reachabilities. If $D_u$ has no in-leaf and at least one
    edge $f\neq e$ incident with $u$ that has label $\beta$,
    we increase all labels $\ge\beta$ on edges of $G_u$ by one.
    This ensures that all vertices in $V(G_u)$ that can be
    reached from $v$ can still be reached from $v$ if traversing
    edge $vu$ at time $\beta$. Intuitively, these modifications
    ensure that temporal paths leaving $G_u$ can reach $v$
    at time $\alpha$ (if there is no out-leaf) and temporal
    paths entering $G_u$ can \tnew{reach
    all vertices in $G_u$ that $v$ can reach even if these paths
    reach $u$ only at time $\beta$} (if there is no in-leaf).
    We apply the same modification to $\lambda_v$.
    
    Now let $\lambda$ be
    the union of $\lambda_u$ and $\lambda_v$. We claim that $\lambda$
    is a realization of~$D$. Reachability requirements within $V(G_u)\cup\{v\}$
    and $V(G_v)\cup \{u\}$ are clearly satisfied because $\lambda_u$ and
    $\lambda_v$ realize $D_u$ and $D_v$, respectively.
    Consider an arbitrary $a\in V(G_u)$. Case 1: $D_u$ has no out-leaf.
    If $R_{u\to v}(a)=\emptyset$, then $a$ cannot reach $v$ in $\lambda$
    as it cannot reach $v$ in $\lambda_u$. If $R_{u\to v}(a)=R_{u\to v}(u)$,
    then $a$ can reach $v$ at time $\alpha$ in $\lambda_u$ and can hence
    reach all vertices in $R_{u \to v}(u)$ in $\lambda$.
    Case 2: $D_u$ has an out-leaf~$z$.
    If $R_{u\to v}(a)=\emptyset$, then $a$ cannot reach $v$ in $\lambda$
    as it cannot reach $v$ in $\lambda_u$. If $R_{u\to v}(a)=R_{u\to v}(u)$,
    then $a$ can reach $v$ at time $\alpha$ in $\lambda_u$ as otherwise
    it wouldn't be able to reach $z$. Hence, it can reach all vertices
    in $R_{u\to v}(u)$ in $\lambda$.
    If $\emptyset\neq R_{u\to v}(a)\subset R_{u\to v}(u)$,
    then $a$ can reach $v$ only at time $\beta$ in $\lambda_u$ as otherwise
    it would be able to reach $z$, a contradiction to the definition of $D_u$.
    Hence, in $\lambda$ it can reach all vertices in $V(G_v)$ that can be reached
    after traversing the edge $uv$ from $u$ to $v$ at time $\beta$.
    By definition of $D_v$, these are the vertices in $V(G_v)$ that
    are reachable from the in-leaf $z''$ of $D_v$, and by definition of $D_v$ that set of
    vertices is $R_{u\to v}(c)$, where
    $c$ is an arbitrary vertex in $V(G_u)$ with $\emptyset\neq R_{u\to v}(c)\subset R_{u\to v}(u)$. As the special edge has plausible reachability,
    $R_{u\to v}(c)=R_{u\to v}(a)$. Therefore, reachability from $G_u$ to
    $G_v$ is realized correctly by~$\lambda$. The arguments for
    reachability from $G_v$ to $G_u$ are analogous.
\end{proof}
\fi

Finally, we consider the case of two pendant edges that
must receive the same single label in any realization
and argue that we can remove one of them from the instance
provided their reachability requirements from/to the rest
of the graph are the same.

\begin{lemma}[Removal of redundant pendant vertices]
\label{lem:remove-pendant}
    Let $v$ and $w$ be two vertices of degree $1$ that are both
    adjacent to the same vertex $u$ in $G$ and satisfy
    $(v,w)\notin A$ and $(w,v)\notin A$.
    For all variants of \nstr\URGD (including \pro\URGD and \happy\URGD),
    the instance is a no-instance.
    For \any\str\URGD and \any\simp\URGD, we have that
    if $D_{vx}=D_{wx}$ and $D_{xv}=D_{xw}$ hold
    for all $x\in V\setminus\{v,w\}$, then $D$ is realizable if and only
    if the subinstance $D'$ resulting from $D$ by deleting $v$ is realizable;
    otherwise, $D$ is a no-instance.
\end{lemma}
\iflong
\begin{proof}
For all variants of \nstr\URGD, it is impossible
to have adjacent bridge edges without a temporal
path between their endpoints different from the
common endpoint. Therefore, consider \any\str\URGD and \any\simp\URGD
for the remainder of the proof.
We know that $uv$ and $uw$ must receive
the same single label in any realization and thus
must reach the same set of vertices in $V\setminus\{v,w\}$
and be reached from the same set of vertices in $V\setminus\{v,w\}$.
If they have different reachability, it is clear that $D$ is
a no-instance. If \tnew{their} reachability is the same,
the claim can be shown as follows:
A solution to $D$ clearly induces a solution to $D'$.
A solution to $D'$ can be extended into a solution to $D$ by
assigning $uv$ the label that $uw$ has received in the solution
to~$D'$.
\end{proof}
\fi

\iflong
\section{Algorithms for Instances with a Tree as Solid Graph}\label{sec:tree}
\fi
Since all edges of trees are bridges, we now describe, based on our insights about labels in~\minlab realizations, an algorithm for instances of~\URGD where the solid graph is a tree.

\iflong
\begin{lemma}
\label{lem:tree-nogap}
    If an instance of \URGD\ \tnew{that has a tree as solid graph is realizable}, then there is a \minlab realization such that each special bridge receives two consecutive labels.
\end{lemma}

\begin{proof}
    By Lemma~\ref{lem:tree-atmost2}, we can assume that
    no edge has more than two labels.
    Assume that an edge $e$ has two labels $a$ and $b$ with $b>a+1$.
    Every edge adjacent to $e$ either only has labels $\le a$
    or only has labels $\ge b$, as an adjacent edge with a label strictly
    between $a$ and $b$ or with a label $\le a$ and a label $\ge b$
    would create a triangle.
    Hence, the gap between $a$ and $b$ can be eliminated: In the stars
    centered at the endpoints of~$e$, shift all labels $\le a$ by adding
    $b-a-1$. Then, propagate the shifts to the remainder of the tree as
    follows: If the labels on an edge $f\neq e$ in one of the stars centered at
    endpoints of $e$ were shifted by $b-a+1$, apply that same shift to \emph{all} labels on all edges
    in the subtree reachable over $f$. This modification maintains a valid
    solution and reduces the number of edges with two labels that have
    a gap between them by at least $1$. Repeating this modification yields
    the claim of the lemma. 
    \end{proof}

Let an instance $D$ be given such that the graph of solid edges is a tree.
For any internal node of the tree, we refer to the subgraph induced by
the node and its neighbors as an induced star.
Because we can apply the splitting rule of Lemma~\ref{lem:split-nonspecial} to
any internal tree edge that is not special, it suffices
to consider instances in which all internal tree edges are special.
Furthermore, we can apply Lemma~\ref{lem:remove-pendant}
to reduce the problem to a tree in which there are no two nodes $u$ and $w$
adjacent to the same vertex $v$ such that both $(u,w)\notin A$ and $(w,u)\notin A$.
To see this, note that if neither $(u,w)$ nor $(w,u)$ are in~$A$,
then Observation~\ref{incident same label} together with Lemma~\ref{lem:special} implies that $vu$ and $vw$ are non-special
bridges and hence leaves of the tree, so Lemma~\ref{lem:remove-pendant}
can indeed be used to eliminate one of $u$ and~$w$.
For every induced star in the
resulting instance, the auxiliary graph of dashed arcs between its leaves is an
orientation of a complete graph. Furthermore, this orientation
is acyclic, as a directed cycle on the neighbors
$u_1,u_2,\ldots,u_k$ of $v$, with $\ell_i=\min\lambda(vu_i)$ and
$h_i=\max\lambda(vu_i)$ for $1\le i\le k$, would imply
that $\ell_1<h_2\le \ell_3\le h_3\le\cdots \le \ell_k \le h_k\le \ell_1$,
a contradiction. Therefore, the auxiliary graph of dashed arcs
between the leaves of the induced star is a complete DAG (also called
an acyclic tournament) if the instance is realizable. The following
definition captures the properties of the resulting simplified instances.

\begin{definition}
\label{def:simplifiedinstance}
    An instance of \RGD is a \emph{simplified tree \tnew{instance}}
if the solid edges form a tree where every
internal edge is special, every pendant edge is non-special,
and the auxiliary graph of dashed arcs between the leaves
of each induced star is a complete DAG.
\end{definition}

In order to solve the problem on an instance~$D$
such that the solid graph of solid edges is a tree, we can
split the instance into simplified tree
instances in such a way that the original instance
is realizable if and only if all the resulting simplified instances are
realizable. The total size of the resulting simplified instances
is polynomial in the size of the original instance as every
edge of the solid graph of the original instance occurs in at
most two of the simplified instances. Furthermore,
the splitting operations can be applied in polynomial time.
The discussion can be summarized in the following lemma.

\begin{lemma}
\label{lem:reductiontournament}
    There is a polynomial-time algorithm for the
    realization problem for instances where the solid edges
    form a tree if and only if there is a polynomial-time
    algorithm for the realization problem for simplified
    tree instances.
\end{lemma}

\ifold
Recall that we call an edge $e$ \emph{single-labeled} if $|\lambda(e)|=1$ (i.e., if it is a non-special bridge) and
\emph{double-labeled} if $|\lambda(e)|=2$ (i.e., if it is a special bridge).
\fi

\begin{observation}
\label{obs:nothreedouble}
    Let~$\lambda$ be an undirected realization of a directed graph~$D$.
    No three special bridges sharing one common endpoint share a label.
\end{observation}

\begin{proof}
    Let~$\lambda$ be an undirected realization for~$D$.
     Assume three solid bridges that share a common endpoint have the common label~$p$. There must exist
    two edges whose labels are $p$ and a label greater than $p$,
    or two edges whose labels are $p$ and a label smaller than $p$.
    In both cases, there must be a solid edge between the non-shared endpoints of these edges, a contradiction to the assumption
    that all edges are bridges.
\end{proof}

\begin{lemma}
    \label{lem:bundledmiddle}
    Consider an induced star in a simplified tree instance
    with center node $u$ and leaves $v_1,\ldots,v_k$ indexed in topological order. If the instance is realizable, the following must hold:
    If $uv_i$ and $uv_{i+2}$ are bundled for some $i\in [k-2]$, then $uv_i$ and
    $uv_{i+1}$ as well as $uv_{i+1}$ and $uv_{i+2}$ are also bundled.
\end{lemma}

\begin{proof}
    If $uv_i$ and $uv_{i+2}$ are bundled, then
    $uv_i$ and $uv_{i+2}$ must be special,
    $uv_{i+1}$ must be non-special, and all three
    edges share a label $\gamma$ in any realization.
    By Definition~\ref{def:bundled} one of the two edges
    $uv_i$ and $uv_{i+2}$
    must be special in $G[V_{v_i}\cup V_{v_{i+2}} \cup \{u\}]$.
    Assume that $uv_i$ is special.
    Thus, there must be a vertex $x$ in $V_{v_i}$ and
    a vertex $y$ in $V_{v_{i+2}}$ such that $D_{xu}=D_{v_iy}=1$
    and $D_{xy}=0$. But then we also have $D_{xu}=D_{v_i v_{i+1}}=1$
    and $D_{xv_{i+1}}=0$, and hence
    $uv_i$ is also special in $G[V_{v_i}\cup V_{v_{i+1}} \cup \{u\}]$,
    and thus $uv_i$ and $uv_{i+1}$ are bundled. Similarly, we have
    $D_{v_{i+1}v_{i+2}}=D_{uy}=1$ but $D_{v_{i+1}y}=0$, and hence
    $uv_{i+2}$ is also special in $G[V_{v_{i+1}}\cup V_{v_{i+2}} \cup \{u\}]$,
    and thus $uv_{i+1}$ and $uv_{i+2}$ are bundled.
    The case that $uv_{i+2}$ is special in $G[V_{v_i}\cup V_{v_{i+2}} \cup \{u\}]$
    can be handled analogously.
\end{proof}

We now give an algorithm for solving instances whose solid edges
form a tree.

\fi

\begin{theorem}
\label{th:treeCombinatorial}
    There is a polynomial-time algorithm for solving
    instances of \any\str\URGD for which the solid edges
    form a tree.
\end{theorem}
\iflong
\begin{proof}
    By Lemma~\ref{lem:reductiontournament}, it suffices to consider simplified tree instances. In particular,
    the auxiliary graphs of dashed arcs for all induced stars are
    complete DAGs.
    We compute a labeling for every star and merge the labelings in DFS order starting with an arbitrary star.

    Consider the computation of the labeling for the star centered at a node $u$.
    We know that every special edge incident with $u$ receives two consecutive
    labels and every non-special edge receives a single label. Furthermore, we know
    that the dashed arcs between the leaves of the star form a complete DAG.
    Denote the leaves by $v_1,v_2,\ldots,v_k$ in topological
    order. For $1\le i\le k$, call the edge $uv_i$ a \emph{double edge}
    if it is special and a \emph{single edge} otherwise.
    We label the edges in order from $uv_1$ to $uv_k$. Initially, we
    assign $uv_1$ the label $1$ if it is a single edge and the labels $1,2$
    if it is a double edge.
    For $i=2,3,\ldots,k$, distinguish the following cases:
    \begin{itemize}
        \item Case 1: $uv_{i-1}$ is a double edge labeled $a-1,a$, $uv_i$ is a single edge, and $uv_{i+1}$ is a single edge or $i=k$. If there is a vertex $w$ that can reach $u$ via $uv_{i-1}$ but cannot reach $v_{i}$ (which implies that $uv_{i-1}$ and $uv_i$ are bundled), assign label $a$ to $uv_i$. Otherwise, assign label $a+1$ to $uv_i$.
        \item Case 2: $uv_{i-1}$ is a double edge labeled $a-1,a$, $uv_i$ is a single edge, and $uv_{i+1}$ is a double edge. If there is a vertex $w$ that can reach $u$ via $uv_{i-1}$ but
        cannot reach $v_i$ or cannot reach all vertices that $u$ can reach via $v_{i+1}$ (implying
        that $uv_{i-1}$ and $uv_i$ are
        bundled, cf.\ Lemma~\ref{lem:bundledmiddle}), assign label $a$ to $uv_i$. Otherwise, assign
        label $a+1$ to $uv_i$.
        \item Case 3: $uv_{i-1}$ is a single edge labeled $a$, and $uv_i$ is a single edge. Label $uv_i$ with label $a+1$.
        \item Case 4: $uv_{i-1}$ is a single edge labeled $a$, and $uv_i$ is a double edge. If there
        is a vertex $w$ that can be reached from $u$ via the edge $uv_i$ but cannot be
        reached from \tnew{$v_{i-1}$} (implying that \tnew{$uv_{i-1}$} and $uv_i$ are bundled), assign labels $a,a+1$ to $uv_i$. Otherwise, assign
        labels $a+1,a+2$ to $uv_i$.
        \item Case 5: $uv_{i-1}$ is a double edge labeled $a-1,a$, and $uv_i$ is a double edge.
        If there is a vertex that can reach $u$ via $uv_{i-1}$ but cannot reach
        all vertices that $u$ can reach via $uv_i$ (implying
        that $uv_{i-1}$ and $uv_i$ are bundled), assign labels $a,a+1$ to $uv_i$.
        Otherwise, assign labels $a+1,a+2$ to $uv_i$.
    \end{itemize}
    The resulting solutions for the stars can be merged by starting with an
    arbitrary star and always adding an adjacent star, shifting the labels of
    that star so that the labels on the edge common to the solution so far
    and the new star are identical. If the resulting labeling realizes
    $D$, we output it. Otherwise, we report that the given instance
    is a non-instance.

    To prove that the algorithm is correct, we show that an arbitrary realization
    can be transformed into the one produced by the algorithm. Let $\lambda^*$
    be a \minlab labeling that realizes~$D$ and assigns two consecutive
    labels to each special bridge edge. Such a labeling exists
    by Lemma~\ref{lem:tree-nogap}.
    
    Consider a star $S$ with
    center $u$ and leaves $v_1,\ldots,v_k$ indexed in topological order.
    For $i< j$, it is clear that in $\lambda^*$ no label of \tnew{$uv_j$} can be smaller than
    any label of \tnew{$uv_i$}, as $v_i$ must not be reachable from~$v_j$.
    Besides, no two single edges can have the same label.
    Furthermore, we can assume that there are no gaps in the
    labels assigned to edges in $S$, as such gaps could be removed (and the
    labels of the other stars shifted accordingly). If edges to consecutive
    leaves do not share a label, we therefore have that the lower label of the
    second edge equals the higher label of the first edge plus one.
    The only possibilities for pairs of edges to consecutive leaves to share a label $a$ are:
    \begin{itemize}
        \item A single edge $uv_i$ with label $a$ and a double edge $uv_{i+1}$ with
        labels $a,a+1$, for some $a$.
        \item A double edge $uv_i$ with labels $a-1,a$ and a single edge $uv_{i+1}$ with
        label $a$, for some $i$.
        \item A double edge $uv_i$ with labels $a-1,a$ and a double edge $uv_{i+1}$ with labels $a,a+1$.
    \end{itemize}
    Furthermore, a label $a$ can be shared by edges to at most three consecutive leaves,
    and by Observation~\ref{obs:nothreedouble} at least one of them must be a single
    edge. It is also impossible that two of them are single edges (as the auxiliary
    graph is a complete DAG). Furthermore, the single edge must be the
    middle edge (otherwise it wouldn't be possible to share the label with both
    of the double edges), and so the only possibility for three edges to share
    a label is:
    \begin{itemize}
        \item A double edge $uv_{i-1}$ with labels $a-1,a$.
        \item A single edge $uv_i$ with label $a$.
        \item A double edge $uv_{i+1}$ with labels $a,a+1$.
    \end{itemize}
    Two edges $uv_i$ and $uv_j$ with $i<j$ \emph{must} share a label
    if the edges are bundled, which is the case if there exists a vertex that can reach $u$ via $uv_i$ but cannot reach all vertices
    that $u$ can reach via $uv_j$.
    If there is no such vertex, assigning disjoint sets of labels to $uv_i$ and $uv_j$
    is valid with respect to the vertices involved as all vertices that can reach $u$ 
    via $uv_i$ can then definitely reach all vertices that $u$ can reach via $uv_j$,
    as desired. The algorithm ensures that all edges that are bundled (must share labels)
    do share a label, and assigns disjoint sets of labels to consecutive edges
    otherwise. If $\lambda^*$ differs from the labeling
    produced by the algorithm this can only happen if $\lambda^*$ lets
    consecutive edges $uv_i$ and $uv_{i+1}$ share a label even though every
    vertex that can reach $u$ via $uv_i$ can reach every vertex that $u$ can
    reach via $uv_{i+1}$ and there is also no indirect requirement (stemming
    from a requirement that $uv_i$ and $ uv_{i+2}$ must share a label,
    or that $uv_{i-1}$ and $uv_{i+1}$ must share a label) for
    $uv_i$ and $uv_{i+1}$ to share a label. In such a case, we can shift
    all the labels on $uv_{i+1}$, $uv_{i+2}$, \ldots, $uv_k$ by $+1$ (and
    propagating the shifted labels to the rest of the tree)
    without affecting feasibility of $\lambda^*$. Repeating this operation
    makes $\lambda^*$ identical to the labeling produced by the algorithm.
\end{proof}

The following lemma implies an alternative method for solving the problem of realizing~$D$ in polynomial time in the case where the solid edges
form a tree. 
This method can be generalized to the
case where some edges of the tree are pre-labeled, which will be useful
later.

\begin{lemma}
\label{lem:treeLP}
Let $D$ be an instance
of \any\str\URGD
on $n$ vertices in which the solid edges form a tree~$T$.
Then one can determine in polynomial time a
linear program that is feasible if and only if $D$ is realizable.
If $D$ is realizable, one can compute a labeling that realizes~$D$
from the solution of the LP in polynomial time.
\end{lemma}

\begin{proof}
First, we check conditions that must obviously hold if $D$ is realizable:
If $(u,v)$ is an arc, then $(u,w)$ and $(w,v)$ must also be arcs for all internal
vertices of the path from $u$ to $v$ in~$T$. If $uc$ and $cv$ are
adjacent edges and neither \tnew{$(u,v)$ nor $(v,u)$} is in $A$,
then $uc$ and $cv$ must be non-special edges (as they must
both be assigned the same single label in any realization).
Call a non-arc $(u,v)\notin A$ \emph{minimal} if it is the only arc missing
in the direction from $u$ to $v$ on the path from $u$ to $v$.
Consider any minimal non-arc $(u,v)$ such that the path from $u$
to $v$ in $T$ contains at least three edges.
Let $e_1,e_2,\ldots,e_p$ for some $p\ge 3$ be the edges on the
path from $u$ to~$v$.
For $i\in [p]$, let $e_i=\{u_{i-1},u_i\}$, with $u_{i-1}$ the endpoint closer to~$u$.
Then $e_1$ and $e_p$
must be non-special edges and $e_{i}$ for $i\in[2,p-1]$ must be special edges.
The latter part of the claim follows because $D_{u,u_i}=D_{u_{i-1}v}=1$ but
$D_{uv}=0$. If $e_1$ was special, there would necessarily be a temporal
path from $u$ to $v$ using the lower labels on $e_i$ for $i\in[p-1]$
and any label on $e_p$. Similarly, if $e_1$ was non-special and $e_p$ was special,
there would necessarily
be a temporal path from $u$ to $v$ using the unique label on $e_1$ and
the higher labels on $e_i$ for $i\in[2,p]$. Therefore, $e_1$ and $e_p$
are non-special.
If any of these conditions does not hold, the given instance cannot
be realized, and the algorithm can terminate \tnew{and} output an arbitrary infeasible LP.

Assume now that the conditions above are met.
We show how to formulate the problem of finding the required labels as a linear program (LP) of
polynomial size. The LP solution may yield rational numbers as labels, but we
can transform the labels to integers in the end. Let $\delta=1$.
To express strict inequality between two variables, we will use an inequality constraint
that requires one variable to be at least as large as the other variable plus~$\delta$.
(The choice of $\delta=1$ is arbitrary, any positive $\delta$ would work.)
The LP has two non-negative variables $\ell_e$ and $h_e$ for each edge~$e$ of $T$.
The values of these variables represent the labels assigned to~$e$, using
the convention that $\ell_e=h_e$ represents the case that $e$ has a single label
and $\ell_e<h_e$ \tnew{the case that $e$ has two labels}.
For every edge $e$, we therefore add the constraint
\begin{align}
    \ell_e = h_e
\end{align}
if $e$ is non-special and the constraint
\begin{align}
    \ell_e + \delta \le h_e
\end{align}
if $e$ is special.

For every pair of adjacent  edges $e=\{u,c\}$ and $e'=\{c,v\}$,
we add constraints as follows:
If neither $(u,v)$ nor $(v,u)$ is in $A$, both edges
must be non-special and we add the constraint
\begin{align}
    \ell_e = \ell_{e'}
\end{align}
to ensure that their single labels are equal.
If $(u,v)\in A$ (and necessarily $(v,u)\notin A$), we add the constraints
\begin{align}
\ell_e + \delta \le h_{e'} \\
h_e \le \ell_{e'} \label{eq:consechl}
\end{align}
to ensure there is a temporal path from $u$ to $v$ but none
from $v$ to~$u$.

Call a non-arc $(u,v)\notin A$ \emph{minimal} if it is the only arc missing
in the direction from $u$ to $v$ on the path from $u$ to $v$.
Consider any minimal non-arc $(u,v)$.
Let $e_1,e_2,\ldots,e_p$ for some $p\ge 2$ be the edges on the
path from $u$ to~$v$. If $p=2$, the constraints above already ensure that
$u$ cannot reach $v$, so assume $p\ge 3$. As discussed above,
$e_1$ and $e_p$
must be non-special edges and $e_{i}$ for $i\in[2,p-1]$ must be special edges
in this case.
We add the following constraints to ensure that
$u$ cannot reach $v$:
\begin{align}
    h_{e_i}=\ell_{e_{i+1}}, \quad\mbox {for $i\in [p-1]$}
\end{align}
These constraints ensure that a temporal path from $u$ towards
$v$ uses the single label on $e_1$, then the higher label
on $e_i$ for $i\in [2,p-1]$, and then cannot
traverse $e_p$ because the higher label on $e_{p-1}$ is
equal to the single label of~$e_p$.

Call an arc $(u,v)\in A$ \emph{maximal} if there is no arc $(x,y)\in A$
such that the path from $x$ to $y$ in $T$ contains the path from $u$ to $v$
as a proper subpath. If the path from $u$ to $v$ contains at most
one non-special edge $e$, then $u$ can necessarily reach $v$
as there is a temporal path from $u$ to $v$ that uses the lower label
on the special edges before $e$ (or all the way to $v$ if there
is no non-special edge $e$ on the path) and the higher label on the
edges after~$e$. If the path contains at least two special edges,
it suffices to add constraints to ensure, for
any two non-special edges $e$ and $e'$
such that the edges between $e$ and $e'$ are all special, \tnew{that} the path
with first edge $e$ and final edge $e'$ contains labels that make it a temporal path.
Let $e_1,e_2,\ldots,e_p$ with $p\ge 2$ be the edges on the path with first
edge $e=e_1$ and final edge $e'=e_p$,
ordered from $e$ to $e'$. We must exclude the label assignment that
satisfies $h_{e_i}=\ell_{e_{i+1}}$ for all $i\in [p-1]$, because
with that label assignment $u$ can reach the vertex before $v$
only in step $h_{e_{p-1}}$ and thus cannot take the non-special
edge $e_p$ with label $\ell_{e_p}=h_{e_{p-1}}$. Therefore,
we add the following constraint:
\begin{align}
    \sum_{i=1}^{p-1} (\ell_{e_{i+1}}-h_{e_i}) \ge \delta
\end{align}
As we have $\ell_{e_{i+1}}-h_{e_i}\ge 0$ by (\ref{eq:consechl}),
this constraint ensures that for at least one $i$ we have
$\ell_{e_{i+1}}-h_{e_i}>0$. Let $i^*$ denote such a value of~$i$.
Then the path with first edge $e$ and final edge $e'$ gives
a temporal path because we can use the unique label on $e$,
the larger label on edges $e_2,\ldots,e_i$, the smaller
label on edges $e_{i+1},\ldots, e_{p-1}$, and the unique
label on $e'$.

To show that the LP is feasible if and only if $D$ is realizable,
we can argue as follows: If $\lambda^*$ is a labeling that
realizes $D$, it is clear that setting the variables of the
LP based on $\lambda^*$ is a feasible solution of the LP.
For the other direction, assume that the LP has a feasible
solution, and let $\lambda$ be the labeling with rational
numbers determined by that
feasible solution. The LP constraints were chosen so that
$u$ cannot reach $v$ for every minimal non-arc $(u,v)$
and $u$ can reach $v$ for every maximal arc $(u,v)$,
where the labels are allowed to be rational numbers.
Thus, $\lambda$ realizes~$D$. Furthermore, the labeling
$\lambda$ can be made integral by sorting all labels
and replacing them by integers that maintain their order.

It is clear that the LP has polynomial size and can be constructed
in polynomial time.
\end{proof}
\fi

\iflong
We now show that the LP-based approach of
Lemma~\ref{lem:treeLP} can be extended to a setting
where some edges have been pre-labeled.
The same
approach works for an arbitrary number of edges with pre-labeled labels,
but we state it for the special case of two pre-labeled edges as
this is what we will require to solve a subproblem
that arises in the FPT algorithm in Section~\ref{sec:fes}.

\begin{corollary}
\label{cor:prelabeled_tree}
Let $D$ be an instance of \any\str\URGD on $n$ vertices in which the solid edges form a tree~$T$.
Assume that the labels of two solid edges $f_1$ and $f_2$ have been
pre-determined (such that $f_1$ receives two labels if it is special and one label
otherwise, and the same holds for~$f_2$). Assume further
that the labels assigned to $f_1$ and $f_2$ are multiples of an integer $\ge 2n$.
Then one can determine in polynomial time whether there exists
a \minlab\ realization for $D$ that agrees with the pre-labeling
(and output such a labeling if the answer is yes).
\end{corollary}

\begin{proof}
Construct the LP that expresses realizability of $D$ as in the proof
of Lemma~\ref{lem:treeLP}. For the pre-labeled edges $f_1$ and $f_2$ we add equality
constraints ensuring that the values of the variables representing their labels
are consistent with the pre-labeling. If there exists a labeling $\lambda^*$
that agrees with the pre-labeling and realizes $D$, then setting the variables
of the LP in accordance with those labels constitutes a feasible solution to
the LP. Moreover, if the LP has a feasible solution, the values of the variables
represent a fractional labeling $\lambda$ that realizes $D$, as shown in the proof
of Lemma~\ref{lem:treeLP}. To make the labeling integral, we can sort all labels
and replace them by integers that maintain their order and keep the labelings
of $f_1$ and $f_2$ unchanged. As the labels on $f_1$ and $f_2$ are multiples of a number $\ge 2n$
\tnew{and} we have fewer than $2n$ distinct labels in $\lambda$ (the tree $T$ has $n-1$ edges, and each edge has
at most two labels), the gaps between the labels assigned to $f_1$ and $f_2$ contain at
least $2n-1$ available integers, and this is sufficient even if all labels fall into such a
gap.
\end{proof}
\fi

\iflong
We remark that the variants of \URGD different from
\any\str\URGD can also be solved in polynomial time
as follows if the graph of solid edges is a tree.
By Corollary~\ref{cor:nospecial}, the given instance
is a no-instance for these variants of \URGD if it
contains at least one special bridge, so assume
that all solid edges are non-special.
Thus, we only
need to consider labelings with a single label per
edge.
For \simp\str\URGD, we construct and solve the LP of
Lemma~\ref{lem:treeLP}. If it is infeasible, the given
instance is a no-instance; otherwise, the labeling
obtained from the LP solution is a realization using
a single label per edge.
It remains to consider all non-strict variants
of \URGD, including \pro\URGD and \happy\URGD.
If there exist two adjacent edges $uc$ and $cv$
such that neither $(u,v)$ nor $(v,u)$ is in $A$,
the given instances is a no-instance for all these
variants.
Otherwise, we solve the LP from Lemma~\ref{lem:treeLP}.
If it is infeasible, the given instance has no solution.
Otherwise, the labeling obtained from the LP solution
is a realization with a single label per edge in which no two adjacent
edges receive the same label, and hence that realization
constitutes a solution for all variants of \nstr\URGD including
\pro\URGD and \happy\URGD.
As an alternative to using the LP-based approach from
Lemma~\ref{lem:treeLP} to check the existence of the desired
realizations in the discussion above,
one can also use the combinatorial
algorithm from Theorem~\ref{th:treeCombinatorial}.
\fi

\iflong
\fi

\section{Hardness for Undirected Reachability Graph Realizability}\label{sec:hardundir}
In this section, we show that all considered versions of~\URGDlong are \NP-hard, even on \iflong a family of instances where each instance has a constant maximum degree and, if it is a yes-instance, can be realized with a constant number of different time labels.\fi

\newcommand{\SAT}{\textsc{2P2N-3SAT}\xspace}

\iflong

\begin{observation}\label{forced label}
Let~$D$ be an instance of some version of~\URGD and let~$G$ denote the solid graph of~$D$.
Moreover, let~$v$ be a vertex of~$D$.
For each vertex~$w\in N_G(v)$ that has (i)~no incoming arc from any other vertex of~$N_G(v)$ or (ii)~no outgoing arc to any other vertex of~$N_G(v)$, the edge~$\{v,w\}$ receives at least one label in every realization of~$D$.
\end{observation}
To see this, consider a solid neighbor~$w$ of~$v$ which has no ingoing arcs form any vertex of~$N_G(v)$.
Then the path~$(v,w)$ is the only dense  path from~$v$ to~$w$, which implies that the edge~$\{v,w\}$ receives at least one label.
For a solid neighbor~$w$ with no outgoing arcs in~$N_G(v)$, the path~$(w,v)$ is the only dense path from~$w$ to~$v$.
This then also implies that the edge~$\{v,w\}$ receives at least one label.
\fi

\iflong
\begin{theorem}\label{hardness trianglefree}
\any\str\URGD and~\simp\str\URGD are \NP-hard on directed graphs of constant maximum degree  that have a triangle-free solid graph.
\end{theorem}
\newcommand{\figtriangle}{
\begin{figure}

\centering
\begin{tikzpicture}[scale=1.2]
\tikzstyle{k}=[circle,fill=white,draw=black,minimum size=10pt,inner sep=2pt]

\node[k] (cl) at (-1.5,2) {$s_i$};
\node[k] (cr) at (9.5,2) {$s_j$};

\node[k] (bx) at (2,0) {$b_x$};
\node[k] (bxnx) at (4,0) {$b_{x,\overline{x}}$};
\node[k] (bnx) at (6,0) {$b_{\overline{x}}$};

\node[k] (ax) at (2,4) {$a_x$};
\node[k] (axnx) at (4,4) {$a_{x,\overline{x}}$};
\node[k] (anx) at (6,4) {$a_{\overline{x}}$};

\node[k] (fx) at (2,-1) {$f_x$};
\node[k] (fxnx) at ($(fx) + (2,0)$) {$f_{x,\overline{x}}$};
\node[k] (fnx) at ($(fxnx) + (2,0)$) {$f_{\overline{x}}$};

\node[k] (qx) at (1,5.5) {$q_x$};
\node[k] (qnx) at ($(qx) + (6,0)$) {$q_{\overline{x}}$};
\node[k] (qax) at ($(qx) + (1,0)$) {$q_{a_x}$};
\node[k] (qanx) at ($(qx) + (5,0)$) {$q_{a_{\overline{x}}}$};

\node[k] (cx) at (4,2) {$c_x$};
\node[k] (x) at (2,2) {$x$};
\node[k] (nx) at (6,2) {$\overline{x}$};

\node[k] (t1) at (2.8,8) {$t_k$};
\node[k] (t2) at (5.2,8) {$t_\ell$};

\draw[draw=black, dashed] ($(qx) + (-.5,.5)$) rectangle ++(7,-7.5);

\draw[thick] (x) -- (qx) node [midway, fill=white] {$1$};
\draw[thick] (nx) -- (qnx) node [midway, fill=white] {$1$};
\draw[thick] (ax) -- (qax) node [midway, fill=white] {1};
\draw[thick] (anx) -- (qanx) node [midway, fill=white] {1};

\draw[thick] (cx) -- (bxnx) node [midway, fill=white] {1};

\draw[thick] (bx) -- (bxnx) node [midway, fill=white] {2};
\draw[thick] (bnx) -- (bxnx) node [midway, fill=white] {3};

\draw[thick] (x) -- (ax) node [midway, fill=white] {12};
\draw[thick] (nx) -- (anx) node [midway, fill=white] {13};
\draw[thick] (axnx) -- (ax) node [midway, fill=white] {13};
\draw[thick] (axnx) -- (anx) node [midway, fill=white] {14};
\draw[thick] (axnx) -- (cx) node [midway, fill=white] {25};

\draw[thick,red] (x) -- (cx) node [midway, fill=white] {13};
\draw[thick,blue] (nx) -- (cx) node [midway, fill=white] {14};

\draw[thick] (x) edge[bend right=35] node[fill=white]{13} (fx);
\draw[thick] (nx) edge[bend left=35] node[fill=white]{14} (fnx);

\draw[thick] (cx) edge[bend left=13] node[fill=white]{$\alpha$} (t1);
\draw[thick] (cx) edge[bend right=13] node[fill=white]{$\beta$} (t2);

\draw[thick] (fxnx) -- (fx) node [midway, fill=white] {13};
\draw[thick] (fxnx) -- (fnx) node [midway, fill=white] {14};

\draw[thick] (x) -- (bx) node [midway, fill=white] {14};
\draw[thick] (nx) -- (bnx) node [midway, fill=white] {15};

\draw[thick] (cl) -- (x) node [midway, fill=white] {$\chi(y_i) + 1\in [2,11]$};
\draw[thick] (cr) -- (nx) node [midway, fill=white] {$\chi(y_j) + 1\in [2,11]$};

\foreach \x in {qx,qnx,qax,qanx} {
    \draw[thick] (\x) -- (t1) node [midway, fill=white] {$\alpha$};               
    \draw[thick] (\x) -- (t2) node [midway, fill=white] {$\beta$};    
}

\end{tikzpicture}

~
~

~
~

~
~

~
~\scalebox{.75}{
\begin{tabular}{c|||c|c|c|c||c|c|c|c|c|c|c|c|c|c|c|c|c|c|c|c} & $s_i$  & $s_j$  & $t_k$  & $t_\ell$  & $x$  & ${\overline{x}}$  & $a_x$  & $a_{\overline{x}}$  & $a_{x,\overline{x}}$  & $b_x$  & $b_{\overline{x}}$  & $b_{x,\overline{x}}$  & $c_x$  & $f_x$  & $f_{\overline{x}}$  & $f_{x,\overline{x}}$  & $q_x$  & $q_{\overline{x}}$  & $q_{a_x}$  & $q_{a_{\overline{x}}}$  \\\hline\hline\hline
 $s_i$  &   & 0 & \cellcolor{red!25} 1 & \cellcolor{red!25} 1 & 1 & 0 & \cellcolor{black!25} 1 & \cellcolor{black!25} 1 & \cellcolor{black!25} 1 & \cellcolor{black!25} 1 & 0 & 0 & \cellcolor{black!25} 1 & \cellcolor{black!25} 1 & 0 & 0 & 0 & 0 & 0 & 0 \\\hline
 $s_j$  & 0 &   & \cellcolor{blue!25} 1 & \cellcolor{blue!25} 1 & 0 & 1 & 0 & \cellcolor{black!25} 1 & \cellcolor{black!25} 1 & 0 & \cellcolor{black!25} 1 & 0 & \cellcolor{black!25} 1 & 0 & \cellcolor{black!25} 1 & 0 & 0 & 0 & 0 & 0 \\\hline
 $t_k$  & 0 & 0 &   & \cellcolor{black!25} 1 & 0 & 0 & 0 & 0 & \cellcolor{black!25} 1 & 0 & 0 & 0 & 1 & 0 & 0 & 0 & 1 & 1 & 1 & 1 \\\hline
 $t_\ell$  & 0 & 0 & 0 &   & 0 & 0 & 0 & 0 & \cellcolor{black!25} 1 & 0 & 0 & 0 & 1 & 0 & 0 & 0 & 1 & 1 & 1 & 1 \\\hline\hline
 $x$  & 1 & 0 & \cellcolor{black!25} 1 & \cellcolor{black!25} 1 &   & 0 & 1 & \cellcolor{black!25} 1 & \cellcolor{black!25} 1 & 1 & 0 & 0 & 1 & 1 & 0 & 0 & 1 & 0 & 0 & 0 \\\hline
 ${\overline{x}}$  & 0 & 1 & \cellcolor{black!25} 1 & \cellcolor{black!25} 1 & 0 &   & 0 & 1 & \cellcolor{black!25} 1 & 0 & 1 & 0 & 1 & 0 & 1 & 0 & 0 & 1 & 0 & 0 \\\hline
 $a_x$  & 0 & 0 & \cellcolor{black!25} 1 & \cellcolor{black!25} 1 & 1 & 0 &   & \cellcolor{black!25} 1 & 1 & \cellcolor{black!25} 1 & 0 & 0 & \cellcolor{black!25} 1 & \cellcolor{black!25} 1 & 0 & 0 & 0 & 0 & 1 & 0 \\\hline
 $a_{\overline{x}}$  & 0 & 0 & \cellcolor{black!25} 1 & \cellcolor{black!25} 1 & 0 & 1 & 0 &   & 1 & 0 & \cellcolor{black!25} 1 & 0 & \cellcolor{black!25} 1 & 0 & \cellcolor{black!25} 1 & 0 & 0 & 0 & 0 & 1 \\\hline
 $a_{x,\overline{x}}$  & 0 & 0 & 0 & 0 & 0 & 0 & 1 & 1 &   & 0 & 0 & 0 & 1 & 0 & 0 & 0 & 0 & 0 & 0 & 0 \\\hline
 $b_x$  & 0 & 0 & 0 & 0 & 1 & \cellcolor{black!25} 1 & 0 & 0 & 0 &   & \cellcolor{black!25} 1 & 1 & 0 & 0 & 0 & 0 & 0 & 0 & 0 & 0 \\\hline
 $b_{\overline{x}}$  & 0 & 0 & 0 & 0 & 0 & 1 & 0 & 0 & 0 & 0 &   & 1 & 0 & 0 & 0 & 0 & 0 & 0 & 0 & 0 \\\hline
 $b_{x,\overline{x}}$  & 0 & 0 & \cellcolor{black!25} 1 & \cellcolor{black!25} 1 & \cellcolor{black!25} 1 & \cellcolor{black!25} 1 & 0 & 0 & \cellcolor{black!25} 1 & 1 & 1 &   & 1 & 0 & 0 & 0 & 0 & 0 & 0 & 0 \\\hline
 $c_x$  & 0 & 0 & 1 & 1 & 1 & 1 & 0 & 0 & 1 & \cellcolor{black!25} 1 & \cellcolor{black!25} 1 & 1 &   & 0 & 0 & 0 & 0 & 0 & 0 & 0 \\\hline
 $f_x$  & 0 & 0 & 0 & 0 & 1 & 0 & 0 & 0 & 0 & \cellcolor{black!25} 1 & 0 & 0 & 0 &   & \cellcolor{black!25} 1 & 1 & 0 & 0 & 0 & 0 \\\hline
 $f_{\overline{x}}$  & 0 & 0 & 0 & 0 & 0 & 1 & 0 & 0 & 0 & 0 & \cellcolor{black!25} 1 & 0 & 0 & 0 &   & 1 & 0 & 0 & 0 & 0 \\\hline
 $f_{x,\overline{x}}$  & 0 & 0 & 0 & 0 & 0 & 0 & 0 & 0 & 0 & 0 & 0 & 0 & 0 & 1 & 1 &   & 0 & 0 & 0 & 0 \\\hline
 $q_x$  & \cellcolor{black!25} 1 & 0 & 1 & 1 & 1 & 0 & \cellcolor{black!25} 1 & \cellcolor{black!25} 1 & \cellcolor{black!25} 1 & \cellcolor{black!25} 1 & 0 & 0 & \cellcolor{black!25} 1 & \cellcolor{black!25} 1 & 0 & 0 &   & 0 & 0 & 0 \\\hline
 $q_{\overline{x}}$  & 0 & \cellcolor{black!25} 1 & 1 & 1 & 0 & 1 & 0 & \cellcolor{black!25} 1 & \cellcolor{black!25} 1 & 0 & \cellcolor{black!25} 1 & 0 & \cellcolor{black!25} 1 & 0 & \cellcolor{black!25} 1 & 0 & 0 &   & 0 & 0 \\\hline
 $q_{a_x}$  & 0 & 0 & 1 & 1 & \cellcolor{black!25} 1 & 0 & 1 & \cellcolor{black!25} 1 & \cellcolor{black!25} 1 & \cellcolor{black!25} 1 & 0 & 0 & \cellcolor{black!25} 1 & \cellcolor{black!25} 1 & 0 & 0 & 0 & 0 &   & 0 \\\hline
 $q_{a_{\overline{x}}}$  & 0 & 0 & 1 & 1 & 0 & \cellcolor{black!25} 1 & 0 & 1 & \cellcolor{black!25} 1 & 0 & \cellcolor{black!25} 1 & 0 & \cellcolor{black!25} 1 & 0 & \cellcolor{black!25} 1 & 0 & 0 & 0 & 0 &    
  \end{tabular}}
\caption{The solid graph and adjacency matrix for the variable gadget in the hardness reduction for~\any\str\URGD and~\simp\str\URGD version on triangle-free graph (\Cref{hardness trianglefree}). 
The highlighted cells indicate the dashed arcs.
As we show, no realization assigns labels to both the red and the blue edge.
If the red (blue) edge receives a label, the two red (blue) cells are realized via paths through the variable gadget~$V_x$.
Here, $\alpha = \chi(y_k) + 14$ and~$\beta = \chi(y_\ell) + 14$, which are both values in~$[15,24]$.
}
\label{fig:hardness triangle label}
\end{figure}
}
\begin{proof}
We reduce from \SAT which is known to be \NP-hard~\cite{T84}.

\prob{\SAT}{A CNF formula~$\phi$, where each clause has size at most~$3$ and each variable occurs exactly twice positively and twice negatively.}{Is~$\phi$ satisfiable?}

Let~$\phi$ be an instance of~\SAT with variables~$X$ and clauses~$Y := \{y_1, \dots, y_{|Y|}\}$.
Moreover, let~$\chi \colon Y \to [1,10]$ be a coloring of the clauses, such that for each two distinct clauses~$y_i$ and~$y_j$ of~$Y$ that share a variable, $\chi(y_i) \neq \chi(y_j)$.
Such a coloring~$\chi$ exists and can be greedily computed in polynomial time.\footnote{To see this, consider the auxiliary graph~$H_Y$ with vertex set~$Y$ and edges between to clauses~$y_i$ and~$y_j$ of~$Y$ if and only if~$y_i$ and~$y_j$ share a variable.
The graph~$H_Y$ has a maximum degree of~$9$, since each clause contains three literals, and each variable occurs in exactly four clauses.
Hence,~$\chi$ is a greedy proper~$(\Delta(H_Y) + 1)$-vertex coloring of~$H_Y$ which exists and can be computed in polynomial time.}
We define a directed graph~$D=(V,A)$ in polynomial time, such that there is a simple undirected temporal graph with strict reachability graph equal to~$D$ if and only if~$\phi$ is satisfiable.

\figtriangle

We initialize~$V$ and~$A$ as empty sets.
For each variable~$x\in X$, we add the following vertices to~$V$: $V_x := \{x,\overline{x},a_x,a_{\overline{x}},a_{x,\overline{x}},b_x,b_{\overline{x}},b_{x,\overline{x}},c_x,f_x,f_{\overline{x}},f_{x,\overline{x}},q_x,q_{a_x},q_{\ol},q_{a_{\ol}}\}$.
See~\Cref{fig:hardness triangle label} for the vertices and arcs of~$D[V_x]$.
For each clause~$y_i\in Y$, we add two vertices~$s_i$ and~$t_i$ to~$V$, that is, a source~$s_i$ and a terminal~$t_i$ for each clause.
Finally, for each two distinct clauses~$y_i$ and~$y_j$ from~$Y$ that share a variable, we add vertices~$q_{(i,j)}$ and~$q_{(j,i)}$ to~$V$.
Note that these vertices and their incident edges are not shown in~\Cref{fig:hardness triangle label}.

Next, we describe the arcs~$A$ of~$D$.
To this end, we first describe the solid edges~$E$ of~$G$. 
We initialize~$E$ as the empty set.
For each variable~$x\in X$, we add to~$E$ the solid edges between the vertices of~$V_x$ that appear in the graph of~\Cref{fig:hardness triangle label}. For each clause~$y_i$, we add a solid edge to~$E$ between the vertex~$s_i$ and each vertex $\ell$ for which~$\ell$ is a literal of~$y_i$. 
Moreover, for each variable~$x$ occurring in clause~$y_i$, we add the solid edges~$\{t_i, q_{x}\},\{t_i, q_{a_x}\},\{t_i, q_{\ol}\}, \{t_i, q_{a_{\ol}}\}$, and~$\{t_i, c_x\}$ to~$E$.
Finally, for each two distinct clauses~$y_i$ and~$y_j$ that share a variable, we add the solid edges~$\{s_i,q_{(i,j)}\},\{q_{(i,j)}, t_j\}, \{s_j,q_{(j,i)}\}$, and~$\{q_{(j,i)}, t_i\}$ to~$E$. 
Next, we describe the dashed arcs of~$D$.
For each variable~$x\in X$, we add to~$A$ the dashed arcs between the vertices of~$V_x$ as described in the matrix of~\Cref{fig:hardness triangle label}. Let~$y_i$ be a clause of~$Y$ and let~$x$ be a variable occurring in~$y_i$.
We add the dashed arcs~$(s_i,c_x), (s_i,a_{x,\ol}),$ $(s_i,a_{\ol}), (x,t_i),(\overline{x},t_i),(a_x,t_i),(a_{\ol},t_i),$ and~$(b_{x,\ol},t_i)$ to~$A$.
If~$x$ occurs positively in clause~$y_i$, we add the dashed arcs~$(s_i,a_x),(s_i,b_x),(s_i,f_x),$ and~$(q_x,s_i)$ to~$A$.
Otherwise, we add the dashed arcs~$(s_i,b_{\ol}),(s_i,f_{\ol})$ and~$(q_{\ol},s_i)$ to~$A$.
Let~$y_i$ and~$y_j$ be distinct clauses from~$Y$ that share a variable and assume without loss of generality that~$\chi(y_i) < \chi(y_j)$.
We add the dashed arcs~$(s_i,t_j),(s_j,t_i)$, and~$(t_i,t_j)$ to~$A$.
If~$y_i$ and~$y_j$ share a literal, we also add the arc~$(s_i,s_j)$ to~$A$.
Finally, for each clause~$y_i\in Y$, we add the arc~$(s_i,t_i)$ to~$A$.

This completes the construction of~$D$.

\begin{claim}
$D$ has a maximum degree of~$\Oh(1)$. Moreover, $G=(V,E)$ is triangle-free.
\end{claim}
\begin{claimproof}
The solid graph induced by the vertices of each variable gadget has constant degree by inspection. Every vertex $s_i$ has at most $3+9=12$ edges of type $\{s_i,\ell\}$, $\{s_i,q_{(i,j)}\}$, respectively, at most $13$ arcs due to variables occurring in $y_i$, at most $9$ arcs due to $y_i$ sharing variables with other clauses, and at most $6$ arcs due to $y_i$ sharing literals with other clauses. Every $t_i$ has at most $15$ edges and $15$ arcs due to variables occurring in $y_i$, and at most $9$ edges and $18$ arcs due to $y_i$ sharing variables with other clauses. For the same (constant number of) types of edges and arcs, each of the other vertices of $D$ is related to a variable or literal, which can occur in at most 4 clauses; by inspection, it follows that the degree contributed to each such vertex by each edge or arc type is constant. Thus, $D$ has a maximum degree of~$\Oh(1)$.   
The solid graph induced by the vertices of each variable gadget is triangle-free by inspection. The edge set of $G[\bigcup_i(\{s_i\}\cup\{t_i\})]$ is empty, so any triangle must be using at most one source or terminal vertex. Each $q_{(i,j)}$ vertex is adjacent to no other vertex than $s_i$ and $t_j$, thus cannot be part of a triangle. As no two vertices belonging to different gadgets are adjacent, it remains to show that a source or terminal cannot form a triangle with two vertices belonging to the same gadget. A source is only adjacent to the vertices corresponding to literals and the terminals are only adjacent to vertices $q_x$, $q_{a_x}$, $q_{\ol}$, $q_{a_{\ol}}$, and $c_x$, none of which are adjacent to each other. Thus, $G=(V,E)$ is triangle-free.
\end{claimproof}

Next, we show that~$\phi$ is satisfiable if and only if there is an undirected temporal graph with strict reachability graph~$D$.
More precisely, we show that if~$\phi$ is satisfiable, then there is a simple undirected temporal graph with strict reachability graph~$D$.

$(\Leftarrow)$
Let~$\mg = (G=(V,E),\lambda\colon E \to 2^{\mathbb{N}})$ be an undirected temporal graph with strict reachability graph equal to~$D$.
We show that~$\phi$ is satisfiable.

To this end, we first show that for each variable~$x\in X$, at most one of the edges of~$\{x,c_x\}$ and~$\{\ol,c_x\}$ receives a non-empty set of labels under~$\lambda$.
Assume towards a contradiction that there is a variable~$x\in X$, such that both~$\{x,c_x\}$ and~$\{\ol,c_x\}$ receives a non-empty set of labels under~$\lambda$.
We will show that the arc~$(f_x,f_{\ol})$ is not realized.
Since~$(x,\ol)\notin A$ and~$(\ol,x)\notin A$, \Cref{incident same label} implies that there is a label~$\alpha$, such that~$\lambda(\{x,c_x\}) = \lambda(\{\ol,c_x\}) = \{\alpha\}$.
Now, consider the neighborhood of~$f_x$ ($f_{\ol}$).
This vertex has only two solid neighbors: $x$ and~$f_{x,\ol}$ ($\ol$ and~$f_{x,\ol}$).
Moreover, both these neighbors are not joined by a solid edge.
Due to~\Cref{forced label}, this implies that the edges~$\{x,f_x\},\{f_x,f_{x,\ol}\},\{f_{x,\ol},f_{\ol}\}$, and~$\{f_{\ol},\ol\}$ all receive at least one label under~$\lambda$.
Since~$A$ contains none of the arcs~$(c_x,f_x),(f_x,c_x),(x,f_{x,\ol})$, and~$(f_{x,\ol},\ol)$, \Cref{incident same label} implies that~$\lambda(\{f_{x,\ol},f_x\}) = \lambda(\{f_x,x\}) = \lambda(\{x,c_x\}) = \{\alpha\}$.
Similarly, $\lambda(\{f_{x,\ol},f_{\ol}\}) = \lambda(\{f_{\ol},\ol\}) = \lambda(\{\ol,c_x\}) = \{\alpha\}$.
Thus, each edge incident with both~$f_x$ and~$f_{\ol}$ receives the label set~$\{\alpha\}$ under~$\lambda$.
This contradicts the assumption that the arc~$(f_x,f_{\ol})\in A$ is realized by the temporal graph~$\mg$.
Consequently, for each variable~$x\in X$, at most one of the edges of~$\{x,c_x\}$ and~$\{\ol,c_x\}$ receives a non-empty set of labels under~$\lambda$.

Based on this observation, we define a truth assignment~$\pi$ of the variables of~$X$ that satisfies~$\phi$.
For each variable~$x\in X$, we set $\pi(x) := \texttt{true}$ if and only if the edge~$\{x,c_x\}$ receives at least one label under~$\lambda$.
Next, we show that~$\pi$ satisfies~$\phi$.

Let~$y_i$ be a clause of~$\phi$.
We show that~$y_i$ is satisfies by~$\pi$. 
Since (i)~the strict reachability graph of~$\mg$ is equal to~$D=(V,A)$ and (ii)~$A$ contains the arc~$(s_i,t_i)$, there is a strict temporal path~$P$ from~$s_i$ to~$t_i$ in~$\mg$.
Hence, $P$ is a dense path in~$D$.
Consider the first internal vertex~$v$ of~$P$.
Recall that~$v$ is a solid neighbor of~$s_i$.
This implies that~$v$ is either the vertex corresponding to a literal of the clause~$y_i$, or a vertex~$q_{(i,j)}$ for some clause~$y_j$ distinct from~$y_i$ that shares a variable with~$y_i$.
Recall that~$q_{(i,j)}$ has only the two out-neighbors~$s_i$ and~$t_j\neq t_i$ in~$D$.
Hence,  the dense path~$P$ cannot contain~$q_{(i,j)}$, since~$(q_{(i,j)},t_i)\notin A$.
Thus, the first internal vertex of~$P$ is a literal~$\ell$ of clause~$y_i$.
Let~$x$ denote the variable for which~$\ell$ is the literal.

We show that the edge~$\{\ell,c_x\}$ receives at least one label under~$\lambda$.
This then implies that~$\pi$ satisfies~$y_i$, since by the initial observation, at most one of the edges~$\{x,c_x\}$ or~$\{\ol,c_x\}$ received a label under~$\lambda$.
First, note that~$P$ cannot contain another source vertex~$s_j\neq s_i$, since~$\ell$ has no arc towards any out-neighbor of~$s_j$ in~$D$.
Thus, $P$ continues by traversing to some solid neighbor of~$\ell$ in~$V_x$.
By construction, $A$ contains none of the arcs~$(a_{x,\ol},t_i)$, $(b_{\ell},t_i)$, $(s_i,q_{\ell})$, $(s_i,q_{a_\ell})$.
This implies that~$P$ contains none of the vertices of~$\{a_{x,\ol},b_\ell,q_{\ell},q_{a_\ell}\}$.
Consequently, $P$ contains no solid neighbor of~$a_\ell$ besides~$\ell$, which also implies that~$P$ does not contain~$a_\ell$.
Hence, by definition of~$D[V_x]$, the only solid neighbor of~$\ell$ in~$V_x$ that can possibly be contained in~$P$ is~$c_x$.
Consequently, $P$ traverses the edge~$\{\ell,c_x\}$, which implies that~$\{\ell,c_x\}$ receives at least one label under~$\lambda$.
This then implies that~$\pi$ satisfies clause~$y_i$.

$(\Rightarrow)$
Let~$\pi$ be a satisfying assignment of~$\phi$.
We define an undirected temporal graph~$(G=(V,E'),\lambda\colon E' \to \mathbb{N})$ for some~$E'\subseteq E$ as follows:
For each clause~$y_i$ of~$Y$, we set the labels of all edges incident with vertex~$s_i$ to~$\chi(y_i) + 1$ and the labels of all edges incident with vertex~$t_i$ to~$\chi(y_i) + 14$.
It remains to define the labels of the edges between vertices of~$V_x$ for each variable~$x\in X$.
These labels are depicted in~\Cref{fig:hardness triangle label}.
Formally, for each variable~$x\in X$, we set
\begin{itemize}
\item $\lambda(\{c_x,b_{x,\ol}\}) = 1$,
\item $\lambda(\{v,q_v\}) = 1$ for each~$v\in \{x,\ol,a_x, a_{\ol}\}$,
\item $\lambda(\{b_x,b_{x,\ol}\}) = 2$,
\item $\lambda(\{b_{\ol},b_{x,\ol}\}) = 3$,
\item $\lambda(\{x,a_{x}\}) = 12$,
\item $\lambda(\{a_{x},a_{x,\ol}\}) = \lambda(\{{\ol},a_{\ol}\}) = \lambda(\{{x},f_{x}\}) = \lambda(\{f_{x},f_{x,\ol}\}) = 13$,
\item $\lambda(\{a_{\ol},a_{x,\ol}\}) = \lambda(\{{x},b_{x}\}) = \lambda(\{{\ol},f_{\ol}\}) = \lambda(\{f_{\ol},f_{x,\ol}\}) = 14$,
\item $\lambda(\{\ol,b_{\ol}\}) = 15$, and 
\item $\lambda(\{c_x,a_{x,\ol}\}) = 25$.
\end{itemize}
If~$x$ is assigned to~\texttt{true} by~$\pi$, we set~$\lambda(\{x, c_x\})  = 13$ and $\lambda(\{\overline{x},c_x\}) = \emptyset$.
If~$x$ is assigned to~\texttt{false} by~$\pi$, we set~$\lambda(\{x, c_x\}) = \emptyset$ and $\lambda(\{\overline{x},c_x\}) = 14$.
This completes the definition of~$\lambda$.
Next, we show that the strict reachability graph of~$(G,\lambda)$ is exactly~$D$.

All solid edges that receive a label are realized through the label they receive. Any $\{x,c_x\}$ or $\{c_x,\overline{x}\}$ edge that does not receive a label is realized by the temporal paths on the smallest squares including the edge (one drawn above and one below it in~\Cref{fig:hardness triangle label}), and these are the only solid edges that might not receive a label.

As all edges incident to an $s_i$ receive the same label $\chi(y_i)+1$ and the same holds for all edges incident to a $t_i$, receiving $\chi(y_i)+14$ in this case, the temporal paths of~$(G,\lambda)$ only use $s_i$ or $t_j$ vertices as endpoints. 
It follows that none of the arcs that can only be realized by a path using $s_i$ or $t_i$ as an internal vertex, all of which are missing from $A$, is incorrectly realized. This includes all missing arcs incident to a $q_{(i,j)}$ vertex. It also follows that the realization of any (present or missing) arc incident to at least one $V_x$ vertex can be verified by inspection of the solid graph and matrix of~\Cref{fig:hardness triangle label}. It remains to argue about the arcs whose endpoints are both $s_i$ or $t_j$ vertices.

Every $(s_i,t_j)$ arc for $i\neq j$ is realized by the temporal path $s_i,q_{(i,j)},t_j$ using labels $\chi(y_i)+1$, $\chi(y_j)+14$, where $\chi(y_i)+1\leq 11<\chi(y_j)+14$. Every $(s_i,s_j)$ arc for $\chi(y_i) < \chi(y_j)$, present when $y_i$ and $y_j$ share a literal~$\ell$, is realized by the temporal path $s_i,\ell,t_j$ using labels $\chi(y_i)+1$, $\chi(y_j)+1$. Every $(t_i,t_j)$ arc for $\chi(y_i) < \chi(y_j)$, present when $y_i$ and $y_j$ share a variable~$x$, is realized by the temporal path $t_i,c_x,t_j$ using labels $\chi(y_i)+14$, $\chi(y_j)+14$. 

Every $(s_i,t_i)$ arc is realized as follows. Let $\ell$ be a literal of $y_i$ which evaluates to~\texttt{true} under $\pi$, corresponding to a variable $x$. As $x$ is a variable of $y_i$, $s_i$ and $t_i$ are both connected to $V_x$. If $x$ appears positive in $\ell$, then $s_i,x,c_x,t_i$ using labels $\chi(y_i)+1,13,\chi(y_i)+14$ is a realizing temporal path. If $x$ appears negated in $\ell$, then $s_i,\overline{x},c_x,t_i$ using labels $\chi(y_i)+1,14,\chi(y_i)+14$ is a realizing temporal path.

For clauses~$y_i$ and~$y_j$ that share no variable, $(s_i,s_j)$, $(t_i,t_j)$, and $(t_i,s_j)$ is not in $A$ and is not realized by~$(G,\lambda)$, because each path between these vertices in~$G$ goes trough at least one other source or terminal vertex, which by the initial argument cannot be a temporal path.
The same holds for all clauses~$y_i$ and~$y_j$ that share a variable but not a literal with respect to the non-arcs~$(s_i,s_j)$ and~$(s_j,s_i)$. 
\end{proof}

\ifold
By reducing from~\textsc{SAT} and slightly adapting the above reduction, we can also show the following.

\todomi{this should also follow from the W2-hardness. so maybe}
\begin{corollary}\label{no kernel}
Neither of~\any\str\URGD and~\simp\str\URGD admits a polynomial kernel when parameterized by the vertex cover number of the solid graph, unless~\bth.
\end{corollary}
\begin{proof}[Proof Sketch]
We present a polynomial parameter transformation from~\textsc{SAT} when parameterized by the number of variables, where each clause has arbitrary size and each variable may occur in arbitrary many clauses.
The reduction is similar to the previous one.
Intuitively, the only thing that changes is that all terminal vertices are merged into a single vertex~$\top$ and all connector vertices between a source and a terminal vertex are removes.

Formally, let~$\phi$ be an instance of~\textsc{SAT} with variables~$X$ and clauses~$Y$.
We start by a graph~$D$ having only a single vertex~$\top$.
For each variable~$x\in X$, add the variable gadget~$V_x$ as in the last reduction.
For each clause~$y_i\in Y$, we also again add a source vertex~$s_i$ to~$D$.
The arcs of~$D$ are as follows: 
\begin{itemize}
\item there is no arc between any two source vertices,
\item there is an arc~$(s_i,\top)$ in~$D$ for each clause~$y_i\in Y$,
\item there arcs between source vertices and vertices of any variable gadget are identical as in the last reduction (see \Cref{fig:hardness triangle label}), and
\item there arcs between $\top$ and any variable gadget~$V_x$ is identical to the arcs between the variable gadget~$V_x$ and any terminal vertex~$t_i$ for which clause~$y_i$ contains a literal of~$x$ in the last reduction (see \Cref{fig:hardness triangle label}), that is, $q_x,q_{\ol},q_{a_x},q_{a_{\ol}}$, and~$c_x$ are the solid neighbors of~$\top$ in~$V_x$, and~$\{(\top, a_{x,\ol}),(x, \top),(\ol, \top),(a_x, \top),(a_{\ol}, \top)\}$ are the dashed arc between~$\top$ and the vertices of~$V_x$.
\end{itemize}
Finally, for each source vertex~$s_i$, we add the dashed arc~$(s_i,\top)$ to~$D$.
Note that the vertex cover number of~$G$ (and also of~$D$) is~$\Oh(|X|)$, since there is no arc between any two source vertices and thus $\{\top\} \cup \bigcup_{x\in X} V_x$ is a vertex cover of size~$1 +  16 \cdot |X| \in \Oh(|X|)$ for~$D$.
Since \textsc{SAT} does not admit a polynomial kernel when parameterized by~$|X|$, unless~\bth~\cite{FS11}, to show that statement, it remains to show that~$D$ is realizable if and only if~$\phi$ is satisfiable.
That is, it remains to show that the reduction is correct.

Due to the similarity with the previous reduction, we only sketch the correctness proof of this reduction.

Since the induced subgraph on~$V_x$ is identical to the one in the previous reduction, one can again show that for each realization of~$D$ at most one of~$\{x,c_x\}$ and~$\{\ol,c_x\}$ receives a label.
One can then (again, similar to the previous proof) show that these edges then encodes a satisfying truth assignment for~$\phi$.

If~$\phi$ is satisfiable, then the realizing~$D$ can also be obtained similar to the one in the previous reduction.
Assign the labels according to~\Cref{fig:hardness triangle label} to the edges in~$G[V_x]$.
For each source vertex, label all incident edges with label~$2$ and label all edges incident with~$\top$ with label~$15$.
Similar to the previous proof, this labeling thus realizes~$D$.
\end{proof}
\fi

By slightly modifying the previous reduction, we can show that also all other version of~\URGD under consideration are \NP-hard.

\begin{theorem}\label{hardness proper}
Each version of~\URGD under consideration is \NP-hard on directed graphs of constant maximum degree.
\end{theorem}
\newcommand{\fighappy}{
\begin{figure}

\centering
\begin{tikzpicture}[scale=1.2]
\tikzstyle{k}=[circle,fill=white,draw=black,minimum size=10pt,inner sep=2pt]

\node[k] (cl) at (-1,2) {$s_i$};
\node[k] (cr) at (9,2) {$s_j$};

\node[k] (bx) at (2,0) {$b_x$};
\node[k] (bxnx) at (4,0) {$b_{x,\overline{x}}$};
\node[k] (bnx) at (6,0) {$b_{\overline{x}}$};

\node[k] (ax) at (2,4) {$a_x$};
\node[k] (axnx) at (4,4) {$a_{x,\overline{x}}$};
\node[k] (anx) at (6,4) {$a_{\overline{x}}$};

\node[k] (qx) at (1,5.5) {$q_x$};
\node[k] (qnx) at ($(qx) + (6,0)$) {$q_{\overline{x}}$};
\node[k] (qax) at ($(qx) + (1,0)$) {$q_{a_x}$};
\node[k] (qanx) at ($(qx) + (5,0)$) {$q_{a_{\overline{x}}}$};

\node[k] (cx) at (4,2) {$c_x$};
\node[k] (x) at (2,2) {$x$};
\node[k] (nx) at (6,2) {$\overline{x}$};

\node (t1) at (2.1,8) {};
\node (t2) at (5.9,8) {};

\node[] (t11) at ($(t1)+(-1.1,0)$) {\textbf{\dots}};
\node[k] (t11) at ($(t1)+(-.6,0)$) {};
\node[k] (t12) at ($(t11)+(.45,0)$) {};
\node[k] (t13) at ($(t12)+(.45,0)$) {};
\node[k] (t14) at ($(t13)+(.45,0)$) {};
\node[k] (t15) at ($(t14)+(.45,0)$) {};

\node[] (t21) at ($(t2)-(-1.1,0)$) {\textbf{\dots}};
\node[k] (t21) at ($(t2)-(-.6,0)$) {};
\node[k] (t22) at ($(t21)-(.45,0)$) {};
\node[k] (t23) at ($(t22)-(.45,0)$) {};
\node[k] (t24) at ($(t23)-(.45,0)$) {};
\node[k] (t25) at ($(t24)-(.45,0)$) {};

\draw[thick] ($(t1)$) ellipse (1.7cm and .4cm);
\node at ($(t1) + (0,.7)$) {$T_k$};
\draw[thick] ($(t2)$) ellipse (1.7cm and .4cm);
\node at ($(t2) + (0,.7)$) {$T_\ell$};

\draw[thick] ($(cl)$) ellipse (.4cm and 1.3cm);
\node at ($(cl) + (0,.75)$) {\rotatebox{90}{\textbf{\dots}}}; 
\node at ($(cl) + (0,-.75)$) {\rotatebox{90}{\textbf{\dots}}}; 
\node at ($(cl) + (0,1.7)$) {$S_i$};

\draw[thick] ($(cr)$) ellipse (.4cm and 1.3cm);
\node at ($(cr) + (0,.75)$) {\rotatebox{90}{\textbf{\dots}}}; 
\node at ($(cr) + (0,-.75)$) {\rotatebox{90}{\textbf{\dots}}}; 
\node at ($(cr) + (0,1.7)$) {$S_j$};

\draw[draw=black, dashed] ($(qx) + (-.5,.5)$) rectangle ++(7,-6.5);

\draw[thick] (x) -- (qx) node [midway, fill=white] {$1$};
\draw[thick] (nx) -- (qnx) node [midway, fill=white] {$1$};
\draw[thick] (ax) -- (qax) node [midway, fill=white] {1};
\draw[thick] (anx) -- (qanx) node [midway, fill=white] {1};

\draw[thick] (cx) -- (bxnx) node [midway, fill=white] {1};

\draw[thick] (bx) -- (bxnx) node [midway, fill=white] {2};
\draw[thick] (bnx) -- (bxnx) node [midway, fill=white] {3};

\draw[thick] (x) -- (ax) node [midway, fill=white] {12};
\draw[thick] (nx) -- (anx) node [midway, fill=white] {13};
\draw[thick] (axnx) -- (ax) node [midway, fill=white] {13};
\draw[thick] (axnx) -- (anx) node [midway, fill=white] {14};
\draw[thick] (axnx) -- (cx) node [midway, fill=white] {25};

\draw[thick,red] (x) -- (cx) node [midway, fill=white] {13};
\draw[thick,blue] (nx) -- (cx) node [midway, fill=white] {14};

\draw[thick] (cx) edge[bend left=10] node[very near end, fill=white]{$\alpha$} (t13);
\draw[thick] (cx) edge[bend right=10] node[very near end, fill=white]{$\beta$} (t23);

\draw[thick] (x) -- (bx) node [midway, fill=white] {14};
\draw[thick] (nx) -- (bnx) node [midway, fill=white] {15};

\draw[thick] (cl) -- (x) node [midway, fill=white] {$\in [2,11]$};
\draw[thick] (cr) -- (nx) node [midway, fill=white] {$\in [2,11]$};

\foreach \x/\y in {qx/1,qnx/5,qax/2,qanx/4} {
    \draw[thick] (\x) -- (t1\y) node [midway, fill=white] {$\alpha$};               
}

\foreach \x/\y in {qx/5,qnx/1,qax/4,qanx/2} {
    \draw[thick] (\x) -- (t2\y) node [midway, fill=white] {$\beta$};    
}

\end{tikzpicture}

~
~

~
~

~
~

~
~\scalebox{.8}{
\begin{tabular}{c|||c|c|c|c||c|c|c|c|c|c|c|c|c|c|c|c|c} & $\in S_i$  & $\in S_j$  & $\in T_k$  & $\in T_\ell$  & $x$  & ${\overline{x}}$  & $a_x$  & $a_{\overline{x}}$  & $a_{x,\overline{x}}$  & $b_x$  & $b_{\overline{x}}$  & $b_{x,\overline{x}}$  & $c_x$    & $q_x$  & $q_{\overline{x}}$  & $q_{a_x}$  & $q_{a_{\overline{x}}}$  \\\hline\hline\hline
 $\in S_i$  &   & 0 & \cellcolor{red!25} 1 & \cellcolor{red!25} 1 & 1 & 0 & \cellcolor{black!25} 1 & \cellcolor{black!25} 1 & \cellcolor{black!25} 1 & \cellcolor{black!25} 1 & 0 & 0 & \cellcolor{black!25} 1 & 0 & 0 & 0 & 0 \\\hline
 $\in S_j$  & 0 &   & \cellcolor{blue!25} 1 & \cellcolor{blue!25} 1 & 0 & 1 & 0 & \cellcolor{black!25} 1 & \cellcolor{black!25} 1 & 0 & \cellcolor{black!25} 1 & 0 & \cellcolor{black!25} 1 & 0 & 0 & 0 & 0 \\\hline
 $\in T_k$  & 0 & 0 &   & \cellcolor{black!25} 1 & 0 & 0 & 0 & 0 & \cellcolor{black!25} 1 & 0 & 0 & 0 & 1 & 1 & 1 & 1 & 1 \\\hline
 $\in T_\ell$  & 0 & 0 & 0 &   & 0 & 0 & 0 & 0 & \cellcolor{black!25} 1 & 0 & 0 & 0 & 1 & 1 & 1 & 1 & 1 \\\hline\hline
 $x$  & 1 & 0 & \cellcolor{black!25} 1 & \cellcolor{black!25} 1 &   & 0 & 1 & \cellcolor{black!25} 1 & \cellcolor{black!25} 1 & 1 & 0 & 0 & 1 & 1 & 0 & 0 & 0 \\\hline
 ${\overline{x}}$  & 0 & 1 & \cellcolor{black!25} 1 & \cellcolor{black!25} 1 & 0 &   & 0 & 1 & \cellcolor{black!25} 1 & 0 & 1 & 0 & 1 & 0 & 1 & 0 & 0 \\\hline
 $a_x$  & 0 & 0 & \cellcolor{black!25} 1 & \cellcolor{black!25} 1 & 1 & 0 &   & \cellcolor{black!25} 1 & 1 & \cellcolor{black!25} 1 & 0 & 0 & \cellcolor{black!25} 1 & 0 & 0 & 1 & 0 \\\hline
 $a_{\overline{x}}$  & 0 & 0 & \cellcolor{black!25} 1 & \cellcolor{black!25} 1 & 0 & 1 & 0 &   & 1 & 0 & \cellcolor{black!25} 1 & 0 & \cellcolor{black!25} 1 & 0 & 0 & 0 & 1 \\\hline
 $a_{x,\overline{x}}$  & 0 & 0 & 0 & 0 & 0 & 0 & 1 & 1 &   & 0 & 0 & 0 & 1 & 0 & 0 & 0 & 0 \\\hline
 $b_x$  & 0 & 0 & 0 & 0 & 1 & \cellcolor{black!25} 1 & 0 & 0 & 0 &   & \cellcolor{black!25} 1 & 1 & 0 & 0 & 0 & 0 & 0 \\\hline
 $b_{\overline{x}}$  & 0 & 0 & 0 & 0 & 0 & 1 & 0 & 0 & 0 & 0 &   & 1 & 0 & 0 & 0 & 0 & 0 \\\hline
 $b_{x,\overline{x}}$  & 0 & 0 & \cellcolor{black!25} 1 & \cellcolor{black!25} 1 & \cellcolor{black!25} 1 & \cellcolor{black!25} 1 & 0 & 0 & \cellcolor{black!25} 1 & 1 & 1 &   & 1 & 0 & 0 & 0 & 0 \\\hline
 $c_x$  & 0 & 0 & 1 & 1 & 1 & 1 & 0 & 0 & 1 & \cellcolor{black!25} 1 & \cellcolor{black!25} 1 & 1 &   & 0 & 0 & 0 & 0 \\\hline
 $q_x$  & \cellcolor{black!25} 1 & 0 & 1 & 1 & 1 & 0 & \cellcolor{black!25} 1 & \cellcolor{black!25} 1 & \cellcolor{black!25} 1 & \cellcolor{black!25} 1 & 0 & 0 & \cellcolor{black!25} 1 &   & 0 & 0 & 0 \\\hline
 $q_{\overline{x}}$  & 0 & \cellcolor{black!25} 1 & 1 & 1 & 0 & 1 & 0 & \cellcolor{black!25} 1 & \cellcolor{black!25} 1 & 0 & \cellcolor{black!25} 1 & 0 & \cellcolor{black!25} 1 & 0 &   & 0 & 0 \\\hline
 $q_{a_x}$  & 0 & 0 & 1 & 1 & \cellcolor{black!25} 1 & 0 & 1 & \cellcolor{black!25} 1 & \cellcolor{black!25} 1 & \cellcolor{black!25} 1 & 0 & 0 & \cellcolor{black!25} 1 & 0 & 0 &   & 0 \\\hline
 $q_{a_{\overline{x}}}$  & 0 & 0 & 1 & 1 & 0 & \cellcolor{black!25} 1 & 0 & 1 & \cellcolor{black!25} 1 & 0 & \cellcolor{black!25} 1 & 0 & \cellcolor{black!25} 1 & 0 & 0 & 0 &   \\\hline
 \end{tabular}}
\caption{Figure and matrix for the variable gadget in the hardness reduction for the happy version.
Only solid edges that receive a label are depicted.}
\label{fig:hardness happy}
\end{figure}
}
\begin{proof}
Due to~\Cref{hardness trianglefree}, it remains to show the \NP-hardness only for (i)~\pro\str\URGD, (ii)~\happy\str\URGD, and (iii)~each version of~\nstr\DRGD. 
We capture all of these versions in one reduction.

We again reduce from \SAT by slightly adapting the reduction from the proof of~\Cref{hardness trianglefree}.
Let~$\phi$ be an instance of~\SAT with variables~$X$ and clauses~$Y := \{y_1, \dots, y_{|Y|}\}$.
Moreover, let~$D'$ be the instance of our problem that was constructed in the reduction of~\Cref{hardness trianglefree} for instance~$\phi$, and let~$G'$ be the solid graph of~$D'$.
We obtain a directed graph~$D$ as follows:
We initialize~$D$ as~$D'$.
First, we remove the vertices~$f_x, f_{\ol}$, and~$f_{x,\ol}$ for each variable~$x\in X$.
Next, for each clause~$y_i\in Y$, we replace the vertex~$s_i$ ($t_i$) by a solid clique~$S_i$ ($T_i$).
We initialize~$S_i := \{s_i\}$ ($T_i := \{t_i\}$).
Afterwards, we add a twin of~$s_i$ ($t_i$) to~$S_i$ ($T_i$), that is, a vertex with the same in- and out-neighborhood as~$s_i$ ($t_i$).
Note that each vertex of~$S_i$ ($T_i$) has the same solid neighbors as each other vertex of~$S_i$ ($T_i$). 
We continue this process, until~$S_i$ ($T_i$) has an odd size of at least~$\max(9, r_{i})$, where~$r_i$ is the degree of~$s_i$ ($t_i$) in~$G'$.
Finally, we make~$S_i$ ($T_i$) a clique in~$G$.
This completes the construction of~$D$.
\fighappy

Since the maximum degree of~$D'$ was a constant and we only added a constant number of twins for each source and terminal vertex, the maximum degree of~$D$ is only increased by a constant factor.

We now show the correctness of the reduction.
To this end, we show that
\begin{itemize}
\item if~$\phi$ is satisfiable, then there is a happy undirected temporal graph with reachability graph~$D$, and
\item if~$\phi$ is not satisfiable, then there is neither a proper undirected temporal graph with reachability graph~$D$ nor an undirected temporal graph with non-strict reachability graph~$D$.
\end{itemize}
This then implies that the reduction is correct for all stated versions of our problem.

$(\Leftarrow)$
Suppose that~$\phi$ is satisfiable and let~$\lambda'$ be the realization for~$D'$ described in the proof of~\Cref{hardness trianglefree}.
Note the only adjacent edge that shared a label under~$\lambda'$ (i)~contained an edge incident with a vertex~$f_\ell$ for some~$\ell \in \{x,\ol\mid x\in X\}$ or (ii)~are both incident with the same source or incident the same terminal vertex.
Since~$D$ does not contain any of the~$f_\ell$-vertices, the only adjacent edge of~$G$ under~$\lambda'$ are both incident with the same source or incident the same terminal vertex.
Moreover, note that no arc of both~$D$ and~$D'$ was realized via temporal paths that go over any~$f_\ell$-vertex.  
To obtain a happy realization for~$D$, we essentially spread all the labeled edges incident with a source~$s_i$ (a terminal~$t_i$) to an individual vertex of~$S_i$ ($T_i$) each. 
By definition of~$S_i$ ($T_i$), this vertex set contain at least as many vertices each as the number of solid edges incident with~$s_i$ ($t_i$) in~$G'$.
Formally, we initialize a labeling~$\lambda$ of some edges of~$G$ as the restriction of~$\lambda'$ to the common edges of~$G$ and~$G'$.
Afterwards we iterate over all sources and terminals.
For each source~$s_i$, we take an arbitrary maximal matching~$M$ between the solid neighbors of~$s_i$ in~$G'$ and~$S_i$. 
We remove the labels of all edges incident with~$s_i$ from~$\lambda$ (which all had the same label, say~$\alpha_{s_i}$) and add the label~$\alpha_{s_i}$ to all edges of~$M$.
We do the same analogously for all terminal vertices.
This ensures that~$\lambda$ is proper.

Finally, for each~$s_i$, we take four arbitrary and pairwise edge-disjoint Hamiltonian cycles~$C_j$ with~$j\in [1,4]$. 
Such Hamiltonian cycles exist~\cite{L1882}, since~$S_i$ has an odd size of at least~$9$.
We label the edges of these Hamiltonian cycles as follows:
Fix an arbitrary vertex~$s_i^*\in S_i$ of~$S_i$ and label the edges of the Hamiltonian cycles with consecutive labels starting from~$s_i^*$, such that (i)~each label of~$C_1$ is smaller than each label of~$C_2$, (ii)~each label of~$C_2$ is smaller than~$\alpha_{s_i}$, (iii)~each label of~$C_3$ is larger than~$\alpha_{s_i}$, and (iv)~each label of~$C_4$ is larger than each label of~$C_3$.
Note that this can be done in a straightforward way but might result in possibly negative labels.
The latter is no issue, since we can arbitrarily shift all labels at the same time until they are positive integers, while preserving the same reachabilities.
Note that the labeling of the four Hamiltonian ensure that (a)~prior to time~$\alpha_{s_i}$, each vertex of~$S_i$ can reach each other vertex of~$S_i$ and (b)~after time~$\alpha_{s_i}$, each vertex of~$S_i$ can reach each other vertex of~$S_i$.
This in particular implies that all arcs of the clique~$G[S_i]$ are realized.
Moreover, it implies that all vertices of~$S_i$ have the same in- and out-reachabilities, since only edges with label~$\alpha_{s_i}$ are incident with both a vertex of~$S_i$ and any other vertex outside of~$S_i$.

We do the same for the terminal vertices and their respective cliques.

The resulting labeling is then proper and assigns at most one label to each edge of~$G$.
Moreover, $\lambda$ realizes~$D$, since, intuitively we can merge each source clique~$S_i$ into a single vertex~$s_i$ and merge each terminal clique~$T_i$ into a single vertex~$t_i$.
The resulting graph~$D''$ is then the graph obtained for~$D'$ by removing the vertices of~$\{f_x,f_{\ol},f_{x,\ol}\mid x\in X\}$, for which~$\lambda'$ is a realization, when restricted to the solid edges of~$D''$.   
Hence, there is a happy realization for~$D$.

$(\Rightarrow)$
We show this statement via contraposition. 
Let~$\mg =(G = (V,E), \lambda \colon E \to 2^\mathbb{N})$ be an undirected temporal graph such that (i)~$\lambda$ is a proper labeling and~$D$ is the reachability graph of~$\mg$, or (ii)~$D$ is the non-strict reachability graph of~$\mg$. 
We show that~$\phi$ is satisfiable.
Similar to the proof of~\Cref{hardness trianglefree}, we show that for each variable~$x\in X$, at most one of the edges~$\{x,c_x\}$ and~$\{\ol,c_x\}$ receives a label under~$\lambda$.

Let~$x$ be a variable of~$X$.
Recall that~$D$ contains neither of the arcs~$(x,\ol)$ nor~$(\ol,x)$. 
Assume towards a contradiction that $\lambda(\{x,c_x\}) \neq \emptyset$ and $\lambda(\{\ol,c_x\}) \neq \emptyset$.
If~$\min \lambda(\{x,c_x\}) < \max \lambda(\{\ol,c_x\})$ or $\max \lambda(\{x,c_x\}) > \min \lambda(\{\ol,c_x\})$, the strict and the non-strict reachability graph of~$\mg$ contain at least one arc between~$x$ and~$\ol$; a contradiction to the fact that~$\lambda$ realizes~$D$.
Hence, assume that this is not the case.
That is~$\lambda(\{x,c_x\}) = \lambda(\{\ol,c_x\}) \neq  \emptyset$.
Note that this is not possible in the case where~$\lambda$ is a proper labeling.
Hence, we only have to consider the case where the non-strict reachability graph of~$\mg$ is~$D$.
Since~$\lambda(\{x,c_x\}) = \lambda(\{\ol,c_x\}) \neq  \emptyset$, the non-strict reachability graph of~$D$ contains both arcs~$(x,\ol)$ and~$(\ol,x)$; a contradiction.
Thus, not both edges~$\{x,c_x\}$ and~$\{\ol,c_x\}$ receives a label under~$\lambda$.

The remainder of the proof that~$\phi$ is satisfiable is identical to the one in~\Cref{hardness trianglefree} and thus omitted.
\end{proof}

In all the above reductions, the number of vertices and arcs in the constructed graph~$D$ were~$\Oh(|\phi|)$, where~$\phi$ is the size of the respective instance of~\SAT.
Since~\SAT cannot be solved in $2^{o(|\phi|)} \cdot |\phi|^{\Oh(1)}$~time, unless the ETH fails~\cite{T84}, this implies the following.

\begin{corollary}
No version of~\URGD under consideration can be solved in $2^{o(|V| + |A|)} \cdot n^{\Oh(1)}$~time, unless the ETH fails.
\end{corollary}
\fi

Hence, the running time of~\Cref{exact algos} can presumably not be improved significantly.

\subsection{Parameterized hardness}
We now strengthen our hardness result for~\any\str\URGD and~\simp\str\URGD, which will highly motivate the analysis of parameterized algorithms for the parameter `feedback edge set number' of the solid graph, which we consider in~\Cref{sec:fes}.

\begin{theorem}\label{param hardness}
\any\str\URGD and~\simp\str\URGD are W[2]-hard when parameterized by the feedback vertex set number and the treedepth of the solid graph.
\end{theorem}

\begin{proof}
We reduce from \SC which is W[2]-hard when parameterized by~$k$~\cite{DF13}.

\prob{\SC}{A universe~$U$, a collection~$\mf$ of subsets of~$U$ (called hyperedges), and an integer~$k$.}{Is there a subset~$S\subseteq \mf$ of size at most~$k$, such that each element~$u\in U$ is contained in at least one hyperedge of~$S$?}

Let~$I:=(U,\mf,k)$ be an instance of~\SC with~$\mf = \{F_1, \dots, F_r\}$.
Assume without loss of generality that each element of~$U$ is contained in at least one hyperedge and that~$\mf$ has size at least~$k$, as otherwise~$I$ could be solved trivially.
We obtain a directed graph~$D=(V,A)$ with solid graph~$G=(V,E)$ as follows:
The graphs contain the vertex~$\top$, each element~$u\in U$ as a vertex, and for each~$i\in [1,k]$ the vertices of~$\{a_i,a_i',b_i,b_i',c_i,\bot_i\}$.
Additionally, for each~$i\in [1,k]$, each~$F\in \mf$, and each~$u\in F$, the graphs contain the vertices~$w_i^{F,u}$, $v_i^{F,u}$, and~$q_i^{F,u}$.

Next, we describe the solid edges of~$G$, that is, the bidirectional arcs of~$D$.
\nnew{See~\Cref{tab:whardness} for the solid edges and dashed arcs of the main connection gadget of the instance.}

\begin{table}
\caption{A part of the adjacency matrix of the W[2]-hardness reduction.}
\label{tab:whardness}
\begin{center}
\scalebox{.75}{
\begin{tabular}{c|||c|c|c|c|c||c|c|c|c|c|c}
   & $u$  & $\top$  & $v_i^{F,u}$  & $w_i^{F,u}$  & $q_i^{F,u}$  & $a_i$  & $a'_i$  & $b_i$  & $b'_i$  & $c_i$  & $\bot_i$  \\\hline\hline\hline
 $u$  &   & \cellcolor{black!25} 1 & \cellcolor{black!25} 1 & 1 & 0 & \cellcolor{black!25} 1 & 0 & \cellcolor{black!25} 1 & \cellcolor{black!25} 1 & \cellcolor{black!25} 1 & 0 \\\hline
 $\top$  & 0 &   & 0 & 0 & 1 & 0 & 0 & 0 & \cellcolor{black!25} 1 & 1 & 0 \\\hline
 $v_i^{F,u}$  & 0 & 0 &   & 1 & 0 & \cellcolor{black!25} 1 & 0 & 0 & 0 & 0 & 1 \\\hline
 $w_i^{F,u}$  & 1 & \cellcolor{black!25} 1 & 1 &   & 1 & 1 & 0 & 1 & \cellcolor{black!25} 1 & 1 & 0 \\\hline
 $q_i^{F,u}$  & 0 & 1 & \cellcolor{black!25} 1 & 1 &   & \cellcolor{black!25} 1 & 0 & \cellcolor{black!25} 1 & \cellcolor{black!25} 1 & \cellcolor{black!25} 1 & 0 \\\hline\hline
 $a_i$  & 0 & 0 & 0 & 1 & 0 &   & 1 & 0 & 0 & 0 & 0 \\\hline
 $a'_i$  & 0 & \cellcolor{black!25} 1 & 0 & \cellcolor{black!25} 1 & 0 & 1 &   & 0 & \cellcolor{black!25} 1 & 1 & 0 \\\hline
 $b_i$  & 0 & \cellcolor{black!25} 1 & \cellcolor{black!25} 1 & 1 & 0 & \cellcolor{black!25} 1 & 0 &   & 1 & \cellcolor{black!25} 1 & 0 \\\hline
 $b'_i$  & 0 & 0 & 0 & 0 & 0 & 0 & 0 & 1 &   & 1 & 0 \\\hline
 $c_i$  & 0 & 1 & 0 & 1 & 0 & \cellcolor{black!25} 1 & 1 & 0 & 1 &   & 0 \\\hline
 $\bot_i$  & 0 & 0 & 1 & 0 & 0 & 0 & 0 & 0 & 0 & 0 &
\end{tabular}
}
\end{center}
\end{table}

Let~$i\in [1,k]$.
The graph $G$ contains the edges~$\{a_i,a_i'\}$, $\{a_i',c_i\}$, $\{c_i,b_i'\}$, $\{b_i',b_i\}$, and $\{c_i,\top\}$.
Moreover, for each~$F\in \mf$ and each~$u\in F$, $G$ contains the edges~$\{u, w_i^{F,u}\}$, $\{w_i^{F,u}, a_i\}$, $\{w_i^{F,u}, b_i\}$, $\{w_i^{F,u}, c_i\}$, $\{w_i^{F,u}, v_i^{F,u}\}$, $\{w_i^{F,u}, q_i^{F,u}\}$, $\{v_i^{F,u},\bot_i\}$, and~$\{q_i^{F,u}, \top\}$.

Finally, we describe the dashed arcs of~$D$.
For each~$i\in [1,k]$, $D$ contains the arcs~$(a'_i,b_i)$, $(b_i,a_i)$, $(b_i,c_i)$, $(c_i,a_i)$, $(a'_i,\top)$, $(b_i,\top)$, and~$(\top,b_i')$.
Let~$u\in U$.
Then, $D$ contains the arc~$(u,\top)$.
Moreover, for each~$i\in [1,k]$, $D$ contains the arcs~$(u,a_i)$, $(u,b_i)$, $(u,b_i')$, and~$(u,c_i)$.
Additionally, for each hyperedge~$F\in \mf$ with~$u\in F$, $D$ contains the arcs~$(u, v_i^{F,u})$, $(v_i^{F,u},a_i)$, $(b_i,v_i^{F,u})$, $(w_i^{F,u},\top)$, $(w_i^{F,u},b_i')$, $(a_i',w_i^{F,u})$, $(q_i^{F,u},v_i^{F,u})$, $(q_i^{F,u},a_i)$, $(q_i^{F,u},b_i)$, $(q_i^{F,u},b_i')$, and $(q_i^{F,u},c_i)$.
The only other dashed arcs in~$D$ are between~$v_i^{F,u}$-vertices.
Let~$F_x,F_y\in \mf$ with $x<y$ and let~$u_x\in F_x$ and~$u_y\in F_y$.
For each~$i\in [1,k]$, $D$ contains the arc~$(v_i^{F_x,u_x},v_i^{F_y,u_y})$.

This completes the construction of~$D$.
First, we show the parameter bounds.

\iflong
\begin{claim}
$G$ has a feedback vertex set of size~$\Oh(k)$ and~$G$ has treedepth~$\Oh(k)$.
\end{claim}
\fi
\iflong
\begin{claimproof}
Let~$X := \{\top\} \cup \{a_i,a_i',b_i,b_i',c_i,\bot_i \mid i\in [1,k]\}$.
We show that~$X$ is a feedback vertex set of~$G$ and more further, that each tree in~$G-X$ has depth at most~2.
Since~$|X| = 6k + 1$, this then proves the statement.
Let~$u\in U$.
The solid neighborhood of~$u$ in~$G$ is~$R = \{w_i^{F,u}\mid i\in [1,k], F\in \mf, u\in F\}$, which is an independent set in~$G$.
Now consider the other solid neighbors of the vertices of~$R$ in~$V\setminus X$.
Let~$w_i^{F,u}\in R$.
Then, the only two solid neighbors of~$w_i^{F,u}$ in~$V\setminus X$ besides~$u$ are~$v_i^{F,u}$ and~$q_i^{F,u}$, since~$a_i,b_i,$ and~$c_i$ are in~$X$.
Furthermore, $w_i^{F,u}$ is the only solid neighbor of both~$v_i^{F,u}$ and~$q_i^{F,u}$ in~$V\setminus X$. 
This implies that~$G[V_u]$ is a tree of depth two (rooted in~$u$), where~$V_u := \{w_i^{F,u},v_i^{F,u},q_i^{F,u} \mid i\in [1,k], F\in\mf, u\in F\} \cup \{u\}$.
Moreover, each vertex of~$V\setminus X$ is contained in exactly one set~$V_u$, since we assumed that each hyperedge is non-empty.

This implies that each connected component in~$G-X$ is a tree of depth at  most~$2$, which implies the statement.
\end{claimproof}
\fi

We show that~$I$ is a yes-instance of~\SC if and only if~$D$ is realizable via strict temporal paths.
More precisely, we show that if~$I$ is a yes-instance of~\SC, then there is a simple undirected temporal graph with strict reachability graph~$D$.
This then implies the correctness of the reduction for both~\any\str\URGD and~\simp\str\URGD.

\iflong

$(\Leftarrow)$
Let~$\mg = (G=(V,E),\lambda\colon E \to 2^{\mathbb{N}})$ be a undirected temporal graph with strict reachability graph equals to~$D$.
We show that there is a hitting set of size at most~$k$ for~$U$.

To this end, we first analyze for each~$i\in [1,k]$ the structure of the labeling~$\lambda$ with respect to the vertices~$\{a_i,a_i',b_i,b_i',c_i,\bot_i\} \cup \{v_i^{F,u},w_i^{F,u},q_i^{F,u}\mid F\in \mf, u\in F\}$.
Let~$i\in [1,k]$.
For a hyperedge~$F\in \mf$, we denote with~$F(i)$ the set of vertices~$\{w_i^{F,u}\mid u\in F\}$.
We show that there is at most one hyperedge~$F\in \mf$ for which edges between~$c_i$ and~$F(i)$ receive a non-empty set of labels under~$\lambda$.  

Assume towards a contradiction that there are two hyperedges~$F_x\in \mf$ and~$F_y\in \mf$ with~$x<y$, such that there is at least one edge~$e_x$ between~$F_x(i)$ and~$c_i$ receives a non-empty set of labels under~$\lambda$, and there is at least one edge~$e_y$ between~$F_y(i)$ and~$c_i$ receives a non-empty set of labels under~$\lambda$.
Let~$w_i^{F_x,u_x}$ ($w_i^{F_y,u_y}$) be the other endpoint of~$e_x$ ($e_y$) besides~$c_i$.
Moreover, let~$z\in \{x,y\}$.
By construction~$v_i^{F_z,u_z}$ has only two solid neighbors in~$G$, namely~$w_i^{F_x,u_x}$ and~$\bot_i$.
This implies that both solid edges incident with~$v_i^{F_z,u_z}$ each receive at least one label under~$\lambda$ (see~\Cref{forced label}).
Between these two neighbors, there is no dashed arc in~$D$ and no solid edge in~$G$.
Due to~\Cref{incident same label}, this implies that there is some~$\alpha_z\in \mathbb{N}$, such that~$\lambda(\{\bot_i,v_i^{F_z,u_z}\}) = \lambda(\{v_i^{F_z,u_z},w_i^{F_z,u_z}\}) = \{\alpha_z\}$.
Similarly, since~$e_z$ receives at least one label under~$\lambda$ and there is no arc between~$c_i$ and~$v_i^{F_z,u_z}$, $\lambda(e_z) = \lambda(\{v_i^{F_z,u_z},w_i^{F_z,u_z}\}) = \{\alpha_z\}$.
We now show that~$\alpha_x < \alpha_y$.
This is due to the fact that~$(v_i^{F_x,u_x},\bot_i,v_i^{F_y,u_y})$ is the only dense~$(v_i^{F_x,u_x},v_i^{F_y,u_y})$-path in~$D$ and~$(v_i^{F_x,u_x},v_i^{F_y,u_y})$ is an arc of~$D$.
Concluding, $\lambda(e_x) = \{\alpha_x\}$ and~$\lambda(e_y) = \{\alpha_y\}$, which implies that there is a temporal~$(w_i^{F_x,u_x},w_i^{F_y,u_y})$-path under~$\lambda$.
This contradicts the fact that~$\lambda$ realizes~$D$, since~$(w_i^{F_x,u_x},w_i^{F_y,u_y})$ is not an arc of~$D$.
Consequently, there is at most one hyperedge~$F\in \mf$ for which edges between~$c_i$ and~$F(i)$ receive a non-empty set of labels under~$\lambda$. 
 
For each~$i\in [1,k]$, let~$F_i'$ denote the unique hyperedge~$F\in \mf$ for which edges between~$c_i$ and~$F(i)$ receive a non-empty set of labels under~$\lambda$, if such a hyperedge~$F$ exists.
We set~$S:= \{F_i'\mid i\in [1,k], F_i'~\text{exists}\}$ and show that~$S$ is a set cover for~$U$.
Let~$u\in U$.
We show that there is a hyperedge in~$S$ that contains~$u$.
Since~$\lambda$ realizes~$D$ and~$(u,\top)$ is an arc of~$D$, there is a temporal~$(u,\top)$-path~$P$ under~$\lambda$.
Recall that~$P$ is thus a dense path in~$D$.
Since~$\{u,\top\}$ is not a solid edge of~$G$, $P$ has length at least~$3$.
Let~$x$ be the first internal vertex of~$P$, which is a solid neighbor of~$u$ in~$G$.
By definition, this implies that~$x= w_i^{F,u}$ for some~$i\in [1,k]$ and some~$F\in \mf$ with~$u\in F$.
Consider the solid neighbors of~$x$ in~$G$ besides~$u$.
These are the vertices~$a_i,b_i,c_i,v_i^{F,u},$ and~$q_i^{F,u}$.
Since~$P$ is a dense~$(u,\top)$-path and~$D$ contains none of the arcs~$(a_i,\top)$, $(v_i^{F,u},\top)$, or~$(u,v_i^{F,u})$, $P$ visits none of the vertices~$a_i,v_i^{F,u},$ or~$q_i^{F,u}$.
Hence, $P$ continues from~$w_i^{F,u}$ to either~$b_i$ or~$c_i$.
Consider the solid neighbors of~$b_i$ besides~$w_i^{F,u}$.
These are the vertices of~$R:=\{b_i'\} \cup (\{w_i^{F',u'}\mid F'\in \mf, u'\in F'\} \setminus \{w_i^{F,u}\})$.
Since~$D$ does not contain the arc~$(b_i',\top)$ and contains none of the arcs~$(w_i^{F,u},r)$ with~$r\in R\setminus \{b_i'\}$, the dense path~$P$ visits no solid neighbor of~$b_i$ besides~$w_i^{F,u}$.
This also implies that~$P$ does not visit~$b_i$.
Thus, the only solid neighbor of~$w_i^{F,u}$ besides~$u$ in~$P$ is~$c_i$.
Consequently, $P$ traverses the edge~$e=\{w_i^{F,u},c_i\}$.
Thus, $e$ receives at least one label under~$\lambda$, which implies that~$F = F_i' \in S$, that is, at least one hyperedge of~$S$ contains~$u$.
This implies that~$S$ is a set cover of~$U$.  

$(\Rightarrow)$
Let~$S$ be a set cover of size at most~$k$ for~$U$.
Without loss of generality assume that~$S$ has size exactly~$k$ and let the hyperedges in~$S$ be denoted as~$F_i'$ for~$i\in [1,k]$.
We show that there is a labeling~$\lambda \colon E\to \mathbb{N} \cup \emptyset$ that realizes~$D$.
This then implies that there is a simple temporal graph with strict reachability graph equals to~$D$.

For each hyperedge~$F_x\in\mf$, let~$\alpha_{F_x} := x+2$.
Moreover, recall that~$r :=|\mf|$.
Hence, $\alpha_x\in [3,r+2]$.
We define~$\lambda$ as follows:
\begin{itemize}
\item For each~$u\in U$, we assign label~$1$ to all edges incident with~$u$.
\item We assign label~$r+3$ to all edges incident with~$\top$.
\item For each~$i\in [1,k]$, we set~$\lambda(\{a_i,a_i'\}) := \lambda(\{b_i,b_i'\}) := r+4$, $\lambda(\{a_i',c_i\}) := 2$, and~$\lambda(\{c_i,b_i'\}) := r+5$.
\item For each~$F\in \mf$, each~$u\in U$, and each~$i\in [1,k]$, we set~$\lambda(\{w_i^{F,u},v_i^{F,u}\}) := \lambda(\{v_i^{F,u},\bot_i\}) := \alpha_F$, $\lambda(\{w_i^{F,u},q_i^{F,u}\}) := 1$, $\lambda(\{w_i^{F,u},a_i\}) := r+5$, and~$\lambda(\{w_i^{F,u},b_i\}) := 2$. 
\end{itemize}
Finally, for each~$i\in [1,k]$ and each~$u\in F_i'$ (that is, the~$i$th hyperedge of~$S$), we set~$\lambda(\{w_i^{F,u},c_i\}) := \alpha_{F_i'}$.
All other edges of~$E$ receive the empty set under~$\lambda$.

We now show that~$\lambda$ realizes~$D$.
To this end, we first show that all solid edges are realized.
Note that only some solid edges between~$c_i$ and vertices of~$\{w_i^{F,u}\mid i\in [1,k], F\in \mf, u\in F\}$ did not receive a label under~$\lambda$.
Let~$x=w_i^{F,u}$.
Then, there is the temporal paths~$(x,b_i,b_i',c_i)$ with label sequence~$(2,r+4,r+5)$ under~$\lambda$ that realizes the arc~$(x,c_i)$.
Similarly, there is the temporal paths~$(c_i,a_i',a_i,x)$ with label sequence~$(2,r+4,r+5)$ under~$\lambda$ that realizes the arc~$(c_i,x)$.
Consequently, each solid edge is realized.

Next, we show that exactly the specified dashed arcs are realized.
To this end, we first show that all arcs of~$D$ are realized.

For each~$i\in [1,k]$, the arcs~$(a'_i,b_i)$, $(b_i,c_i)$, $(c_i,a_i)$, $(a'_i,\top)$, and~$(\top,b_i')$ are realized via the unique path between their endpoints in~$G[\{\top,a_i,a_i',b_i,b_i',c_i\}]$.
The arc~$(b_i,a_i)$ is realized via the path~$(b_i,w_i^{F,u},a_i)$ for an arbitrary vertex~$w_i^{F,u}\in V$, and the arc~$(b_i,\top)$ is realized via the path~$(b_i,w_i^{F_i',u},c_i,\top)$ for some arbitrary~$u\in F_i'$.
Note that such a vertex exists, since each hyperedge is non-empty.

Let~$u\in U$.
Since~$S$ is a set cover, there is some~$j\in [1,k]$, such that~$u\in F_j'$.
Hence, the arc~$(u,\top)$ is realized via the path~$(u,w_j^{F_j',u},c_j,\top)$.
Let~$F\in \mf$ be a hyperedge that contains~$u$.
For each~$i\in [1,k]$, the arc~$(u,a_i)$ is realized via the path~$(u,w_i^{F,u},a_i)$, and the arcs $(u,b_i)$, $(u,b_i')$, and~$(u,c_i)$ are realized via subpaths of the path~$(u,w_i^{F,u},b_i,b_i',c_i)$.
Additionally, the arcs~$(u, v_i^{F,u})$, $(v_i^{F,u},a_i)$, $(b_i,v_i^{F,u})$, $(w_i^{F,u},\top)$, $(w_i^{F,u},b_i')$, $(a_i',w_i^{F,u})$, $(q_i^{F,u},v_i^{F,u})$, $(q_i^{F,u},a_i)$, $(q_i^{F,u},b_i)$, and~$(q_i^{F,u},b_i')$ are realized via the unique path between their endpoints in~$G[\{u,v_i^{F,u},w_i^{F,u},q_i^{F,u},\top,a_i,a_i',b_i,b_i'\}]$, and the arc~$(q_i^{F,u},c_i)$ is realized via the path~$(q_i^{F,u},w_i^{F,u},b_i,b'_i,c_i)$.
Now consider the arcs in~$D$ between~$v_i$-vertices.
Let~$F_x,F_y\in \mf$ with $x<y$ and let~$u_x\in F_x$ and~$u_y\in F_y$.
For each~$i\in [1,k]$, the arc~$(v_i^{F_x,u_x},v_i^{F_y,u_y})$ is realized via the path~$(v_i^{F_x,u_x},\bot_i,v_i^{F_y,u_y})$ with label sequence~$(\alpha_x,\alpha_y)=(x+2,y+2)$.

Finally, we show that no non-arc of~$D$ is realized.
To this end, we analyze the structure of temporal paths under~$\lambda$.
Observe that for each vertex~$x$ in~$R:=\{\top\} \cup U \cup \{v_i^{F,u}\mid i\in [1,k], F\in \mf, u\in F\}$, $\lambda$ assigns the same label to all edges incident with~$x$.
This implies that no strict temporal path under~$\lambda$ has a vertex of~$R$ as an internal vertex.
Moreover, $(w_i^{F,u},q_i^{F,u},\top)$ is the only temporal path of length at least two has~$q_i^{F,u}$ as an internal vertex, since~$\{w_i^{F,u},q_i^{F,u}\}$ receives label 1 and~$\top\in R$.
Since~$(w_i^{F,u},\top)\in A$, these paths do not realize non-arcs of~$D$.
In particular, no vertex besides~$w_i^{F,u}$ and~$\top$ can reach~$q_i^{F,u}$.
Similarly, the only temporal paths that contain~$\bot_i$ as an internal vertex are the paths of the form~$(v_i^{F_x,u_x},\bot_i,v_i^{F_y,u_y})$ with~$x < y$, since each solid neighbor of~$\bot_i$ is from~$R$.
Since~$(v_i^{F_x,u_x},v_i^{F_y,u_y})\in A$, these paths do not realize non-arcs of~$D$.
Let~$R' := R \cup \{q_i^{F,u}\mid i\in [1,k],F\in \mf,u\in F\} \cup \{\bot_i\mid i\in [1,k]\}$.

Next, we show that for each~$s = w_i^{F,u}\in V$, there is no temporal path from~$s$ to any vertex~$t\in V \setminus \{w_i^{F,u},v_i^{F,u},q_i^{F,u}, a_i,b_i,a_i',b_i',c_i,\top,u\}$.
Assume towards a contradiction that there is such a temporal path~$P$. 
Note that~$P$ has length at least~$2$, since there is no solid edge (and thus no labeled edge) between~$s$ and~$t$.
By the above argumentation, $P$ cannot contain any of the vertices of~$R'$ as internal vertices. 
This implies that the~$P$ has as first internal vertex~$s'$ one of the solid neighbors of~$s$ that is not in~$R'$.
That is, $s' \in \{a_i,b_i,c_i\}$.
For~$s' = a_i$, the edge~$\{s,a_i\}$ has label~$r+5$, which is the largest globally assigned label.
Hence, $P$ cannot go over~$a_i$.
For~$s' = b_i$, the edge~$\{s,b_i\}$ has label~$2$, which is the largest label incident with~$b_i$ besides the one of edge~$\{b_i,b_i'\}$.
Since~$b_i'$ has only one other solid neighbor (namely~$c_i$), $P$ would also have to traverse the edge~$\{b_i',c_i\}$ at time~$r+5$, which is the largest globally assigned label.
Hence, $P$ cannot go over~$b_i$.
For~$s' = c_i$, the edge~$\{s,c_i\}$ has either no label or label~$\alpha_F$ if~$F_i' = F$.
The only two edges incident with~$c_i$ with a larger label are the edges towards~$\top$ (which is in~$R'$) and towards~$b_i'$ (for which the edge has label~$r+5$).
In all cases, $P$ cannot reach~$t$; a contradiction.

Similarly, only vertices of~$Z:= \{w_i^{F,u},v_i^{F,u},q_i^{F,u}, a_i,b_i,a_i',b_i',c_i,\top,u\}$ can reach~$w_i^{F,u}$ under~$\lambda$.
This in particular implies that each temporal path that contains vertex~$w_i^{F_u}$ contains only vertices of~$Z$.
Since~$G[Z]$ is a graph on ten vertices, it can easily be verified that no non-arc between the vertices of~$Z$ is realized in~$G[Z]$.

Hence, if there would be a temporal path under~$\lambda$ that realizes a non-arc, it has to be in~$G':= G- \{w_i^{F,u}\mid i\in [1,k], F\in \mf, u\in U\}$.
Note that the only connected components of~$G'$ are (i)~components of size~$1$ containing a single vertex of~$U$, (ii)~stars with center~$\bot_i$ for some~$i\in [1,k]$ with all its solid neighbors (that is, the vertices of~$\{v_i^{F,u}\mid i\in [1,k], F\in \mf, u\in U\}$) and (iii)~the component~$G'':=[\{\top\} \cup \{a_i,a'_i,b_i,b'_i,c_i\mid i\in [1,k]\} \cup \{q_i^{F,u}\mid i\in [1,k], F\in \mf, u\in U\}]$.
In the first type of components, there are trivially no temporal paths realizing non-arcs of~$D$.
In the second type of components, there are also no temporal paths realizing non-arcs of~$D$, as we already discussed the temporal paths that go over~$\bot_i$.
Thus, consider the last component~$G''$.
Recall that~$\top$ is in~$R$ which implies that no temporal path can pass through~$\top$.
Hence, if there would be a temporal path under~$\lambda$ that realizes a non-arc, it has to be in~$G[C \cup \{\top\}]$ for some connected component~$C$ of~$G''-\{\top\}$.
Each such resulting graph has size at most~$6$ and it can be verified easily that no included temporal path under~$\lambda$ realizes a non-arc.

Concluding, $\lambda$ realizes~$D$.
\fi
\end{proof}

\section{Parameterizing by the Feedback Edge Set Number}\label{sec:fes}

In this section we generalize our algorithm on tree instances of~\URGD to tree-like graphs.
A feedback edge set of a graph~$G$ is a set~$F$ of edges of~$G$, such that~$G-F$ is acyclic.
We denote by~$\fes$ the feedback edge set number of the solid graph~$G$ of~$D$, that is, the size of the smallest feedback edge set of~$G$.

\begin{theorem}\label{fes algo}
Each version of~\URGD can be solved in~$\fes^{\Oh(\fes^2)} \cdot n^{\Oh(1)}$~time.
\end{theorem}

Recall that the parameter can presumably not be replaced by a smaller parameter like feedback vertex set number or the treedepth (see~\Cref{param hardness}).
In the remainder, we present the proof of~\Cref{fes algo}.
We will only describe the algorithm for the most general version~\any\str\URGD, but for all more restrictive versions, all arguments work analogously.

\subparagraph{Definitions and Notations.}

\iflong
An abstract overview of the basic definitions is given in the following table.\\

\begin{tabular}{l|l}
Object & meaning\\\hline
$D=(V,A)$ & input graph\\
$G=(V,E)$ & solid graph of~$D$\\
$F$ & a fixed feedback edge set of~$G$ of minimum size\\
$X$ & endpoints of the edges of~$F$\\
$V^*$ & the 2-core of~$G$ \\
$V' := V \setminus V^*$ & non-core vertices (no vertex of~$V'$ is in a cycle in~$G$)\\
$T_v$ for some~$v\in V^*$ & pendant tree with root~$v$\\
$Y_3$ & vertices of degree at least~$3$ in~$G[V^*]-F$\\
$Y_2$ & vertices of~$V^*$ that are adjacent to some vertex of~$X \cup Y_3$\\
$X^*:=X \cup Y_2 \cup Y_3$ & \tnew{At least} all vertices of degree at least 3 in~$G[V^*]$ and their neighbors 
\end{tabular}
\fi 

Let~$G =(V,E)$ be the solid graph and let~$F$ be a minimum size feedback edge set of~$G$.
\nnew{We assume without loss of generality that~$G$ is connected, as otherwise, we can solve each connected component independently, or detect in polynomial time that~$D$ is not realizable if there are dashed arcs between different connected components of the solid graph.}
Moreover, let~$X$ denote the endpoints of the edges of~$F$.
Note that~$|X| \leq 2 \cdot |F|$.
We define the set~$V^*$ as the (unique) largest subset~$S$ of vertices of~$V$ that each have at least two neighbors in~$G[S]$, that is, $V^*$ is the~2-core~\cite{S83} of~$G$.  
This set can be obtained from~$V$ by repeatedly removing degree-1 vertices.
Note that~$V^*$ contains all vertices of~$V$ that are part of at least one cycle in~$G$ \tnew{(but may also contain vertices that are not part of any cycle)}.
Moreover, we define~$V' := V \setminus V^*$.
Note that~$X \subseteq V^*$, since~$F$ is a minimum-size feedback edge set.
Since~$G$ is connected, for each vertex~$v' \in V'$, there is a unique vertex~$v\in V^*$ which has closest distance to~$v'$ among all vertices of~$V^*$.
For a vertex~$v\in V^*$, denote by~$V_v$ the set of all vertices of~$V'$ for which~$v$ is the closest neighbor in~$V^*$.
Note that~$V_v$ might be empty.
Since~$v'$ is in no cycle, removing~$v$ from~$G$ would result in~$v'$ ending up in a component that has no vertex of~$V^*$.
Hence, $T_v := G[\{v\} \cup V_v]$  is a tree for which we call~$v$ the~\emph{root}.
Moreover, we call~$T_v$ the~\emph{pendant tree} of~$v$.

Recall that each vertex of~$V^*$ has degree at least~$2$ in~$G[V^*]$. 
\tnew{Furthermore,}~$G - F$ is acyclic and thus~$G[V^*] - F$ is a tree, where each leaf is a vertex of~$X$.
Let~$Y_3$ denote the set of all vertices of degree at least~$3$ in~$G[V^*] - F$, and let~$Y_2$ be the neighbors of~$X\cup Y_3$ in~$V^* \setminus (X \cup Y_3)$.
Recall that each leaf of~$G[V^*] - F$ is a vertex of~$X$.
Moreover, in each tree with~$\ell$ leaves, there are~$\Oh(\ell)$ vertices having (i)~degree at least 3 or (ii)~a neighbor of degree at least 3.
This implies that~$|Y_2\cup Y_3| \in \Oh(|X|)$.
Let~$X^* := X \cup Y_2\cup Y_3$. See Figure \ref{fig:sets-FPT} for an illustration of these sets of vertices and some of the main definitions used in this section.
By the above, $|X^*| \in \Oh(|X|) \subseteq \Oh(\fes)$.
Recall that each vertex of~$Y_2$ has degree exactly~$2$ in~$G[V^*]$.
Moreover, note that each vertex of~$V^* \setminus X^*$ is part of a unique path~$P$ where each internal vertex of~$P$ is from~$V^*\setminus X^*$ and the endpoints of~$P$ are from~$Y_2$.
For each such path, each internal vertex has degree exactly~$2$ in~$G[V^*]$ and there are at most~$|X^*| - 1$ such paths, since~$G[V^*]-F$ is acyclic.

\begin{figure}
\centering
\iflong
\includegraphics[width=0.8\linewidth]{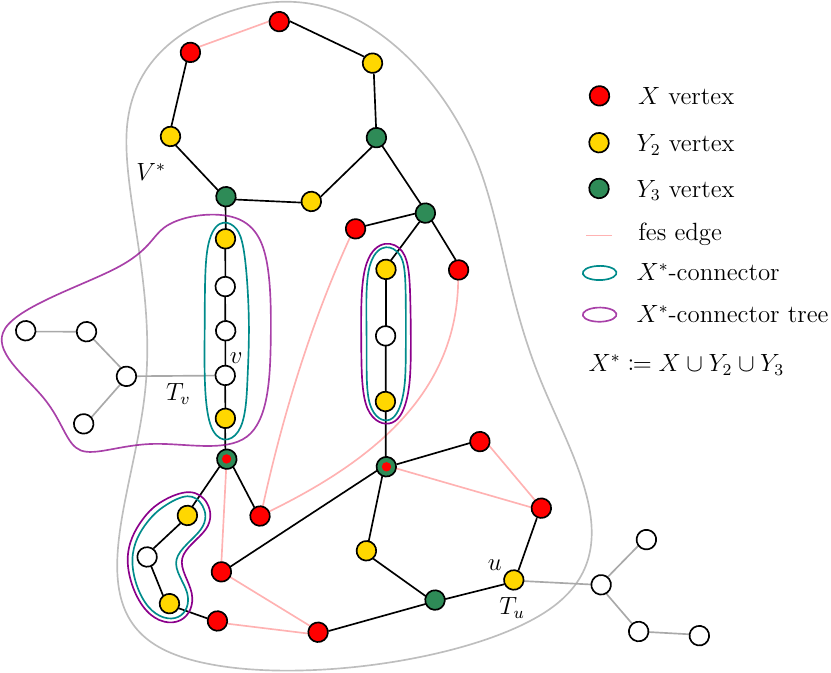}
\fi
\caption{An illustration of the main definitions used by the FPT algorithm for parameter~$\fes$.
The grey area contains the vertices of~$V^*$, that is, the~$2$-core of~$G$.
}
\label{fig:sets-FPT}
\end{figure}

Note that~$X^*$ contains all vertices that are part of triangles in~$G[V^*]$ and thus in~$G$.
This implies that in each minimal realization of~$D$, each edge incident with at least one vertex of~$V \setminus X^*$ receives at most two labels (see~\Cref{lem:tree-atmost2}).
We now formally define the above mentioned paths between the vertices of~$Y_2$ is a more general way.

\begin{definition}
Let~$W$ with~$X^* \subseteq W \subseteq V^*$ and let~$P$ be a path of length at least 2 in~$G[V^*]$ with endpoints~$a$ and~$b$ in~$W$ and all internal vertices from~$V^*\setminus W$.
We call~$P$ a~\emph{$W$-connector}.
Moreover, let~$C := G[V(P) \cup \bigcup_{q\in V(P) \setminus \{a,b\}} V_q]$.
We call~$C$ a~\emph{$W$-connector tree} and the~\emph{extension of~$P$}.
\end{definition}

Note that the endpoints of each~$X^*$-connector are from~$Y_2$ and thus have degree~$2$ in~$G[V^*]$.
Moreover, each internal vertex~$v$ of a~$W$-connector~$P$ has degree exactly~$2$ in~$G[V^*]$, that is, the only two neighbors of~$v$ in~$G[V^*]$ are the predecessor and the successor of~$v$ in~$P$.
This also implies that each endpoint of a~$W$-connector has degree exactly 2 in~$G[V^*]$.
Since each edge of~$P$ is incident with at least one vertex of~$V \setminus W \subseteq V \setminus X^*$, in each minimal realization of~$D$, each edge of each~$W$-connector receives at most two labels (see~\Cref{lem:tree-atmost2}).
The latter also holds for each~$W$-connector tree.

\iflong
\begin{observation}
Let~$W$ with~$X^* \subseteq W \subseteq V^*$ and let~$q\in V \setminus (W \cup \bigcup_{w\in W} V_w)$.
There is exactly one~$W$-connector tree that contains~$q$.\end{observation}

\begin{observation}\label{properties of small sets}
Let~$W$ with~$X^* \subseteq W \subseteq V^*$.
Then, there are $\Oh(|W| + \fes)$~edges between vertices of~$W$ in~$G$ and there are $\Oh(|W| + \fes)$~$W$-connectors.
\end{observation}
\fi

\subparagraph{Abstract Description of the Algorithm}
        Our algorithm uses two preprocessing steps.

        In the first step, we will present a polynomial-time reduction rule based on the splitting operations presented in~\Cref{sec:bridge}.
        With this reduction rule, we will be able to remove vertices \tnew{from pendant trees that have many vertices}.
        After this reduction rule is applied exhaustively, for each vertex~$x\in V^*$, the pendant tree~$T_x$ will have $\Oh(d_x)$~vertices, where~$d_x$ denotes the degree of~$x$ in~$G[V^*]$.

      In the second step, we extend the set~$X^*$ to a set~$W^* \subseteq V^*$ such that each~$W^*$-connector has some useful properties (we will define connectors with these properties as \emph{nice connectors}).
      Intuitively, a connector~$P$ with extension~$C$ is nice if (i)~in a realization for~$D$, the arcs in~$D[V(C)]$ can only be realized by temporal paths that are contained in the connector tree~$C$, and (ii)~for each arc between a vertex outside of~$C$ and a vertex inside of~$C$, we can in polynomial time detect over which (unique) edge of the connector incident with one of its endpoints this arc is realized.
      As we will show, we can compute such a set~$W^*$ of size~$\Oh(\fes)$       in polynomial time, or correctly detect that the input graph is not realizable.
      By the first part, we additionally get that the total number of vertices in pendant trees that have their root in~$W^*$ is~$\Oh(|W^*| + \fes) = \Oh(\fes)$.

The algorithm then works as follows:
We iterate over all possible partial labelings~$\lambda$ on the $\Oh(\fes)$ edges incident with vertices of~$W^*$ or vertices of pendant trees that have their root in~$W^*$.
As we will show, we can assume that each of these edges will only receive $\Oh(\fes)$~labels in each minimal realization.
For each such labeling~$\lambda$, \tnew{we need} to check whether we can extend the labeling to the so far unlabeled edges, that is, the edges that are part of any~$W^*$-connector.
As we will show, we can compute for each such connector~$P$ with extension~$C$ in polynomial time a set~$L_P$ of $\Oh(1)$~labelings for the edges of~$C$, such that if~$\lambda$ can be extended to a realization for~$D$, then we can extend~$\lambda$ (independently from the other connectors) by one of the labelings in~$L_P$.
Since~$W^*$ has size~$\Oh(\fes)$, there are only $\Oh(\fes)$~many~$W^*$-connectors (see~\Cref{properties of small sets}) and for each such connector~$P$, the set~$L_P$ of labelings has constant size.
Our algorithm thus iterates over all possible~$\Oh(1)^{\Oh(\fes)} = 2^{\Oh(\fes)}$ possible labeling combinations for the connectors and checks whether one of these combinations extends~$\lambda$ to a realization for~$D$.

Summarizing, we will show the following.

\begin{proposition}\label{fpt simplified}
In polynomial time, we can detect that~$D$ is not realizable, or compute a set~$W^*$ with~$X^* \subseteq W^* \subseteq V^*$, such that the set of edges~$E^*$ that are (i)~incident with at least one vertex of~$W^*$ or (ii)~part of some pendant tree with root in~$W^*$, has size~$\Oh(\fes)$ and where for each labeling~$\lambda$ of the edges of~$E^*$, we can in $2^{\Oh(\fes)} \cdot n^{\Oh(1)}$~time
\iflong
\begin{itemize}
\item detect that there \tnew{is} no \minlab realization for~$D$ that agrees with~$\lambda$ on the labels of all edges of~$E^*$, or
\item compute a labeling that realizes~$D$.
\end{itemize}
\fi
\end{proposition}

The running time of this algorithm is then $\fes^{\Oh(\fes^2)} \cdot n^{\Oh(1)}$, since (i)~all preprocessing steps run in polynomial time, (ii)~we only have to consider $\fes^{\Oh(\fes^2)}$ labelings~$\lambda$ (since we prelabel $\Oh(\fes)$~edges with $\Oh(\fes)$~labels each), and (iii)~for each labeling~$\lambda$, we only have to check for~$2^{\Oh(\fes)}$ possible ways to extend~$\lambda$ to a realization for~$D$.
To see the second part, we show that it is sufficient to assign only labels of~$\{i\cdot 2\cdot n \mid i\in \Oh(\fes^2)\}$ in~$\lambda$.

\iflong
To describe our algorithm, we will use \tnew{the notion of ``merging''} several partial labelings.
This is formally defined as follows.

\begin{definition}
Let~$E_1$ and~$E_2$ be (not necessarily \tnew{disjoint}) subsets of~$E$. 
For~$\lambda_1\colon E_1 \to 2^{\mathbb{N}}$ and~$\lambda_2\colon E_2 \to 2^{\mathbb{N}}$, we define the labeling~$\lambda_1 \ltimes \lambda_2\colon E_1 \cup E_2 \to 2^{\mathbb{N}}$ by
$(\lambda_1 \ltimes \lambda_2) (e) := \lambda_2(e)$ if~$e\in  E_2$ and~$(\lambda_1 \ltimes \lambda_2) (e) := \lambda_1(e)$ otherwise.  
\end{definition}
\fi

\iflong
\subsection{Independent structural implication}
In our algorithm, we will enumerate all reasonable prelabelings of a set of $\Oh(\fes)$~edges.
To ensure that this enumeration can be done in FPT-time for~$\fes$, we need to ensure that 'few' labels per edges are sufficient.
The following result ensures this property, but is also of independent interest for the understanding of our problem.

\begin{theorem}
Let~$D$ be an instance of~\URGD with solid graph~$G$.
In each minimal realization of~$D$, each edge of~$G$ receives~$\Oh(\fes)$ labels.
\end{theorem}
\begin{proof}
Let~$\lambda$ be a realization of~$D$.
Let~$F$ be a feedback edge set of~$G$ of size~$\fes$.
Moreover, let~$V^*$ and~$X^*$ be defined as above based on~$F$ and recall that~$|X^*| \in \Oh(|F|) = \Oh(\fes)$.
If each edge receives $\Oh(\fes)$~labels, we are done.
So assume that there is an edge~$e$ that receives~$\omega(\fes)$ labels.
We will show that~$\lambda$ is not a minimal realization for~$D$.
Note that if not both endpoints of~$e$ are in~$X^*$, then~$e$ is not part of a triangle in~$G$, which implies that~$\lambda$ is not a minimal realization, since in each such realization, $e$ would receive at most two labels (see~\Cref{lem:tree-atmost2}).
Thus, assume in the following that both endpoints of~$e$ are from~$X^*$.

Similar to the proof of~\Cref{max number label}, we denote for each vertex $v\in V$ with~$\alpha_v$ the smallest label assigned to~$e$ for which there is a temporal path under~$\lambda$ that starts in~$v$ and traverses~$e$ at time~$\alpha_v$, if such a label exists.
Let~$L_e := \{\alpha_v \mid v\in V, \alpha_v \text{ exists}\}$.
Note that if there is at least one label~$\ell$ assigned to~$e$ that is not contained in~$L_e$, then removing label~$\ell$ from~$e$ still preserves at least one temporal path between any two vertices for which a temporal path exists under~$\lambda$.
Consequently, if such a label~$\ell$ exists, we can immediately conclude that~$\lambda$ is not a minimal realization for~$D$.
Recall that~$e$ receives $\omega(\fes)$ labels.
Hence, to show that~$\lambda$ is not a minimal realization, it suffices to show that~$L_e$ has size~$\Oh(\fes)$.

To this end, note that each vertex of~$V$ is (a)~a vertex of~$X^*$, (b)~contained in some~$X^*$-connector tree~$C$ (excluding the two vertices \tnew{of~$C$} in~$X^*$), or (c)~a vertex of some pendant tree rooted in~$X^*$ (excluding the root).

Since~$X^*$ has size~$\Oh(\fes)$, the vertices of~$X^*$ contribute~$\Oh(\fes)$ different labels to~$L_e$.

Let~$C$ be a~$X^*$-connector tree where~$a$ and~$b$ are the \tnew{two unique} vertices of~$X^*$ in~$C$.
By definition of connector trees, $a$ and~$b$ are both leaves in~$C$.
Moreover, since these vertices are the only vertices of~$X^*$ in~$C$ and both endpoints of~$e$ are from~$X^*$, each temporal path starting in~$C$ and traversing edge~$e$ has to leave the connector tree via the unique edge~$\{a,a'\}$ incident with~$a$ in~$C$ or via the unique edge~$\{b,b'\}$ incident with~$b$ in~$C$. 
Since both edges are not part of a triangle, they both receive at most two labels each under~$\lambda$, or~$\lambda$ is not a minimal realization for~$D$ (see~\Cref{lem:tree-atmost2}).
Thus, assume that both edges receive at most two labels each.
Hence, all vertices of~$V(C) \setminus \{a,b\}$ contribute in total at most four labels to~$L_e$.
Since there are $\Oh(\fes)$~$X^*$-connectors (see~\Cref{properties of small sets}), all vertices that are \nnew{contained in~$X^*$-connectors} thus contribute $\Oh(\fes)$~labels in total to~$L_e$.

To show that~$L_e$ has size~$\Oh(\fes)$, it suffices to show that all vertices that are contained in pendant trees with roots in~$X^*$ contribute $\Oh(\fes)$~labels in total to~$L_e$.
Let~$x\in X^*$.
We show that all vertices of~$T_{x}$ contribute $\Oh(d_x)$~labels in total to~$L_e$, where~$d_x$ denotes the degree of~$x$ in~$G[V^*]$.
Let~$v\in V_x$ and let~$P$ be a temporal path from~$v$ that traverses the edge~$e$ at time~$\alpha_v$.
Note that~$P$ traverses \tnew{one} edge~$e_1$ incident with~$x$ in~$T_x$ and then traverses an edge~$e_2$ between~$x$ and some other vertex of~$V^*$.
Since~$e_1$ is an edge of~$T_x$, $e_1$ is a bridge in~$G$, which implies that~$\max \lambda(e_1) \leq \min \lambda(e_2)$.
Hence, we can assume that~$P$ either traverses the edge~$e_2$ at time~$\min \lambda(e_2)$ or at time~$\min (\lambda(e_2) \setminus \min \lambda(e_2))$ (if~$e_2$ receives at least two labels).
Thus, for each edge between~$x$ and some other vertex of~$V^*$, at most two labels are used that define the~$\alpha_v$ values for the vertices~$v\in V_x$.
This implies that the vertices of~$T_{x}$ contribute $\Oh(d_x)$~labels in total to~$L_e$.
Note that~\Cref{properties of small sets} implies that the sum of degrees of all vertices of~$X^*$ in~$G[V^*]$ is~$\Oh(\fes)$.
Hence, all vertices that are contained in pendant trees with roots in~$X^*$ contribute $\Oh(\fes)$~labels in total to~$L_e$.

Consequently, $L_e$ has size~$\Oh(\fes)$.
This implies that~$\lambda$ is not a minimal realization, since~$e$ receives $\omega(\fes)$ labels under~$\lambda$.
\end{proof}

\subsection{Dealing with the pendant trees}

Our goal in this subsection is to prove the following.

\begin{proposition}\label{shrink pendant}
In polynomial time, we can (i)~detect that~$D$ is not realizable or (ii)~obtain an equivalent instance by (a)~adding polynomially many degree-1 vertices to~$D$ and (b)~removing vertices from~$D$, such that in the resulting instance, in each pendant tree~$T_v$ for~$v\in V^*$, $\Oh(d_v)$ vertices remain, where~$d_v$ denotes the degree of~$v$ in~$G[V^*]$.
\end{proposition}

For each vertex~$v\in V^*$, let~$B_v$ denote the vertices of~$T_v$ that are adjacent to~$v$ in~$G$.
By definition of~$T_v$, for each vertex~$b \in B_v$, the edge~$\{v,b\}$ is a bridge in~$G$.
We call~$B_v$ the \emph{bridge vertices} of~$T_v$.

For each vertex~$v\in V^*$, apply the operations behind~\Cref{lem:split-nonspecial}, \Cref{lem:split-special}, and~\Cref{lem:remove-pendant} to each bridge between~$v$ and~$B_v$.
These operations might possibly detect that~$D$ is not realizable.
In this case, report that~$D$ is \tnew{not} realizable.
Otherwise, each such operation individually splits our instance into at most two instances.
Moreover, at most one of these instances contains cycles, since each bridge in~$T_v$ has a tree on one side.
For all resulting tree instances, we check in polynomial time whether they are realizable.
If at least one is not realizable, we correctly output that~$D$ is not realizable.
Otherwise, we only (again) remain with one instance, and this instance contains cycles.
In particular, this instance is (besides some possible degree-1 vertices) a subgraph of~$D[V^* \cup N(V^*)]$.
For convenience, assume that~$D$ already is the unique resulting instance.
Note that~$G[\{v\} \cup B_v]$ is an induced star with center~$v$, where each edge between~$v$ and a vertex of~$B_v$ is a bridge in~$G$.
Hence, with the same arguments behind~\Cref{lem:reductiontournament}, $D$ is not realizable if~$D[B_v]$ is not a complete DAG, since the operations behind~\Cref{lem:split-nonspecial} and~\Cref{lem:remove-pendant} were applied on all \tnew{bridges} between~$v$ and~$B_v$.
Hence, if~$D[B_v]$ is not a complete DAG, we correctly output that~$D$ is not realizable.
Otherwise, we continue with the knowledge that~$D[B_v]$ is a complete DAG.
This complete DAG~$D[B_v]$ thus implies a unique total order on the vertices of~$B_v$ starting from the unique source in~$D[B_v]$.
In the following, we denote the~$i$th vertex of this total order of~$B_v$ by~$b_v^i$.
Consequently, for each~$1\leq i < j \leq |B_v|$, $(b_v^i, b_v^j) \in A$.

\nnew{
Additionally, we check for any two adjacent bridges~$e$ and~$e'$ whether they are both bundled and separated.
If this is the case for any pair of adjacent bridges, we can correctly detect that~$D$ is not realizable due to~\Cref{bundledseparated}.
Since this property can be checked in polynomial time, we assume in the following that no pair of adjacent bridges is both bundled and separated.
}

We now observe a necessary property for realizable instances.

\begin{lemma}\label{prefix suffix}
Let~$v\in V^*$ and let~$u\in V \setminus (\{v\} \cup V_v)$.
If~$D$ is realizable, then the set~$[1,|B_v|]$ can be partitioned in three consecutive and pairwise disjoint (possibly empty) intervals~$I_1$, $I_2$, and~$I_3$, such that 
\begin{itemize}
\item $(b_v^i,u)\in A$  and~$(u,b_v^i)\notin A$ for each~$i\in I_1$,
\item $(b_v^i,u)\notin A$ and~$(u,b_v^i)\notin A$ for each~$i\in I_2$, and
\item $(u,b_v^i)\in A$ and~$(b_v^i,u)\notin A$ for each~$i\in I_3$.
\end{itemize}
\end{lemma}
\begin{proof}
Since~$u\in V \setminus (\{v\} \cup V_v)$, for each~$b_v^i \in B_v$, $A$ contains at most one of the arcs~$(u,b_v^i)$ and~$(b_v^i,u)$.
Let~$\lambda$ be a realization of~$D$.

If~$(b_v^i,u)\in A$ for some~$i\in[2,|B_v|]$, then there is a temporal $(b_v^i,u)$-path under~$\lambda$.
This temporal path starts by traversing the edge~$\{b_v^i,v\}$ at time at least~$\min \lambda(\{b_v^i,v\})$.
Since for each~$j\in [1,i-1]$, $A$ contains the arc~$(b_v^j,b_v^i)$, $\min \lambda(\{b_v^j,v\}) \leq \min \lambda(\{b_v^i,v\})$.
Hence, there is also a temporal~$(b_v^j,u)$-path under~$\lambda$.
This implies that~$(b_v^j,u)\in A$.
Moreover, this implies that the interval~$I_1$ exists as stated and extends from the start of~$B_v$ (ending with the largest~$i$ for which~$(b_v^i,u)\in A$). 

Similarly, if~$(u,b_v^i)\in A$ for some~$i\in[1,|B_v|-1]$, then there is a temporal $(u,b_v^i)$-path under~$\lambda$.
This temporal path ends by traversing the edge~$\{v,b_v^i\}$ at time at most~$\max \lambda(\{v,b_v^i\})$.
Since for each~$j\in [i+1,|B_v|]$, $A$ contains the arc~$(b_v^i,b_v^j)$, $\min \lambda(\{v,b_v^j\}) \geq \max \lambda(\{v,b_v^i\})$.
Hence, there is also a temporal~$(u, b_v^j)$-path under~$\lambda$.
This implies that~$(u,b_v^j)\in A$.
Moreover, this implies that the interval~$I_3$ exists as stated and extends to the end of~$B_v$ (starting from the smallest~$i$ for which~$(b_v^i,u)\in A$).

Consequently, the interval~$I_2$ contains the remaining vertices of~$B_v$.
None of the vertices of~$I_2$ can have an arc from or towards~$u$ as otherwise, one of the other two intervals could be extended.
\end{proof}

Note that this property can be checked for each pair of vertices~$v$ and~$u$ in polynomial time.
Hence, if the property does not hold for some vertex pair, we correctly output that~$D$ is not realizable.
In the remainder, assume that these intervals exist for each pair of vertices~$v$ and~$u$.
Moreover, note that they can be computed in polynomial time.
We may call the interval~$I_2$ the~\emph{middle interval of~$u$ with respect to~$T_v$}.

Based on these intervals, we now define  for each vertex~$v\in V^*$ and each neighbor~$u$ of~$v$ in~$V^*$ a set of $\Oh(1)$~edges in~$T_v$ that 'surround' the edge~$\{u,v\}$.
Intuitively, these edges will be the only edges that can possibly share a label with~$\{u,v\}$ in any realization of~$G$.

\begin{definition}
Let~$e := \{u,v\} \subseteq V^*$ with~$e \in E$, and let~$B_v$ denote the bridge vertices of~$T_v$.
Moreover, let~$I_2$ denote \tnew{the} middle interval of~$u$ with respect to~$T_v$ according to~\Cref{prefix suffix}.
If~$B_v\neq \emptyset$, we define the~\emph{edges of~$T_v$ that surround~$e$} as:
\begin{itemize}
\item $\{\{v, b_v^j\}\mid j \in [i-2, i +2] \cap [1,|B_v|]\}$, if the interval~$I_2$ has size at most one,  where~$i\in [1, |B_v|]$ is the smallest index, such that~$(b_v^j, u) \in A$ for each~$1 \leq j < i$ and~$(u, b_v^k) \in A$ for each~$i < k  \leq |B_v|$, and as
\item $\emptyset$, otherwise.
\end{itemize}
In the latter case, we say that the edges~$\{b_v^j,v\}$ and~$\{b_v^k,v\}$ \emph{block the edge~$e$}, where~$j = \min I_2$ and~$k = \max I_2$.
\end{definition}

We now show that this definition matches the previous intuition.

\begin{lemma}\label{surrounding property}
Let~$e := \{u,v\} \subseteq V^*$ with~$e \in E$, and let~$S$ denote the edges of~$T_v$ that surround~$e$.
Then, in each realization of~$D$, $e$ shares no label with any edge of~$T_v$ that is incident with~$v$ and not in~$S$.
Moreover, for each~$i\in [1,|B_v|-1]$ where~$\{v,b_v^i\}$ and~$\{v,b_v^{i+1}\}$ are not in~$S$, in each \minlab realization~$\lambda$ of~$D$, $\lambda(e) \cap [\min \lambda(\{v,b_v^i\}), \max \lambda(\{v,b_v^{i+1}\})] = \emptyset$.
\end{lemma}
\begin{proof}
Recall that each bridge~$e$ of~$G$ receives at least one label in each realization of~$D$.
Let~$B_v$ be the bridge vertices of~$T_v$.
Recall that~$D[B_v]$ is a complete DAG.
This implies that for each realization~$\lambda$ for~$D$ and for each two~$b_v^k$ and~$b_v^\ell$ from~$B_v$ with~$k < \ell$, $\min \lambda(\{v, b_v^k\}) < \max \lambda(\{v, b_v^\ell\})$ and~$\max \lambda(\{v, b_v^k\}) \leq \min \lambda(\{v, b_v^\ell\})$.

We distinguish two case.

\textbf{Case 1:} $S = \emptyset$\textbf{.} 
That is, there is some~$i\in [1, |B_v|-1]$, such that there is no arc between~$u$ and any vertex of~$\{b_v^i,b_v^{i+1}\}$.
Let~$\lambda$ be a realization for~$D$.
We show that~$\lambda(e) = \emptyset$.
Assume towards a contradiction that~$\lambda(e) \neq \emptyset$.
Since there is no arc between~$u$ and~$b_v^i$, \Cref{incident same label} implies that there is some~$\alpha\in \mathbb{N}$, such that~$\lambda(e) = \lambda(\{v,b_v^i\}) = \{\alpha\}$.
Since~$(b_v^i,b_v^{i+1})\in A$, $\max \lambda(\{v,b_v^{i+1}\}) > \min \lambda(\{v,b_v^i\}) = \alpha$.
Thus, $(u,v,b_v^{i+1})$ is a temporal path under~$\lambda$.
This contradicts the fact that~$\lambda$ realizes~$D$, since~$(u,b_v^{i+1})$ is not an arc of~$D$.
Hence, the statement holds if~$S = \emptyset$.

\textbf{Case 2:} $S \neq \emptyset$\textbf{.} 
Assume towards a contradiction that~$e$ shares a label with at least one edge~$\{v,b_v^j\}$ of~$T_v$ that is incident with~$v$ and does not surround~$e$.
Let~$i\in [1,|B_v|]$ be the index over which~$S$ was defined.
We distinguish two subcases.

\textbf{Case 2.1:}~$j < i$\textbf{.}
By definition~$\{\{v,b_v^\ell\}\mid \ell \in \{i-2,i-1,i\}\} \subseteq S$.
Hence, $j \leq i-3$.
By choice of~$i$, $(b_v^{i-1},u)\in A$.
Hence, $\max \lambda(\{v, b_v^{i-1}\}) \leq \min \lambda(e) \leq \max \lambda(e)$.
By the initial argumentation, this implies that~$\max \lambda(\{v, b_v^j\}) \leq  \min \lambda(\{v, b_v^{i-2}\}) < \max \lambda(\{v, b_v^{i-1}\}) \leq \min \lambda(e)$.
This contradicts the assumption that~$e$ shares a label with~$\{v,b_v^j\}$.

\textbf{Case 2.2:}~$j > i$\textbf{.}
This case can be shown analogously.

Consequently, no edge of~$T_v$ incident with~$v$ outside of~$S$ shares a label with~$e$.
We now show the second part of the statement.
That is, let~$i\in[1,|B_v|-1]$ where both~$e_i:= \{v,b_v^i\}$ and~$e_{i+1}:= \{v,b_v^{i+1}\}$ are not in~$S$, we show that~$e$ has no label from~$[\min\lambda(e_i), \max\lambda(e_{i+1})]$.
Assume towards a contradiction that this is not the case.
By the above, we know that~$\lambda(e) \cap (\lambda(e_i) \cup \lambda(e_{i+1})) = \emptyset$.
Moreover, $e$ cannot have a label between~$\min \lambda(e_i)$ and~$\max \lambda(e_{i})$, as otherwise, a solid edge between~$u$ and~$b_v^i$ would be realized.
Similarly, $e$ cannot have a label between~$\min \lambda(e_{i+1})$ and~$\max \lambda(e_{i+1})$.
Hence, the only remaining option is that~$e$ receives a label strictly between~$\max\lambda(e_i)$ and~$\min\lambda(e_{i+1})$.
This then implies that the arcs~$(b_v^i,u)$ and~$(u,b_v^{i+1})$ are realized.
Hence, $b_v^i$ ($b_v^{i+1}$) is \tnew{from} the interval~$I_1$ ($I_3$) of~$u$ with respect to the pendant tree~$T_v$ according to~\Cref{prefix suffix}.
By definition of~$S$, both~$e_i$ and~$e_{i+1}$ are in~$S$; a contradiction.
\end{proof}

Based on this insight about surrounding edges, we now define another sanity check.
This will then ensure that we can assume some consistent reachabilities between vertices outside of~$T_v$ and bridge vertices that are in~$T_v$ but that do not surround any external edge.
The sanity check reads as follows.

\newcommand{\iin}{\mathrm{in}}
\newcommand{\out}{\mathrm{out}}

\begin{lemma}\label{intermediate identical}
Let~$v\in V^*$ with~$|B_v| \geq 3$ and let~$i \in [1,|B_v|-1]$ such that neither edge~$\{b_v^i,v\}$ nor edge~$\{b_v^{i+1},v\}$ is surrounding or blocking any edge incident with~$v$ that is outside of~$T_v$.
If~$D$ is realizable, then for each vertex~$x\in V \setminus (\{v\} \cup V_v)$
\begin{itemize}
\item $D_{xb_v^i} = D_{xb_v^{i+1}}$ (and $D_{xb_v^{i+1}} = D_{x\,\out_v^{i+1}}$, where~$\out_v^{i+1}$ is the out-leaf of~$b_v^{i+1}$, if it exists), and
\item $D_{b_v^ix} = D_{b_v^{i+1}x}$ (and $D_{b_v^ix} = D_{\iin_v^ix}$, where~$\iin_v^i$ is the in-leaf of~$b_v^i$, if it exists).
\end{itemize}
\end{lemma}

\begin{proof}
Let~$\lambda$ be a realization for~$D$ and let~$x\in V \setminus (\{v\} \cup V_v)$.
Moreover, let~$\alpha$ be the smallest label such that there is a temporal~$(x,v)$-path~$P$ under~$\lambda$ that reaches~$v$ at time~$\alpha$.
Since~$x\notin V \setminus (\{v\} \cup V_v)$, the path~$P$ traverses an edge~$e$ between~$v$ and some other vertex of~$V^*$ at time~$\alpha$.
Since neither edge~$\{b_v^i,v\}$ nor edge~$\{b_v^{i+1},v\}$ is surrounding~$e$, \Cref{surrounding property} implies that~$\alpha < \min \lambda(\{b_v^{i},v\})$ or~$\alpha > \max \lambda(\{b_v^{i+1},v\})$.
If~$\alpha < \min \lambda(\{b_v^{i},v\})$, $P$ can be extended to reach the vertices~$b_v^{i}$, $b_v^{i+1}$, and $\out^{i+1}$ (if the latter exists).
In this case, $D_{xb_v^i} = D_{xb_v^{i+1}} = D_{x\,\out_v^{i+1}} = 1$.
Otherwise, that is, if~$\alpha > \max \lambda(\{b_v^{i+1},v\})$ , $x$ can reach none of~$b_v^{i}$, $b_v^{i+1}$, and $\out^{i+1}$.
In this case, $D_{xb_v^i} = D_{xb_v^{i+1}} = D_{x\,\out_v^{i+1}} = 0$.

The reachability towards~$x$ can be shown analogously.
\end{proof}
This property can be checked for each vertex~$v\in V^*$ and each~$b_v^i\in V_v$ in polynomial time.
If the property does not hold for some vertex pair, we correctly output that~$D$ is not realizable.
In the remainder, thus assume that the property holds for all vertex pairs.

Based on this sanity check, we now define our final reduction rule to reduce the size of pendant trees.

\begin{lemma}
Let~$v\in V^*$ and let~$b_v^i \in B_v$ with~$6 \leq i \leq |B_v| - 5$, such that none of the edges~$\{\{v,b_v^j\} \mid i-5 \leq j \leq i+5\}$ is surrounding or blocking any edge~$e$ incident with~$v$ and some other vertex of~$V^*$.
Then, $D$ is realizable if and only if (i)~$D[V(T_v)] = D[\{v\} \cup V_v]$ is realizable and (ii)~$D'$ is realizable, where~$D'$ is the instance obtained from~$D$ by removing~$b_v^i$ and its potential in- and out-leaf. 
\end{lemma}
\begin{proof}
Let~$D'$ be the resulting instance with solid graph~$G'$ after removing~$b_v^i$ and its possible leaf-neighbors form~$D$. 
Let~$D_v:= D[V(T_v)] = D[\{v\} \cup V_v]$.
Moreover, for each~$j\in [1,|B_v|]$, let~$e_i$ denote the bridge~$\{b_v^i,v\}$.

The first direction of the statement follows immediately:  
Let~$\lambda$ be a realization for~$D$.
Then, the restriction of~$\lambda$ to the edges of~$G'$ realizes~$D'$, since no path in~$G$ between any two vertices of~$G'$ uses vertices that are removed from~$G$.
Similarly, the restriction of~$\lambda$ to the edges of~$T_v$ realizes~$D_v$, since no path in~$G$ between any two vertices of~$T_v$ uses vertices outside of~$T_v$.
Hence, if~$D$ is realizable, then both~$D'$ and~$D_v$ are realizable.

So consider the second direction.
Let~$\lambda'$ be a \minlab realization for~$D'$
 and let~$\lambda_v$ be a \minlab realization for~$D_v$.
We show that there is a realization~$\lambda$ for~$D$.
To this end, we will combine both realizations for~$D'$ and~$D_v$.

Recall that~$e_i$ is not an edge that surrounds or blocks any external edge incident with~$v$.
Note that for each external edge~$e$ incident with~$v$, in both~$D$ and~$D'$ the same edges in~$T_v$ are surrounding~$e$.
If~$e$ is surrounded by at least one edge in~$T_v$ under~$D$, then this trivially follows by definition of surrounding edges.
Otherwise, if no edge surrounds~$e$ under~$D$, then, since~$e_i$ is not an edge blocking~$e$, $D'$ still contains both edges that block~$e$ in~$D$.
Based on these two edges, no edge surrounds~$e$ under~$D'$.
Due to~\Cref{surrounding property}, this implies that for each external edge~$e$ incident with~$v$, $e$ receives no label in~$[\min \lambda'(e_{i-5}),\max \lambda'(e_{i+5})]$.

Let~$s=i-3$ if~$e_{i-3}$ is a special bridge in~$D$, and let~$s=i-4$, otherwise.
Similarly, let~$t=i+3$ if~$e_{i+3}$ is a special bridge in~$D$, and let~$t=i+4$, otherwise.
The choice of~$s$ and~$t$ \tnew{is} important to ensure that the edges~$e_{s-1}$ ($e_{t-1}$) and~$e_{s+1}$ ($e_{t+1}$) are separated under~$D$, $D'$, and~$D_v$.
We will show that this property holds when providing a realization for~$D$.
Before we can define a realization for~$D$, we observe several properties about the edges~$e_s$ and~$e_t$.

Let~$j\in [1,s] \cup [t,|B_v|]$.
Observe that~$e_j$ and~$e_i$ are separated under~$D$.
This implies that~$e_j$ and~$e_i$ are not bundled under~$D$, since we would have already detected that~$D$ is a no-instance if~$e_j$ and~$e_i$ are both bundled and separated under~$D$ \nnew{(see~\Cref{bundledseparated})}.
Hence, the removal of~$b_v^i$ and its possible leaves did not make~$e_j$ non-special, if~$e_j$ was special in~$D$. 
In other words, $e_j$ is special in~$D$ if and only if~$e_j$ is special in~$D'$.

Let~$j\in [s,t]$.
Since~$e_j$ does not surround any external edge incident with~$v$, $e_j$ is special in~$D$ if and only \tnew{if}~$e_j$ is special in~$D_v$. 
This is due to the fact that~$e_j$ cannot share a label with any external edge incident with~$v$ (see~\Cref{surrounding property}).

This in particular implies that~$e_s$ ($e_t$) is special in~$D'$ if an only if~$e_s$ ($e_t$) is special in~$D_v$.
Moreover, $\max \lambda_v(e_s) < \min \lambda_v(e_t)$, since~$e_s$ and~$e_t$ are separated under~$D_v$.
Similarly, $\max \lambda'(e_s) < \min \lambda'(e_t)$, since~$e_s$ and~$e_t$ are separated under~$D'$.
We can thus assume without loss of generality that~$\lambda'(e_s) = \lambda_v(e_s)$ and~$\lambda'(e_t) = \lambda_v(e_t)$ (by introducing arbitrarily many empty snapshots in both labelings). 

Based on this property, we now define a realization~$\lambda$ for~$D$.
The labeling~$\lambda$ agrees with~$\lambda_v$ on all edges incident with at least one vertex of~$\{b_v^j\mid j\in [s+1,t-1]\}$, and agrees with~$\lambda'$ on all other edges.
We show that~$\lambda$ realizes~$D$.
Let~$W$ contain the vertices~$b_v^j$ with~$j\in [s+1,t-1]$ and their possible leaves.
Since paths between any two vertices of~$V\setminus W$ only use vertices outside of~$W$, and~$\lambda'$ realizes~$D'$, exactly the arcs of~$D$ between vertices of~$V\setminus W$ are realized by~$\lambda$.
Similarly, paths between vertices of~$W$ only use vertices of~$W \cup \{v\}$, since all vertices of~$W$ are in~$T_v$.
Since~$\lambda$ agrees with~$\lambda_v$ on all edges incident with at least one vertex of~$W$ and~$\lambda_v$ realized~$D_v$, exactly the arcs of~$D$ between vertices of~$W$ are realized by~$\lambda$.

It thus remains to show that exactly the arcs of~$D$ between vertices of~$W$ and vertices of~$V\setminus W$ are realized by~$\lambda$.

Let~$x\in V \setminus (\{v\}\cup V_v)$ and let~$j\in [s,t]$.
Then, since we assumed that the property of~\Cref{intermediate identical} holds (as otherwise we already detected that~$D$ is not realizable), $D_{xb_v^j} = D_{xb_v^{s-1}}$ (and~$D_{x\, \out^j} = D_{xb_v^{s-1}}$ if the out-leaf~$\out^j$ of~$b_v^j$ exists).
Since~$\lambda'$ realizes~$D'$, if~$D_{xb_v^{s-1}} = 1$, then there is a temporal~$(x,b_v^{s-1})$-path~$P$ under~$\lambda'$ (and thus under~$\lambda$).
The path~$P$ traverses the edge~$e_{s-1}$ at time at most~\tnew{$\max\lambda'(e_{s-1}) = \max\lambda(e_{s-1}) < \min\lambda(e_{j})$}.
This implies that there is also a  temporal~$(x,b_v^{j})$-path (and a  temporal~$(x,\out^j)$-path, if~$\out^j$ exists) under~$\lambda$.
Similarly, if~$D_{xb_v^{s-1}} = 0$, then there is no temporal~$(x,b_v^{s-1})$-path under~$\lambda'$ (and thus also not under~$\lambda$).
In particular, each temporal~$(x,v)$-path under~$\lambda$ reaches~$v$ at time at least~$\min \lambda(e_{s-1})$.
As initially \tnew{discussed}, no external edge incident with~$v$ has a label in~$[\min \lambda'(e_{s-1}), \max \lambda'(e_{t+1})]$.
Hence, each temporal~$(x,v)$-path under~$\lambda$ reaches~$v$ at time larger than~$\max \lambda(e_{t+1}) > \max \lambda(e_j)$.
This implies that there is also no temporal~$(x,b_v^{j})$-path (and no  temporal~$(x,\out^j)$-path, if~$\out^j$ exists) under~$\lambda$.

Hence, the possible arcs from~$x$ to~$b_v^j$ and~$\out^j$ are realized by~$\lambda$ if and only if they exist.
Similarly, the possible arcs from~$b_v^j$ and~$\iin^j$ to~$x$ are realized if \tnew{and only if} they exist. 

It thus remains to consider reachability between vertices of~$W$ and vertices of~$V_v\setminus W$.
To analyze this case, we finally show the reason why we chose the specific values for~$s$ and~$t$.
We show that in both~$D'$ and~$D_v$, $e_{s-1}$ and~$e_{s+1}$ are separated.
This is due to the fact that (i)~if~$s = i-3$, then~$e_{i-3}$ is special which implies that~$e_{s-1}$ and~$e_{s+1}$ are separated, and (ii)~if~$s = i-4$, then~$e_{i-3}$ is non-special which implies that~$e_{s-1}$ and~$e_{s+1}$ are separated.
By the earlier argumentation, $e_{s}$ is special under either all or none of the graphs~$D$, $D'$, and~$D_v$.
Thus, each temporal path that traverses first an edge in~$T_v$ not incident with a vertex \tnew{in}~$W$ and afterwards an edge incident with a vertex of~$W$ exists under~$\lambda_v$ if and only if it exists under~$\lambda$.
This implies that~$\lambda$ realizes exactly the arcs between the vertices of~$T_v$, since~\tnew{$\lambda_v$ realizes}~$D_v$.

\tnew{In conclusion},~$\lambda$ realizes~$D$.
\end{proof}

Note that this implies that we can in polynomial time (correctly detect that~$D$ is not realizable or) reduce the size of each fixed pendant tree~$T_x$ with~$x\in V^*$ to a size of~$\Oh(d_x)$, where~$d_x$ denotes the degree of~$x$ in~$G[V^*]$.
This is due to the fact that each of the $\Oh(d_x)$ edges defines at most~$5$ surrounding or blocking edges in~$T_x$, and we can remove neighbors of~$x$ in~$T_x$ if there are more than~11 consecutive (with respect to the topological order on~$B_x$) bridges incident with~$x$ that are all not surrounding or blocking any external edge.
This thus proves~\Cref{shrink pendant}.

\subsection{Ensuring connectors with nice properties}
In this section we show how to extend the vertex set~$X^*$ to a set~$W^*$ of size~$\Oh(\fes)$, such that each~$W^*$-connector has some useful properties that we will use in our algorithm.

To define these useful properties, we distinguish between several types of connectors.
\begin{definition}
Let~$W$ with~$X^*\subseteq W \subseteq V^*$ and let~$P$ be a~$W$-connector with endpoints~$a$ and~$b$ from~$V^*$.
We say that~$P$ is \emph{trivial} if~$P$ has length at most~$3$.
Otherwise, let~$a'$ be the successor of~$a$ in~$P$ and let~$a''$ be the successor of~$a'$ in~$P$.
\tnew{Similarly}, let~$b'$ be the predecessor of~$b$ in~$P$ and let~$b''$ be the predecessor of~$b'$ in~$P$.
We say that~$P$ is an~\emph{(in,out)-connector} if~$(a,a'')\in A$ and~$(b'',b)\in A$.
Moreover, $P$ a~\emph{dense (in,out)-connector}, if the path~$P$ is a dense path in~$D$ and~$P$ is an~(in,out)-connector.
\end{definition}

\subparagraph{Robust Connectors}
The first property we consider are robust connectors, which are defined as follows.

\begin{definition}
Let~$W$ with~$X^*\subseteq W \subseteq V^*$ and let~$P$ be a~$W$-connector with endpoints~$a$ and~$b$ and extension~$C$.
We say that~$P$ is~\emph{robust} if there is no dense path between~$a$ and~$b$ outside of~$C$, that is, if there is no dense path from~$a$ to~$b$ in~$D[V \setminus (V(C) \setminus \{a,b\})]$ and there is no dense path from~$b$ to~$a$ in~$D[V \setminus (V(C) \setminus \{a,b\})]$. 
\end{definition}

In other words, in a robust connector~$P$, the reachability between the vertices of the extension of~$P$ has to be achieved via temporal paths that only use subpaths of the extension~$C$ of~$P$.

\begin{observation}\label{robust has internal solution}
Let~$W$ with~$X^*\subseteq W \subseteq V^*$ and let~$P$ be a robust~$W$-connector with extension~$C$.
Then, for each realization~$\lambda^*$ of~$D$, $\lambda^*_C$ realizes~$D[V(C)]$, where~$\lambda^*_C$ is the restriction of~$\lambda^*$ to the edges of~$C$.
\end{observation}

This then directly implies the following for robust connectors.

\begin{observation}
Let~$W$ with~$X^*\subseteq W \subseteq V^*$ and let~$P$ be a robust~$W$-connector.
Then for each arc~$(q,q') \in A$ with~$q,q'\in V(P)$, the subpath of~$P$ from~$q$ to~$q'$ is dense or~$D$ is not realizable.
\end{observation}

Additionally, we observe the following hereditary property for robust connectors.

\begin{observation}\label{robust property is hereditary}
Let~$W$ and~$W'$ with~$X^*\subseteq W  \subseteq W' \subseteq V^*$.
Moreover, let~$P$ be a~$W$-connector and let~$P'$ be a~$W'$-connector such that~$P'$ is a subpath of~$P$.
Then, $P'$ is robust if~$P$ is robust.
\end{observation}

Finally, since the definition of robust connectors only depends on the non-existence of paths between the endpoints that are outside of the connector, a connector~$P$ remains robust, even if adding vertices outside of~$P$ to~$W$.
Formally, this reads as follows.

\begin{observation}\label{robust property is independent}
Let~$W$ and~$W'$ with~$X^*\subseteq W \subseteq W' \subseteq V^*$.
Moreover, let~$P$ be a path in~$G[V^*]$, such that~$P$ is both a~$W$-connector and a~$W'$-connector.
Then, $P$ is a robust~$W$-connector if and only if~$P$ is a robust~$W'$-connector.
\end{observation}

Based on these properties, we now show that we can find a constant number of vertices of a given connector~$P$ to add to~$W$ to ensure that each resulting subconnector is robust or trivial. 

\begin{lemma}\label{make robust connectors}
Let~$W$ with~$X^*\subseteq W \subseteq V^*$ and let~$P$ be a non-trivial~$W$-connector.
Then, we can compute in polynomial time a set~$U_P \subseteq V(P)$ of size at most~$6$, such that (i) there are at most two~$(W \cup U_P)$-connectors which are subpaths of~$P$ and (ii)~each~$(W \cup U_P)$-connector which is a subpath of~$P$ is trivial or robust.
\end{lemma}
\begin{proof}
Let~$a,a',a''$ be the first three vertices of~$P$ and let~$b'',b',b$ be the last three vertices of~$P$.
We distinguish two cases.

\textbf{Case 1:}~$P$ is not an (in,out)-connector\textbf{.}
We set~$U_P := \{a',a'',b',b''\}$.
Clearly, there is at most one $(W \cup U_P)$-connector which is a subpath of~$P$, namely the subpath~$P'$ of~$P$ from~$a''$ to~$b''$.
Assume towards a contradiction that~$P'$ is not robust.
This implies that there is a dense path~$R$ in~$D$ between~$a''$ and~$b''$ that uses no edge of~$P'$.
Assume without loss of generality that~$R$ goes from~$a''$ to~$b''$.
Note that the dense path~$R$ has to visit both~$a$ and~$b$.
This implies that~$(a'',a)\in A$ and~$(b,b'')\in A$;
a contradiction to the fact that~$P$ is not an (in,out)-connector.
Consequently, $P'$ is robust.

\textbf{Case 2:}~$P$ is an (in,out)-connector\textbf{.}
Without loss of generality assume that~$(a,a'')\in A$ and~$(b'',b)\in A$.
Moreover, let~$p_i$ denote the~$i$th vertex of~$P$ when going from~$a$ to~$b$.
That is, $p_1 = a$ and~$p_\ell = b$, where~$\ell$ denotes the number of vertices of~$P$.
Recall that~$\ell \geq 5$, since~$P$ is non-trivial.
We distinguish two cases.

\textbf{Case 2.1:} For each~$j\in [3,\ell-1]$, $(p_{j},a)\notin A$\textbf{.}
Note that this implies that~$(b'',a)\notin A$.
Moreover, $(a'',a)\notin A$. 
We set~$U_P := \{a',a'',b',b''\}$.
Clearly, there is at most one $(W \cup U_P)$-connector which is a subpath of~$P$, namely the subpath~$P'$ of~$P$ from~$a''$ to~$b''$.
Note that there is no dense path~$R$ in~$D$ between~$a''$ and~$b''$ that uses no edge of~$P'$, since each such path would contain~$a$ and thus cannot be dense.
This is due to the fact that~$(a'',a)\notin A$ and~$(b'',a)\notin A$.
Hence, $P'$ is robust.

\textbf{Case 2.2:} There is some~$k\in [3,\ell-1]$ with~$(p_{k},a)\in A$\textbf{.}
Let~$j$ be the smallest integer in~$[3,\ell-1]$ for which~$(p_j,a)\in A$.
Hence, $(p_{j-1},a)\notin A$. 
Note that~$j \neq 3$, since~$p_3 = a''$ and~$(a,a'')\in A$ (which implies that~$(a'',a)\notin A$).
We set~$U_P := \{a',a'',b',b'', p_{j-1},p_j\}$.
Clearly, there are at most two $(W \cup U_P)$-connectors which are subpaths of~$P$, namely the subpath~$P_1$ of~$P$ from~$a''$ to~$p_{j-1}$ and the subpath~$P_2$ of~$P$ from~$b''$ to~$p_{j}$.
We now show that each such subpath (if it exists) is a trivial or robust connector.

First consider~$P_1$.
Assume that~\tnew{$P_1$} exists and is non-trivial.
There is no dense path~$R_1$ from~$a''$ to~$p_{j-1}$ outside of~$P_1$, since (i)~such a path contains~$a$ and~$p_{j}$ (in that order) and (ii)~$(a,p_{j})\notin A$ (since~$(p_{j},a)\in A$).
Similarly, there is no dense path~$R_1$ from~$p_{j-1}$ to~$a''$ outside of~$P_1$, since such a path contains~$a$ and~$(p_{j-1},a)\notin A$.
Hence, if~$P_1$ exists, $P_1$ is trivial or robust.

Next, consider~$P_2$. 
Assume that~$P_2$ exists and is non-trivial.
There is no dense path~$R_2$ from~$p_{j}$ to~$b''$ outside of~$P_2$, since (i)~such a path contains~$p_{j-1}$ and~$a$ (in that order) and (ii)~$(p_{j-1},a)\notin A$.
Similarly, there is no dense path~$R_2$ from~$b''$ to~$p_j$ outside of~$P_2$, since such a path contains~$a$ and~$(a, p_{j})\notin A$ (since~$(p_j,a)\in A$).
Hence, if~$P_2$ exists, $P_2$ is trivial or robust.
\end{proof}

\subparagraph{Nice Connectors}
We now define the more general property we are interested in.

\begin{definition}
Let~$W$ with~$X^*\subseteq W \subseteq V^*$ and let~$P$ be a non-trivial~$W$-connector with endpoints~$a$ and~$b$.
Moreover, let~$a'$ be the successor of~$a$ in~$P$, let~$b'$ be the predecessor of~$b$ in~$P$, and let~$C$ denote the extension of~$P$.
We say that~$P$ is \emph{nice} if~$P$ is robust and
\begin{itemize}
\item if for each vertex~$q\in V(P)$, (i)~$(q,a)\notin A$ or~$(q,b)\notin A$ and (ii)~$(a,q)\notin A$ or~$(b,q)\notin A$, or
\item if~$P$ is a dense~(in,out)-connector and (i)~for each vertex~$v \in V \setminus V(C)$ with~$(v,a) \in A$ and~$(v,b)\in A$, there is no dense path from~$v$ to~$b$ in~$D[(V \setminus V(C)) \cup \{a,b\}]$, and (ii)~for each vertex~$v \in V \setminus V(C)$ with~$(a,v) \in A$ and~$(b,v)\in A$, there is no dense path from~$a$ to~$v$ in~$D[(V \setminus V(C)) \cup \{a,b\}]$.
\end{itemize}
\end{definition}

The benefit of nice connectors is the following.

\newcommand{\vin}{v^{\mathrm{in}}}
\newcommand{\vout}{v^{\mathrm{out}}}

\begin{lemma}\label{compute unique entrance point}
Let~$W$ with~$X^*\subseteq W \subseteq V^*$ and let~$P$ be a nice~$W$-connector with endpoints~$a$ and~$b$ from~$V^*$.
Moreover, let~$a'$ be the successor of~$a$ in~$P$, let~$b'$ be the predecessor of~$b$ in~$P$, and let~$C$ denote the extension of~$P$.
Finally, let~$\vin \in V(C) \setminus \{a,b\}$ and let~$\vout \in (V \setminus V(C)) \cup \{a,b\}$ with~$(\vin,\vout) \in A$ or~$(\vout,\vin)\in A$.
In polynomial time we can
\begin{itemize}
\item detect that~$D$ is not realizable, or
\item compute an edge~$e^*\in \{\{a,a'\},\{b,b'\}\}$, such that if~$(\vin,\vout) \in A$ (if $(\vout,\vin)\in A$), each dense~$(\vin,\vout)$-path (dense~$(\vout,\vin)$-path) contains the edge~$e^*$.
\end{itemize}
\end{lemma}
\begin{proof}
We distinguish two cases.

\textbf{Case 1:} $P$ is not a dense (in,out)-connector\textbf{.}
Since~$P$ is nice, there is at most one~$c\in\{a,b\}$, such that~$A$ contains arc~$(\vin,c)$ (arc $(c,\vin)$).
If this holds for neither~$a$ nor~$b$, then~$D$ is not realizable, since each dense~$(\vin,\vout)$-path (dense $(\vout,\vin)$-path) contains at least one vertex of~$\{a,b\}$.
Otherwise, let~$c$ be the unique vertex of~$\{a,b\}$ for which~$A$ contains arc~$(\vin,c)$ (arc $(c,\vin)$).
Then, each dense~$(\vin,\vout)$-path (dense $(\vout,\vin)$-path)
contains the edge~$e^* := \{c,c'\}$.
Since all these checks can be performed in polynomial time, the statement follows.

\textbf{Case 2:} $P$ is a dense (in,out)-connector\textbf{.}
If~$A$ contains neither of the arcs~$(a,\vout)$ or~$(b,\vout)$ (neither of the arcs~$(\vout,a)$ or~$(\vout,b)$), then~$D$ is not realizable, since each dense~$(\vin,\vout)$-path (dense~$(\vout,\vin)$-path) contains at least one of~$a$ and~$b$.
Hence, assume in the following that~$A$ contains at least of the arcs~$(a,\vout)$ or~$(b,\vout)$ ($A$ contains at least one of the arcs~$(\vout,a)$ or~$(\vout,b)$).

If~$A$ contains both~$(a,\vout)$ and~$(b,\vout)$ (both~$(\vout,a)$ or~$(\vout,b)$), then there is no dense~$(a,\vout)$-path (dense~$(\vout,b)$-path) in~$D[(V \setminus V(C)) \cup \{a,b\}]$, since~$D$ is nice.
Hence, $e^* :=\{b,b'\}$ ($e^* :=\{a,a'\}$) is contained in each dense~$(\vin,\vout)$-path (dense~$(\vout,\vin)$-path) in~$D$.

Otherwise, there is a unique vertex~$c\in \{a,b\}$, such that~$(c,\vout) \in A$ ($(\vout,c) \in A$).
Clearly, the edge~$e^* :=\{c,c'\}$ is contained in each dense~$(\vin,\vout)$-path (dense~$(\vout,\vin)$-path) in~$D$.

Since all these checks can be performed in polynomial time, the statement follows.
\end{proof}

Similar to the robust property, the property of being a nice connector is preserved under addition of vertices that are not in the specific connector.

\begin{observation}\label{nice property is independent}
Let~$W$ and let~$W'$ with~$X^* \subseteq W \subseteq W' \subseteq V^*$.
Moreover, let~$P$ be a path in~$G[V^*]$ such that~$P$ is both a~$W$-connector and a~$W'$-connector.
Then, $P$ is nice~$W$-connector if and only if~$P$ is a nice~$W'$-connector.
\end{observation}

Finally, we show that similar to robust connectors, we can compute in polynomial time a set of vertices of a connector~$P$ to add to~$W$ to ensure that each resulting subconnector is nice or trivial.
\nnew{
To this end, recall that due to~\Cref{robust has internal solution}, we can immediately conclude that~$D$ is not realizable if~$D[V(C)]$ is not realizable for the extension~$C$ of a robust connector.
Since~$G[V(C)]$ is a tree, this property can be checked in polynomial time.
Hence, by performing this check whenever we consider the extension~$C$ of a robust connector, we ensure that in polynomial time, we can (i)~detect that~$D$ is not realizable or (ii)~ensure that~$D[V(C)]$ is realizable.
}

\begin{lemma}\label{make nice connectors}
Let~$W$ with~$X^*\subseteq W \subseteq V^*$ and let~$P$ be a non-trivial and robust~$W$-connector with extension~$C$.
If~$D[V(C)]$ is realizable, then we can compute in polynomial time a set~$U_P \subseteq V(P)$ of size at most~$10$, such that (i)~there are at most two~$(W \cup U_P)$-connectors which are subpaths of~$P$ and (ii)~each~$(W \cup U_P)$-connector which is a subpath of~$P$ is trivial or nice.
\end{lemma}
\begin{proof}
Let~$a,a',a''$ be the first three vertices of~$P$ and let~$b'',b',b$ be the last three vertices of~$P$.
Moreover, let~$C$ be the extension of~$P$.
We distinguish three cases.

\textbf{Case 1:} $P$ is a dense (in,out)-connector\textbf{.}
We set~$U_P := \{a',a'',b',b''\}$.
Let~$P'$ be the subpath of~$P$ from~$a''$ to~$b''$.
If~$a'' = b''$ or~$\{a'',b''\}\in E$, then there is no subpath of~$P$ which is a~$(W\cup U_P)$-connector, which implies that~$U_P$ fulfills the desired properties.
Hence, assume that~$P'$ has length at least 2, which implies that~$P'$ is the only~$(W\cup U_P)$-connector that is a subpath of~$P$.
Moreover, note that~$P'$ is also a dense (in,out)-connector.
If~$P'$ is trivial, $U_P$ fulfills the desired properties.
Hence, in the following assume that~$P'$ in non-trivial.
This implies that~$(b'',a'')\notin A$, since~$P'$ is a dense (in,out)-connector.
We now show that~$P'$ is nice.
Let~$C'$ be the extension of~$P'$~and let~$w \in V \setminus V(C')$.
First, we show that if~$(w,a'')\in A$ and~$(w,b'')\in A$, then there is no dense~$(w,b'')$-path in~$D$ that avoids the edges of~$P'$.
Afterwards, we show that if~$(a'',w)\in A$ and~$(b'',w)\in A$, then there is no dense~$(a'',w)$-path in~$D$ that avoids the edges of~$P'$.

Suppose that~$(w,a'')\in A$ and~$(w,b'')\in A$.
Since~$P$ is a dense path in~$D$ and~$P'$ has length at least 4, $A$ does not contain the arc~$(b',a'')$ (since $(a'',b')\in A$).
Hence, if~$w \in \{b'\} \cup V_{b'}$, then there is no dense\nnew{~$(w,a'')$}-path in~$D$, since each such path would contain the vertex~$b'$.
For each other vertex of~$V \setminus V(C')$, each path in~$G$ from~$w$ to~$b''$ either traverses the entire path~$P'$, or goes over the vertex~$b$. 
This is due to the fact that each vertex in each connector has exactly two neighbors in~$V^*$. 
Thus, there is no dense~$(w,b'')$-path in~$D$ that avoids the edges of~$P'$, since such a path would contain the vertex~$b$, but $A$ does not contain the arc~$(b,b'')$ (since $(b'',b)\in A$).

Now suppose that~$(a'',w)\in A$ and~$(b'',w)\in A$.
First, we show that~$w\notin \{a'\} \cup V_{a'}$.
Since~$P$ is a dense path in~$D$ and~$P'$ has length at least 4, $A$ does not contain the arc~$(b'',a')$ (since $(a',b'')\in A$).
Note that this implies that~$(b'',w)\notin A$.
This is due to the assumption that~$D[V(C)]$ is realizable and the fact that each~$(b'',w)$-path in the tree~$G[V(C)]$ contains the vertex~$a'$.
Hence, $w \notin \{a'\} \cup V_{a'}$.
For each other vertex of~$V \setminus V(C')$, each path in~$G$ from~$a''$ to~$w$ either traverses the entire path~$P'$, or goes over the vertex~$a$. 
This is due to the fact that each vertex in each connector has exactly two neighbors in~$V^*$. 
Thus, there is no dense\nnew{~$(a'',w)$}-path in~$D$ that avoids the edges of~$P'$, since such a path would contain the vertex~$a$, but $A$ does not contain the arc~$(a'',a)$ (since $(a,a'')\in A$).

Thus, $P'$ is nice, which implies that~$U_P$ fulfills the desired properties.

\textbf{Case 2:} $P$ is not a dense (in,out)-connector\textbf{.}
Note that, in this case, we can check in polynomial time, whether~$P$ is nice.
If this is the case, the set~$U_P := \emptyset$ trivially fulfills the desired property.
Hence, assume in the following that~$P$ is not nice.
This implies that there is some vertex~$p\in V(P)$ with (i)~$(a,p)\in A$ and~$(b,p)\in A$, or (ii)~$(p,a)\in A$ and~$(p,b)\in A$.
Due to symmetry, assume that~$(a,p)\in A$ and~$(b,p)\in A$.
Moreover, let~$p'$ be the predecessor of~$p$ in~$P$.
This then implies that the unique~$(a,p')$-path in~$P$ and the unique~$(b,p)$-path in~$P$  are dense.
Hence, there are at most two subpaths of~$P$ that are~$(W\cup \{p,p'\})$-connectors.
Namely, the subpath~$P_1$ from~$a$ to~$p'$ and the subpath~$P_2$ from~$b$ to~$p$.
By the above, each of these at most two resulting connectors is trivial or a dense (in,out)-connector.
If~$P_i$ exists and is non-trivial, let~$U_{P_i} := \emptyset$.
Otherwise, let~$U_{P_i}$ be as defined in the case for dense (in,out)-connectors.
Thus, by setting~$U_P := \{p,p'\} \cup U_{P_1} \cup U_{P_2}$, we obtain at most two~$(W \cup U_P)$-connectors which are subpaths of~$P$ and both of these subpaths are trivial or nice.
Moreover, $U_P$ has size at most~$10$, since $U_{P_i}$ has size at most~$4$ as shown above.
\end{proof}

Based on all these results, we now show the main result of this section.
\begin{proposition}\label{compute w star}
In polynomial time, we can detect that~$D$ is a not realizable, or compute a set~$W^*$ with~$X^*\subseteq W^* \subseteq V^*$ of size~$\Oh(\fes)$, such that each~$W^*$-connector is nice.
\end{proposition}
\begin{proof}
First, we describe how to obtain a set~$W_1$ with~$X^*\subseteq W_1 \subseteq V^*$, such that each~$W_1$-connector is trivial or robust.
Compute the set of~$X^*$-connectors, initialize~$W_1:=X^*$, and iterate over each non-trivial~$X^*$-connector~$P$.
Compute the set~$U_P$ of vertices of~$P$ according to~\Cref{make robust connectors}.
Note that the set of~$W_1$-connectors and the set of~$(W_1\cup U_P)$-connectors differ only by a constant number of subpaths of~$P$.
Then, each subpath of~$P$ which is a~$(W_1\cup U_P)$-connector is trivial or robust due to~\Cref{make robust connectors}.
Moreover, for all other~$(W_1\cup U_P)$-connectors, the properties of being trivial or being robust are preserved due to~\Cref{robust property is independent}.
Now, add all vertices of~$U_P$ to~$W_1$ and proceed with the next~$X^*$-connector (which is also a~$W_1$-connector).
Hence, after the iteration over all~$X^*$ connectors, the resulting set~$W_1$ preserves the stated property that each~$W_1$-connector is trivial or robust.
Moreover, since~$U_P$ has constant size for each~$X^*$-connector, and there are $\Oh(\fes)$~$X^*$-connectors, $W_1$ has size~$\Oh(\fes)$.
This also implies that there are $\Oh(\fes)$~$W_1$-connectors.

Second, we describe how to obtain a set~$W_2$ with~$W_1\subseteq W_2 \subseteq V^*$, such that each~$W_2$-connector is trivial or nice.
We essentially do the same as in the construction of~$W_1$.
Compute the set of~$W_1$-connectors, initialize~$W_2:=W_1$ and iterate over each non-trivial~$W_1$-connector~$P$.
Note that by definition of~$W_1$, $P$ is robust.
Let~$C$ be the extension of~$P$.
Check whether~$D[V(C)]$ is realizable.
Note that this can be done in polynomial time, since~$D[V(C)]$ is a tree.
If this is not the case, then output that~$D$ is not realizable.
This is correct due to~\Cref{robust has internal solution}.
If~$D[V(C)]$ is realizable, proceed as follows.
Compute the set~$U_P$ of vertices of~$P$ according to~\Cref{make nice connectors}.
Note that the set of~$W_2$-connectors and the set of~$(W_2\cup U_P)$-connectors differ only by a constant number of subpaths of~$P$.
Then, each subpath of~$P$ which is a~$(W_2\cup U_P)$-connector is trivial or nice due to~\Cref{make nice connectors}.
Moreover, for all other~$(W_2\cup U_P)$-connectors, the properties of being trivial or being nice are preserved due to~\Cref{nice property is independent}.
Now, add all vertices of~$U_P$ to~$W_2$ and proceed with the next~$W_1$-connector (which is also a~$W_2$-connector).
Hence, after the iteration over all~$W_1$ connectors, the resulting set~$W_2$ preserves the stated property that each~$W_2$-connector is trivial or nice.
Moreover, since~$U_P$ has constant size for each~$W_1$-connector, and there are $\Oh(\fes)$~$W_1$-connectors, $W_2$ has size~$\Oh(\fes)$.

Finally, the set~$W^*$ is obtained from~$W_2$ by adding all vertices of trivial~$W_2$-connectors.
Since each trivial connector contains a constant number of vertices, $W^*$ has size~$\Oh(\fes)$ and each~$W^*$-connector is nice due to~\Cref{nice property is independent}.
Note that all operations described run in polynomial time.
This thus proves the statement.
\end{proof}

\subsection{Computing labelings for prelabeled nice connectors}

In this subsection, we show the following.

\newcommand{\Ext}{\texttt{Ext}}
\newcommand{\Int}{\texttt{Int}}

\begin{lemma}\label{connector labelings}
Let~$W$ with~$X^*\subseteq W \subseteq V^*$.
Moreover, let~$P$ be a nice~$W$-connector with endpoints~$a$ and~$b$, let~$a'$ be the successor of~$a$ in~$P$, and let~$b'$ be the predecessor of~$b$ in~\tnew{$P$}.
Moreover, let~$\alpha_a^1,\alpha_a^2,\alpha_b^1,\alpha_b^2 \in \{i\cdot 2n \mid i \in \mathbb{N}\}$ be (not necessarily distinct) natural numbers, and let~$C$ denote the extension of~$P$.
In polynomial time, we can compute a set~$L_P$ of $\Oh(1)$~labelings~$\lambda_P$ of~$E(C)$ with~$\lambda_P(\{a,a'\}) = \{\alpha_a^1,\alpha_a^2\}$ and~$\lambda_P(\{b,b'\}) = \{\alpha_b^1,\alpha_b^2\}$, such that if there is a realization~$\lambda^*$ for~$D$ with~$\lambda^*(\{a,a'\}) = \{\alpha_a^1,\alpha_a^2\}$ and~$\lambda^*(\{b,b'\}) = \{\alpha_b^1,\alpha_b^2\}$ and which is minimal on both~$\{a,a'\}$ and~$\{b,b'\}$, then there is a labeling~$\lambda_P\in L_P$ for which~$\lambda^* \ltimes \lambda_P$ \tnew{realizes}~$D$. \end{lemma}
\begin{proof}
Without loss of generality assume that~$\alpha_a^1 \leq \alpha_a^2$ and~$\alpha_b^1 \leq \alpha_b^2$.
Moreover, let~$\Ext := V \setminus (V(C) \setminus \{a,b\})$ denote the \emph{external vertices} of~$C$ and let~$\Int := V(C) \setminus \{a,b\}$ denote the \emph{internal vertices} of~$C$.
\newcommand{\md}{\mathcal{D}_P}To show the statement, we compute a constant number of instances~$\md$ of our problem, where for each~$D_i \in \md$, the solid graph~$G_i$ of~$D_i$ is a tree, and~\nnew{$D_i$} is a supergraph of~$D[V(C)]$.
For each such instance~$D_i$, we then check whether there is a labeling~$\lambda_i$ that realizes~$D_i$ with~$\lambda_i(\{a,a'\}) = \{\alpha_a^1,\alpha_a^2\}$ and~$\lambda_i(\{b,b'\}) = \{\alpha_b^1,\alpha_b^2\}$.
As we will show, if there is a realization~$\lambda^*$ for~$D$ with~$\lambda^*(\{a,a'\}) = \{\alpha_a^1,\alpha_a^2\}$ and~$\lambda^*(\{b,b'\}) = \{\alpha_b^1,\alpha_b^2\}$ which is minimal on both~$\{a,a'\}$ and~$\{b,b'\}$, then there is some~$D_i\in \md$, such that~$\lambda^* \ltimes \widehat{\lambda_i}$ realizes~$D$, where~$\widehat{\lambda_i}$ is the restriction to the edges of~$E(C)$ of an arbitrary realization~$\lambda_i$ of~$D_i$ with $\lambda_i(\{a,a'\}) = \{\alpha_a^1,\alpha_a^2\}$ and~$\lambda_i(\{b,b'\}) = \{\alpha_b^1,\alpha_b^2\}$.

\begin{figure}
\centering
\newcommand{\connector}{
\tikzstyle{k}=[circle,fill=white,draw=black,minimum size=10pt,inner sep=2pt]

\node[k] (x) at (0,0) {$a$};
\node[k] (u) at (0,1) {$a'$};
\node[k] (u2) at (0,2) {};

\node[k] (y) at (5,0) {$b$};
\node[k] (v) at (5,1) {$b'$};
\node[k] (v2) at (5,2) {};

\draw [thick] (x) -- (u) node [midway, fill=white] {$\alpha_a^1,\alpha_a^2$};
\draw [thick] (y) -- (v) node [midway, fill=white] {$\alpha_b^1,\alpha_b^2$};

\draw [thick] (u) -- (u2);
\draw [thick] (v) -- (v2);

\draw [thick, bend left=100,decorate, decoration={snake}] (u2) to (v2);
}

\begin{tikzpicture}[scale=1.5]
\connector

\newcommand{\myred}{blue}

\node[\myred,k] (fal) at ($(u)+ (-1.2,0)$) {$f_{a'}^{\leq}$};
\node[\myred,k] (val) at ($(x)+ (-1.2,0)$) {$v_{a}^{\leq}$};
\node[\myred,k] (fbl) at ($(v)+ (-1.2,0)$) {$f_{b'}^{\leq}$};
\node[\myred,k] (vbl) at ($(y)+ (-1.2,0)$) {$v_{b}^{\leq}$};

\node[\myred,k] (fag) at ($(u)+ (1.2,0)$) {$f_{a'}^{\geq}$};
\node[\myred,k] (vag) at ($(x)+ (1.2,0)$) {$v_{a}^{\geq}$};
\node[\myred,k] (fbg) at ($(v)+ (1.2,0)$) {$f_{b'}^{\geq}$};
\node[\myred,k] (vbg) at ($(y)+ (1.2,0)$) {$v_{b}^{\geq}$};

\draw [\myred,thick] (fal) -- (u) node [midway, fill=white] {$\alpha_a^1$};
\draw [\myred,thick] (fag) -- (u) node [midway, fill=white] {$\alpha_a^2$};
\draw [\myred,thick] (fbl) -- (v) node [midway, fill=white] {$\alpha_b^1$};
\draw [\myred,thick] (fbg) -- (v) node [midway, fill=white] {$\alpha_b^2$};

\draw [\myred,thick] (val) -- (x) node [midway, fill=white] {$\alpha_a^1$};
\draw [\myred,thick] (vag) -- (x) node [midway, fill=white] {$\alpha_a^2$};
\draw [\myred,thick] (vbl) -- (y) node [midway, fill=white] {$\alpha_b^1$};
\draw [\myred,thick] (vbg) -- (y) node [midway, fill=white] {$\alpha_b^2$};

\end{tikzpicture}
\caption{An visualization of the solid graph and the labeling for the additional solid edges for the the auxiliary instance~$D_\pa$ with the consistent type~$\pa = (J,A_J)$ with~$J = \{v_c^{\veebar} \mid c\in \{a,b\}, \veebar \in \{\leq, \geq\}\}$.}
\label{consistent instance}
\end{figure}
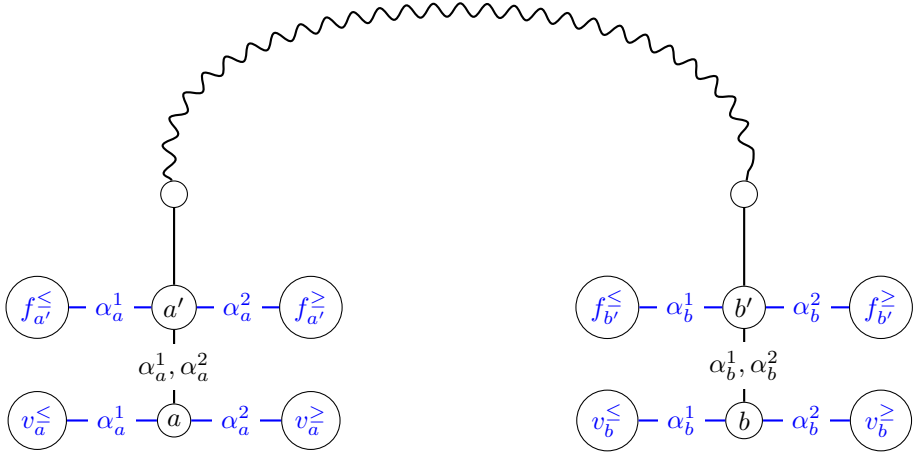

To properly define the instances of~$\md$, we first compute four sets~$V_{c'}^<,V_{c'}^\leq,V_{c'}^\geq,V_{c'}^>$ of vertices of~$\Int$ for each vertex~$c\in \{a,b\}$.
Let~$c\in \{a,b\}$.
We define the set~$V_{c'}^<$ ($V_{c'}^>$) as the set of all vertices~$\vin$ of~$\Int$ for which~$(\vin,c)\in A$ (for which~$(c,\vin)\in A$).
Moreover, we define the set~$V_{c'}^\leq$ ($V_{c'}^\geq$) as a smallest non-empty set of vertices of~$\Int$, such that there is at least one vertex~$\vout \in (V \setminus V(C)) \cup \{a,b\}$ for which~$V_{c'}^\leq$ ($V_{c'}^\geq$) is exactly the set of vertices~$\vin$ of~$\Int$ fulfilling (i)~$(\vin,\vout)\in A$ ($(\vout,\vin)\in A$) and (ii)~$\{c,c'\}$ is the edge of~$\{\{a,a'\}, \{b,b'\}\}$ outputted by the algorithm behind~\Cref{compute unique entrance point} for the arc~$(\vin,\vout)$ (for the arc~$(\vout,\vin)$).

\begin{claim}\label{properties leq sets}
Assume that~$D$ is realizable.
It holds that~$V_{c'}^\leq \subseteq V_{c'}^<$, $V_{c'}^\geq \subseteq V_{c'}^>$, each of these four sets contains~$c'$, and~$V_{c'}^< \cap V_{c'}^> = \{c'\}$.
Moreover, let~$\vout\in \Ext$ for which there is at least one~$\vin \in \Int$ with $(\vin,\vout)\in A$ (with $(\vout,\vin)\in A$), such that the edge~$\{c,c'\}$ is the edge of~$\{\{a,a'\}, \{b,b'\}\}$ outputted by the algorithm behind~\Cref{compute unique entrance point} for the arc~$(\vin,\vout)$ (for the arc~$(\vout,\vin)$).
Then, one of~$V_{c'}^<$ or~$V_{c'}^\leq$ (one of~$V_{c'}^>$ or~$V_{c'}^\geq$) is exactly the set of vertices~$\vin$ of~$\Int$ for which~$\{c,c'\}$ is the edge of~$\{\{a,a'\}, \{b,b'\}\}$ outputted by the algorithm behind~\Cref{compute unique entrance point} for the arc~$(\vin,\vout)$ (for the arc~$(\vout,\vin)$).
\end{claim}
\begin{claimproof}
First, we show that~$V_{c'}^\leq \subseteq V_{c'}^<$. 
Assume towards a contradiction that this is not the case and let~$\vin$ be an arbitrary vertex of~$V_{c'}^\leq \setminus V_{c'}^<$. 
By definition of~$V_{c'}^\leq$, there is some vertex~$\vout \in \Ext$ such that~$(\vin,\vout)\in A$ and in each realization for~$D$, each temporal~$(\vin,\vout)$-path traverses the edge~$\{c,c'\}$.
This then implies that there is a temporal~$(\vin,c)$-path, which implies that~$(\vin,c)\in A$.
This contradicts the fact that~$\vin\notin V_{c'}^<$.
Thus, $V_{c'}^\leq \subseteq V_{c'}^<$.
Similarly, we can show that~$V_{c'}^\geq \subseteq V_{c'}^>$.

Moreover, note that~$c' \in V_{c'}^< \cap V_{c'}^>$, since~$\{c,c'\}$ is a solid edge in~$G$.
Additionally, there is no vertex~$\vin \neq c'$ that is contained in~$V_{c'}^< \cap V_{c'}^>$, as otherwise, by definition of~$V_{c'}^<$ and~$V_{c'}^>$, \tnew{$A$~contains} both arcs~$(\vin,c)$ and~$(c,\vin)$, which is not possible by definition of a connector.

Now, let~$\vout\in \Ext$ for which there is at least one~$\vin \in \Int$ with $(\vin,\vout)\in A$ (with $(\vout,\vin)\in A$), such that the edge~$\{c,c'\}$ is the edge of~$\{\{a,a'\}, \{b,b'\}\}$ outputted by the algorithm behind~\Cref{compute unique entrance point} for the arc~$(\vin,\vout)$ (for the arc~$(\vout,\vin)$).
We show that one of~$V_{c'}^<$ or~$V_{c'}^\leq$ (one of~$V_{c'}^>$ or~$V_{c'}^\geq$) is exactly the set of vertices~$\vin$ of~$\Int$ for which~$\{c,c'\}$ is the edge of~$\{\{a,a'\}, \{b,b'\}\}$ outputted by the algorithm behind~\Cref{compute unique entrance point} for the arc~$(\vin,\vout)$ (for the arc~$(\vout,\vin)$).
Let~$M$ denote \tnew{the set of vertices}~$\vin$ of~$\Int$ for which~$\{c,c'\}$ is the edge of~$\{\{a,a'\}, \{b,b'\}\}$ outputted by the algorithm behind~\Cref{compute unique entrance point} for the arc~$(\vin,\vout)$ (for the arc~$(\vout,\vin)$).
Assume towards a contradiction that~$M \notin \{V_{c'}^<,V_{c'}^\leq\}$ ($M \notin \{V_{c'}^>,V_{c'}^\geq\}$).
Analogously to the proof that~$V_{c'}^\leq \subseteq V_{c'}^<$ ($V_{c'}^\geq \subseteq V_{c'}^>$), one can show that~$M \subseteq V_{c'}^<$ ($M \subseteq V_{c'}^>$).
Thus, $M$ is a proper subset of~$V_{c'}^<$ ($V_{c'}^>$).
This implies that~$V_{c'}^\leq$ is a proper subset of~$V_{c'}^<$ ($V_{c'}^\geq$ is a proper subset of~$V_{c'}^>$).
Hence, the edge~$\{c,c'\}$ needs to be traversed at at least three different time steps to ensure that all three distinct sets~$V_{c'}^<$, $V_{c'}^\leq$, and~$M$ can reach distinct external vertices (all three distinct sets~$V_{c'}^>$, $V_{c'}^\geq$, and~$M$ can be reached from distinct external vertices).
This contradicts the fact that edge~$\{c,c'\}$ receives at most two labels in each realization for~$D$.
\end{claimproof}
Note that the properties described in~\Cref{properties leq sets} can be checked in polynomial time.
Hence, if at least one of the properties does not hold, we can correctly detect in polynomial time that~$D$ is not realizable.
If this is the case, we define~$L_P$ as the empty set, which is correct.

Hence, assume in the following that all described properties of \Cref{properties leq sets} hold.

We now describe how to define the instances of~$\md$.
We define four \emph{potential vertices}~$J:=\{v_c^{\veebar} \mid c\in \{a,b\}, \veebar \in \{\leq, \geq\}\}$.
Moreover, we define a set of three \emph{potential arcs}~$A_J := \{(v_a^{\leq},b),(a,v_b^{\geq}),(v_a^{\leq},v_b^{\geq})\}$.
Let~$Q \subseteq J$ and let~$A_Q \subseteq A_J$.
We call~$(Q,A_Q)$ a~\emph{consistent type} if~$Q \cup \{a,b\}$ contains all endpoints of the arcs of~$A_Q$.

Let~$A_Q\subseteq A_J$.
We define the instance~$D_{(J,A_Q)}$ as follows (see~\Cref{consistent instance}):
the vertex set of~$D_{(J,A_Q)}$ is~$V(C) \cup J \cup \{f_{c'}^{\veebar} \mid c \in\{a,b\}, \veebar \in \{\leq, \geq\}\}$ and

\begin{enumerate}
\item\label{item-induced arcs} $D_{(J,A_Q)}$ contains all arcs from~$A(V(C))$,
\item\label{item-dense arcs} if~$(a,b)$ is an arc of~$D$ (that is, if~$P$ is a dense~(in,out)-connector), $D_{(J,A_Q)}$ contains the arcs~$(f_{a'}^\leq, b)$, $(a, f_{b'}^\geq)$, and~$(f_{a'}^\leq, f_{b'}^\geq)$,

\item\label{item-solid edges} $D_{(J,A_Q)}$ contains the arcs~$(v_{c}^{\veebar},c),(c,v_{c}^{\veebar}),(f_{c'}^{\veebar},c')$ and~$(c,f_{c'}^{\veebar})$ (that is, the solid edges~$\{c, v_{c}^{\veebar}\}$ and~$\{c',f_{c'}^{\veebar}\}$) for each~$c\in \{a,b\}$ and each~$\veebar \in \{\leq, \geq\}$,
\item\label{item-new vert connections} $D_{(J,A_Q)}$ contains the arcs~$(f_{c'}^\leq, f_{c'}^\geq)$ and~$(v_c^\leq, v_c^\geq)$ for each~$c\in \{a,b\}$,
\item\label{item-forced special} $D_{(J,A_Q)}$ contains the arcs~$(c',v_{c}^\geq)$, $(c,f_{c'}^\geq)$, $(v_{c}^\leq,c')$, and~$(f_{c'}^\leq,c)$ for each~$c\in \{a,b\}$,

\item\label{item-early in to end} if~$(v_a^\leq,b) \in A_Q$, $D_{(J,A_Q)}$ contains all arcs from~$\{v_a^\leq,f_{a'}^\geq\} \times \{b, f_{b'}^\geq\}$, 
\item\label{item-early out from start} if~$(a,v_b^\geq) \in A_Q$, $D_{(J,A_Q)}$ contains all arcs from~$\{a, f_{a'}^\leq\} \times \{v_{b}^\geq, f_{b'}^\leq\}$, 
\item\label{item-early in and early out} if~$(v_a^\leq,v_b^\geq) \in A_Q$, $D_{(J,A_Q)}$ contains all arcs from~$\{v_a^\leq,f_{a'}^\geq\} \times \{v_{b}^\geq, f_{b'}^\leq\}$,

\item\label{item-early in} for each~$c\in \{a,b\}$, $D_{(J,A_Q)}$ contains the arc~$(\vin, f_{c'}^\geq)$ for each vertex~$\vin\in V_{c'}^<$,
\item\label{item-late out} for each~$c\in \{a,b\}$, $D_{(J,A_Q)}$ contains the arc~$(f_{c'}^\leq,\vin)$ for each vertex~$\vin\in V_{c'}^>$, 
\item\label{item-late in} for each~$c\in \{a,b\}$, $D_{(J,A_Q)}$ contains the arcs~$(\vin, f_{c'}^\leq)$, $(\vin, f_{c'}^\geq)$, and~$(\vin, v_{c}^\geq)$ for each vertex~$\vin \in V_{c'}^\leq$, and
\item\label{item-early out} for each~$c\in \{a,b\}$, $D_{(J,A_Q)}$ contains the arcs~$(f_{c'}^\leq,\vin)$, $(f_{c'}^\geq,\vin)$, and~$(v_{c}^\leq,\vin)$ for each vertex~$\vin\in V_{c'}^\geq$.
\end{enumerate}

Note that $D_{(J,A_Q)}$ has a tree as solid graph, since (i)~we did not add any new arcs between vertices of~$V(C)$, (ii)~we added at most one arc between any two vertices of~$V(D_{(J,A_Q)}) \setminus V(C)$, and (iii)~due to~\Cref{properties leq sets}, for each~$c\in \{a,b\}$, $V_{c'}^< \cap V_{c'}^> = \{c'\}$.
The latter implies that the solid edges defined in~\Cref{item-solid edges} are the only solid edge between any vertex of~$V(D_{(J,A_Q)}) \setminus V(C)$ and any vertex of~$V(C)$.

For each consistent type~$\pa=(Q,A_Q)$, we define the instance~$D_\pa$ by~$D_\pa := D_{(J,A_Q)}[V_\pa]$, where~$V_\pa := V(C) \cup Q \cup \{f_{c'}^{\leq} \mid v_c^{\geq} \in Q\} \cup \{f_{c'}^{\geq} \mid v_c^{\leq} \in Q\}$.
Since~$D_{(J,A_Q)}$ has a tree as solid graph, this implies that~$D_\pa$ also has a tree as solid graph.

We now define the set~$L_P$ of labelings for the edges of~$C$.
Initially, $L_P$ is the empty set.
Afterwards, we iterate over each consistent type~$\pa=(Q,A_Q)$.
For each such pair~$\pa$, we check whether there is a realization~$\lambda_\pa'$ for~$D_\pa$ with~$\lambda_\pa'(\{a,a'\}) = \{\alpha_a^1,\alpha_a^2\}$ and~$\lambda_\pa'(\{b,b'\}) = \{\alpha_b^1,\alpha_b^2\}$  which is minimal on both~$\{a,a'\}$ and~$\{b,b'\}$.
If this is the case, we add the labeling~$\lambda_\pa$ to~$L_P$, where~$\lambda_\pa$ is the restriction of~$\lambda_\pa'$ to the edges of~$C$.
This completes the construction of~$L_P$.

Note that~$L_P$ has a constant size, since there are at most~$2^4$ options for~$Q$ and at most~$2^3$ options for~$A_Q$.
Moreover, each such labeling can be computed in polynomial time, since the initial computation took polynomial time and for each consistent type~$\pa$, we can check in polynomial time whether there is such a desired realization~$\lambda_\pa'$ for~$D_\pa$ due to~\Cref{cor:prelabeled_tree}.

It thus remains to show that if there is a realization~$\lambda^*$ for~$D$ with~$\lambda^*(\{a,a'\}) = \{\alpha_a^1,\alpha_a^2\}$ and~$\lambda^*(\{b,b'\}) = \{\alpha_b^1,\alpha_b^2\}$  which is minimal on both~$\{a,a'\}$ and~$\{b,b'\}$, then there is a labeling~$\lambda_P \in L_P$, such that~$\lambda^*\ltimes \lambda_P$ realizes~$D$.
Assume that there is a minimum realization~$\lambda^*$ for~$D$ with~$\lambda^*(\{a,a'\}) = \{\alpha_a^1,\alpha_a^2\}$ and~$\lambda^*(\{b,b'\}) = \{\alpha_b^1,\alpha_b^2\}$  which is minimal on both~$\{a,a'\}$ and~$\{b,b'\}$, as otherwise, the statement trivially holds.
Our goal is to show that there is a consistent type $\pa=(Q,A_Q)$ such that (i)~there is a sought realization for~$D_\pa$ (that is, there is a labeling~$\lambda_\pa\in L_P$) and (ii)~$\lambda^*\ltimes \lambda_\pa$ realizes~$D$.

We are now ready to define the consistent type~$\pa=(Q,A_Q)$ for~$\lambda^*$.
First, we define~$Q$.
For each~$c\in \{a,b\}$ where~$\alpha_c^1 < \alpha_c^2$, we do the following:
\begin{itemize}
\item We add the vertex~$v^{\leq}_c$ to~$Q$ if there is an arc~$(u,w)\in A$, such that for each temporal $(u,w)$-path under~$\lambda^*$, the edge~$\{c,c'\}$ is traversed from~$c$ to~$c'$ at time~$\alpha_c^2$.
\item Similarly, we add the vertex~$v^{\geq}_c$ to~$Q$ if there is an arc~$(u,w)\in A$, such that for each temporal $(u,w)$-path under~$\lambda^*$, the edge~$\{c,c'\}$ is traversed from~$c'$ to~$c$ at time~$\alpha_c^2$.
\end{itemize}
Next, we define~$A_Q$.
\begin{itemize}
\item We add the arc~$(v^{\leq}_a,b)$ to~$A_Q$ if~$\alpha_a^1 < \alpha_a^2$ and there is an arc~$(u,w)\in A$, such that for each temporal $(u,w)$-path under~$\lambda^*$, the edge~$\{a,a'\}$ is traversed from~$a$ to~$a'$ at time~$\alpha_a^2$ and the edge~$\{b,b'\}$ is traversed at a later time.
\item We add the arc~$(a, v^{\geq}_b)$ to~$A_Q$ if~$\alpha_b^1 < \alpha_b^2$ and there is an arc~$(u,w)\in A$, such that for each temporal $(u,w)$-path under~$\lambda^*$, the edge~$\{b,b'\}$ is traversed from~$b'$ to~$b$ at time~$\alpha_b^1$ and the edge~$\{a,a'\}$ is traversed at an earlier time.
\item We add the arc~$(v^\leq_a, v^{\geq}_b)$ to~$A_Q$ if~$\alpha_a^1 < \alpha_a^2 < \alpha_b^1 < \alpha_b^2$ and there is an arc~$(u,w)\in A$, such that for each temporal $(u,w)$-path under~$\lambda^*$, the edge~$\{a,a'\}$ is traversed from~$a$ to~$a'$ at time~$\alpha_a^2$ and the edge~$\{b,b'\}$ is traversed from~$b'$ to~$b$ at time~$\alpha_b^1$.
\end{itemize}
This completes the definition of~$\pa$.
\begin{claim}\label{claim reachable late is exactly geq set}
Let~$c\in \{a,b\}$, and let~$\vout\in\Ext$ such that~$(\vout, c')\in A$ ($(c',\vout)\in A$) and each temporal~$(\vout,c')$-path (each temporal~$(c',\vout)$-path) under~$\lambda^*$ traverses the edge~$\{c,c'\}$ from~$c$ to~$c'$ (from~$c'$ to~$c$) at time~$\alpha_c^2$.
Then, $V_{c'}^\geq$ ($V_{c'}^\leq$) is exactly the set of vertices~$\vin$ of~$\Int$ with~$(\vout,\vin) \in A$ ($(\vin,\vout) \in A$) that can be reached by~$\vout$ (that can reach~$\vout$) only via the edge~$\{c,c'\}$.
\end{claim}
\begin{claimproof}
We show that~$V_{c'}^\geq$ is exactly the set of vertices of~$\Int$ that~$\vout$ can reach via any temporal path under~$\lambda^*$ that traverses the edge~$\{c,c'\}$. (The case for~$V_{c'}^\leq$ can then be shown analogously.)
Assume towards a contradiction that this is not the case.
Hence, due to~\Cref{properties leq sets}, $V_{c'}^\geq \neq V_{c'}^>$ and $V_{c'}^>$ is exactly the set of vertices of~$\Int$ that~$\vout$ can reach via any temporal path under~$\lambda^*$ that traverses the edge~$\{c,c'\}$.
Since~$V_{c'}^\geq \neq V_{c'}^>$, there is a vertex~$\vout_2 \in \Ext$ for which~$V_{c'}^>$ is exactly the set of vertices of~$\Int$ that~$\vout_2$ can reach via any temporal path under~$\lambda^*$ that traverses the edge~$\{c,c'\}$.
Since~$V_{c'}^\geq$ is non-empty, there is at least one temporal path that starts at~$\vout_2$ and traverses the edge~$\{c,c'\}$ at time at most~$\alpha_c^2$ from~$c$ to~$c'$.
Such a path can then be extended to reach all vertices that~$\vout$ can reach via temporal paths that go over~$\{c,c'\}$, which would imply~$V_{c'}^\geq = V_{c'}^>$, a contradiction.
\end{claimproof}

We now show the statement in two steps.

\begin{claim}
There is a realization~$\lambda_\pa'$ for~$D_\pa$ with~$\lambda_\pa'(\{a,a'\}) = \{\alpha_a^1,\alpha_a^2\}$ and~$\lambda_\pa'(\{b,b'\}) = \{\alpha_b^1,\alpha_b^2\}$ which is minimal on both~$\{a,a'\}$ and~$\{b,b'\}$.
\end{claim}
\begin{claimproof}
Let~$G_\pa$ be the solid graph of~$D_\pa$.
We define the a labeling~$\lambda_\pa'$ for the edges of~$G_\pa$ as follows.
For each edge~$e$ of~$E(C)$, we set~$\lambda_\pa'(e) := \lambda^*(e)$ and for each~$c\in \{a,b\}$, we set
\begin{itemize}
\item $\lambda_\pa'(\{v_c^\leq, c\}) := \alpha_c^1$ and~$\lambda_\pa'(\{f_{c'}^\geq, c'\}) := \alpha_c^2$ if~$v_c^\leq \in Q$ and
\item $\lambda_\pa'(\{v_c^\geq, c\}) := \alpha_c^2$ and~$\lambda_\pa'(\{f_{c'}^\leq, c'\}) := \alpha_c^1$ if~$v_c^\geq \in Q$.
\end{itemize}
This completes the definition of~$\lambda_\pa'$.
Note that~$\lambda_\pa'(\{a,a'\}) = \{\alpha_a^1,\alpha_a^2\}$ and~$\lambda_\pa'(\{b,b'\}) = \{\alpha_b^1,\alpha_b^2\}$.
Hence, to show the statement, it remains to show that~$\lambda_\pa'$ realizes~$D_\pa$ and that for each~$c\in \{a,b\}$ with~$\alpha_c^1 < \alpha_c^2$, the edge~$\{c,c'\}$ is a special edge under~$D_\pa$.
The latter follows by construction of~$\pa$:
If for~$c\in \{a,b\}$, the realization~$\lambda^*$ of~$D$ uses two distinct labels~$\alpha_c^1< \alpha_c^2$ and is minimal on both~$\{a,a'\}$ and~$\{b,b'\}$, then there is some vertex pair~$(u,w)$ with~$(u,w)\in A$, such that each~$(u,w)$-path under~$\lambda^*$ traverses the edge~$\{c,c'\}$ at time~$\alpha_c^2$, as otherwise, the label~$\alpha_c^2$ can be removed from~$\{c,c'\}$ in~$\lambda^*$.
This then implies that~$v_c^\leq \in Q$ or~$v_c^\geq \in Q$.
Due to symmetry assume that~$v_c^\leq \in Q$.
By construction, $D_\pa$ contains the arcs~$(v_c^\leq, c')$ and~$(c, f_{c'}^\geq)$ but not the arc~$(v_c^\leq, f_{c'}^\geq)$.
Consequently, $\{c,c'\}$ is a special bridge in~$D_\pa$, which implies that in each minimum realization for~$D_\pa$, edge~$\{c,c'\}$ receives two labels.

We now show that~$\lambda_\pa'$ realizes~$D_\pa$.

\begin{enumerate}[a)]
\item Note that for each~$c\in \{a,b\}$, each temporal path under~$\lambda_\pa'$ that reaches~$c'$ at a time step smaller than~$\alpha_c^1$ can be extended at the end to reach the vertices~$f_{c'}^\leq$, $f_{c'}^\geq$, $v_{c}^\geq$, and~$c$ (if they exist) by the subpaths with label sequence~$(\alpha_c^1)$, $(\alpha_c^2)$, $(\alpha_c^1,\alpha_c^2)$, and~$(\alpha_c^1)$, respectively, but cannot be extended to reach the vertex~$v_c^\leq$.
\item Similarly, each temporal path under~$\lambda_\pa'$ that reaches~$c'$ at a time step exactly~$\alpha_c^1$ can be extended at the end to reach the vertices~$f_{c'}^\geq$ and~$c$ by the respective edges with label~$\alpha_c^2$, but cannot be extended to reach~$f_{c'}^\leq$, $v_{c}^\leq$, or~$v_{c}^\geq$.
\item All temporal path under~$\lambda_\pa'$ that reaches~$c'$ at a time step larger than~$\alpha_c^1$ cannot be extended to reach any vertex of~$f_{c'}^\leq$, $f_{c'}^\geq$, $v_{c}^\leq$, $v_{c}^\geq$, and~$c$, since each edge incident with~$\{c,c'\}$ has all labels either smaller-equal~$\alpha_c^1$ or larger-equal~$\alpha_c^2$ (as otherwise, $\{c,c'\}$ would be part of a triangle in~$G$), and each subpath leading to any of the vertices of~$f_{c'}^\leq$, $f_{c'}^\geq$, $v_{c}^\geq$, and~$c$ uses an edge first of label at most~$\alpha_c^2$. 
\end{enumerate}

The similar property also holds for the temporal paths that start from vertex~$c'$ either early or late.
 
\begin{enumerate}[a)]
\setcounter{enumi}{3}
\item Each temporal path under~$\lambda_\pa'$ that leaves~$c'$ at a time step larger than~$\alpha_c^2$ can be extended at the start to begin at any of the vertices~$f_{c'}^\leq$, $f_{c'}^\geq$, $v_{c}^\leq$, and~$c$ (if they exist), but cannot be extended to begin at vertex~$v_c^\geq$.
\item Similarly, each temporal path under~$\lambda_\pa'$ that leaves~$c'$ at a time step exactly~$\alpha_c^2$ can be extended at the start to begin with any of the vertices~$f_{c'}^\leq$ and~$c$, but cannot be extended to start from any of the vertices~$f_{c'}^\geq$, $v_{c}^\leq$, or~$v_{c}^\geq$.
\item All temporal path under~$\lambda_\pa'$ that leave~$c'$ at a time step smaller than~$\alpha_c^2$ cannot be extended at the start to begin from any of the vertices~$f_{c'}^\leq$, $f_{c'}^\geq$, $v_{c}^\leq$, $v_{c}^\geq$, and~$c$, since each edge incident with~$\{c,c'\}$ has all labels either smaller-equal~$\alpha_c^1$ or larger-equal~$\alpha_c^2$ (as otherwise, $\{c,c'\}$ would be part of a triangle in~$G$). 
\end{enumerate}

We now show that each arc of~$D_\pa'$ is realized.
To this end, we may implicitly assume that all vertices we talk about exist in~$D_\pa'$. 
\begin{itemize}
\item Note that all arcs described in~\Cref{item-induced arcs} are realized since~$\lambda^*$ realizes~$D$, $P$ is a robust connector, and~$D[V(C)] = D_\pa[V(C)]$.
\item Since~$\lambda^*$ realizes~$D$, if~$(a,b)\in A$, then there is a temporal path under~$\lambda^*$ (and thus~$\lambda_\pa'$) from~$a$ to~$b$ that only uses edges in~$C$.
The subpath~$R$ from~$a'$ to~$b'$ of this path thus leaves~$a'$ with a label of at least~$\alpha_a^2$ and reaches~$b'$ at the latest with a label smaller than~$\alpha_b^1$.
Hence, by the above argumentation, $R$~can (simultaneously) be extended to start with any of~$a$ or~$f_{a'}^\leq$, and end with any of~$b$ or~$f_{b'}^\geq$.
Thus, for each arc~$(u,w)$ described in~\Cref{item-dense arcs}, there is a temporal path from~$u$ to~$w$ in~$D_\pa$ under~$\lambda_\pa'$.
\item All solid edges described in~\Cref{item-solid edges} receive a label under~$\lambda_\pa'$, which implies that also the respective arcs are realized.
\item The arcs of~\Cref{item-new vert connections} are realized by the temporal paths~$(v_c^\leq, c, v_c^\geq)$ and~$(f_{c'}^\leq, c', f_{c'}^\geq)$ with label sequence~$(\alpha_c^1,\alpha_c^2)$, since~$\alpha_c^1 < \alpha_c^2$ for each~$c\in \{a,b\}$ with~$\{v_c^\leq, v_c^\geq\} \subseteq Q$.
\item Similarly, the arcs of~\Cref{item-forced special} are realized by the temporal paths~$(c',c,v_c^\geq)$, $(c,c',f_{c'}^\geq)$, $(v_c^\leq, c, c')$, and~$(f_{c'}^\leq, c', c)$, each with label sequence~$(\alpha_c^1,\alpha_c^2)$, since~$\alpha_c^1 < \alpha_c^2$ for each~$c\in \{a,b\}$ with~$\{v_c^\leq, v_c^\geq\} \subseteq Q$.
\item If~$(v_a^\leq,b)\in A_Q$, then by choice of~$A_Q$, there is some vertex pair~$(u,w)$ with~$(u,w)\in A$, such that each temporal $(u,v)$-path under~$\lambda^*$ traverses the edge~$\{a,a'\}$ from~$a$ to~$a'$ at time~$\alpha_a^2$ and reaches~$b$ at a later time step.
Consider a subpath~$R$ from~$a'$ to~$b'$ of one such~$(u,v)$-path~$R'$. 
Since~$R'$ traversed the edge~$\{a,a'\}$ at time~$\alpha_a^2$, the path~$R$ traverses its first edge in time strictly larger than~$\alpha_c^2$.
Moreover, $R'$ reaches vertex~$b'$ prior to time step~$\alpha_b^2$.
Hence, the path~$R$ can (simultaneously) be extended to start from any vertex of~$\{v_a^\leq,f_{a'}^\geq\}$ and end at any vertex of~$\{b,f_{b'}^\geq\}$.
Thus, all arcs of~\Cref{item-early in to end} are realized.
\item Similarly, one can show that all arcs of~\Cref{item-early out from start} are realized, due to the existence of some vertex pair~$(u,w)$ with~$(u,w)\in A$, such that each temporal $(u,v)$-path under~$\lambda^*$ traverses the edge~$\{b,b'\}$ from~$b'$ to~$b$ at time~$\alpha_b^1$ and visits~$a$ at an earlier time step.
\item Similarly, one can show that all arcs of~\Cref{item-early in and early out} are realized, due to the existence of some vertex pair~$(u,w)$ with~$(u,w)\in A$, such that each temporal $(u,v)$-path under~$\lambda^*$ traverses the edge~$\{a,a'\}$ from~$a$ to~$a'$ at time~$\alpha_a^2$ and also traverses the edge~$\{b,b'\}$ from~$b'$ to~$b$ at time~$\alpha_b^1$.
\item By definition, $V_{c'}^<$ are exactly the vertices of~$\Int$ that have an arc to~$c$ in~$D$.
Let~$\vin \in V_{c'}^<$.
Since~$(\vin,c)\in A$, there is a temporal~$(\vin,c)$-path~$R$ in~$G[V(C)]$ under~$\lambda^*$ and thus~$\lambda_\pa'$.
This path thus reaches~$c'$ at a time strictly smaller than~$\alpha_c^2$.
The subpath from~$\vin$ to~$c'$ can thus be extended to end in~$f_{c'}^\leq$ by traversing the edge~$\{f_{c'}^\leq,c'\}$ at time~$\alpha_c^2$.
Hence, each arc of~\Cref{item-early in} is realized.
\item Similar to the argument that each arc of~\Cref{item-early in} is realized, we can show that each arc of~\Cref{item-late out} is realized.
\item Suppose that~$v_c^\geq \in Q$.By definition of~$Q$, this implies that there is some~$\vout \in \Ext$ for which there is some vertex~$w \in V$ with~$(\vout,w)\in A$, such that each~$(\vout,w)$-path under~$\lambda^*$ traverses the edge~$\{c,c'\}$ from~$c$ to~$c'$ at time~$\alpha_c^2$.
This implies that~$(\vout, c')\in A$, and that each temporal~$(\vout,c')$-path under~$\lambda^*$ traverses the edge~$\{c,c'\}$ from~$c$ to~$c'$ at time~$\alpha_c^2$.

Due to~\Cref{claim reachable late is exactly geq set}, $V_{c'}^\geq$ is exactly the set of vertices~$\vin$ of~$\Int$ with~$(\vout,\vin) \in A$ that can be reached by~$\vout$ only via the edge~$\{c,c'\}$.
Consequently, all arcs specified in~\Cref{item-late in} are realized, since for each~$\vin\in V_{c'}^\geq$ each temporal $(\vout,\vin)$-path under~$\lambda^*$ traverses the edge~$\{c,c'\}$ at time~$\alpha_c^2$ and thus implies the existence of a temporal~$(c',\vin)$-path starting at a time step strictly larger than~$\alpha_c^2$.

\item Similar to the proof that each arc of~\Cref{item-late in} is realized, we can show that each arc of~\Cref{item-early out} is realized.
\end{itemize}
Note that only the \tnew{arcs described above} are realized.
This is mainly due to the fact that~$D[V(C)] = D_\pa[V(C)]$.
\end{claimproof}

This implies that a labeling~$\lambda_\pa$ for~$D_\pa$ is contained in~$L_P$.
We now complete the proof by showing that~$\lambda^*\ltimes \lambda_\pa$ realizes~$D$.

\begin{claim}
The labeling~$\lambda^*\ltimes \lambda_\pa$ realizes~$D$.
\end{claim}
\begin{claimproof}
First, we observe some properties about the labeling~$\lambda_\pa$.
Let~$\lambda_\pa'$ be the realization for~$D_\pa$ from which we obtained the labeling~$\lambda_\pa$.
For each~$c\in \{a,b\}$, where~$v_c^\leq \in Q$, the edge~$\{v_c^\leq,c\}$ receives label~$\alpha_c^1$ and the edge~$\{f_{c'}^\geq,c'\}$ receives label~$\alpha_c^2$.
This is due to the fact that~$\lambda_\pa(\{c,c'\}) = \{\alpha_c^1,\alpha_c^2\}$ and~$A$ contains the arcs~$(v_c^\leq,c')$ and~$(c,f_{c'}^\geq)$ but not the arc~$(v_c^\leq,f_{c'}^\geq)$. 
Similarly, if~$v_c^\geq \in Q$, the edge~$\{v_c^\geq,c\}$ receives label~$\alpha_c^2$ and the edge~$\{f_{c'}^\leq,c'\}$ receives label~$\alpha_c^1$.

We now show that~$\lambda^*\ltimes \lambda_\pa$ realizes~$D$.

First consider the arcs between the vertices of~$V(C)$.
Since~$D[V(C)]= D_\pa[V(C)]$ and since~$P$ is robust connector (that is, there is no dense path between~$a$ and~$b$ outside of~$V(C)\setminus \{a,b\}$), $\lambda^*\ltimes \lambda_\pa$  surely realizes exactly the arcs between the vertices of~$V(C)$ in~$D$.

Next, consider the arcs between one vertex~$\vin$ of~$\Int = V(C) \setminus \{a,b\}$ and one vertex~$\vout$ of~$\Ext \setminus \{a,b\} = V \setminus V(C)$.
Let~$\vout\in \Ext \setminus \{a,b\}$ and let~$\vin\in \Int$ with~$(\vout,\vin)\in A$ (the case for~$(\vin,\vout)\in A$ can be shown analogously).
Moreover, let~$c\in \{a,b\}$ such that the edge~$\{c,c'\}$ is the edge outputted by~\Cref{compute unique entrance point} for the arc~$(\vout,\vin)$.
This implies that each temporal~$(\vout,\vin)$-path under~$\lambda^*$ traverses the edge~$\{c,c'\}$ from~$c$ to~$c'$.
We consider two cases.

\textbf{Case 1:} each temporal~$(\vout,\vin)$-path under~$\lambda^*$ traverses the edge~$\{c,c'\}$ from~$c$ to~$c'$ at time~$\alpha_c^2$\textbf{.}
Hence, due to~\Cref{claim reachable late is exactly geq set}, $V_{c'}^\geq$ are exactly the vertices of~$\Int$ that can be reached from~$\vout$ over the edge~$\{c,c'\}$.
In particular, $\vin \in V_{c'}^\geq$.
Moreover, by construction of~$Q$, the vertex~$v_c^\leq$ was added to~$Q$.
Since~$\lambda_\pa'(\{v_c^\leq, c\}) = \alpha_c^1$, each path from~$v_c^\leq$ to~$c'$ reaches vertex~$c'$ at time exactly~$\alpha_c^2$ under~$\lambda_\pa'$.
Consequently, since (i)~$V_{c'}^\geq$ are exactly the vertices of~$\Int$ that vertex~$v_c^\leq$ can reach in~$D_\pa$ and (ii)~$\lambda_\pa'$ realizes~$D_\pa$, for exactly the vertices~$V_{c'}^\geq$, there is a temporal path under~$\lambda_\pa'$ starting at a time step larger than~$\alpha_c^2$ from~$c'$ to these vertices of~$\Int$.
Hence, $(\vout,\vin)$ is realized by~$\lambda^*\ltimes \lambda_\pa$.

\textbf{Case 2:} at least one temporal~$(\vout,\vin)$-path under~$\lambda^*$ traverses the edge~$\{c,c'\}$ from~$c$ to~$c'$ at time~$\alpha_c^1$\textbf{.}
This case follows similarly with respect to the set~$V_{c'}^>$.

Finally, consider the arcs between two vertices~$\vout_1$ and~$\vout_2$ of~$\Ext$.
Let~$(\vout_1,\vout_2)\in A$.
If at least one temporal~$(\vout_1,\vout_2)$-path under~$\lambda^*$ does not traverse any edge of~$C$, then this path still exists under~$\lambda^* \ltimes \lambda_\pa$.
Hence, assume in the following that each temporal~$(\vout_1,\vout_2)$-path under~$\lambda^*$ does traverses some edge of~$C$.
In particular, this implies that~$P$ is a dense (in,out)-connector and each such~$(\vout_1,\vout_2)$-path under~$\lambda^*$ traverses the edge~$\{a,a'\}$ from~$a$ to~$a'$ at time~$\alpha_a^i\in \{\alpha_a^1,\alpha_a^2\}$ and traverses the edge~$\{b,b'\}$ from~$b'$ to~$b$ at time~$\alpha_b^j\in \{\alpha_b^1,\alpha_b^2\}$ with~$\alpha_a^i < \alpha_b^j$.

We show that there is a temporal~$(a,b)$-path under~$\lambda_\pa$ that traverses the edge~$\{a,a'\}$ from~$a$ to~$a'$ at time~$\alpha_a^i$ and traverses the edge~$\{b,b'\}$ from~$b'$ to~$b$ at time~$\alpha_b^j$. 
By construction of~$A_Q$, $(a,v_b^\geq)$ is contained in~$A_Q$ if~$\alpha_b^j = \alpha_b^1$, $(v_a^\leq,b)$ is contained in~$A_Q$ if~$\alpha_a^i = \alpha_b^2$, and $(v_a^\leq,v_b^\geq)$ is contained in~$A_Q$ if~$\alpha_a^i = \alpha_b^2$ and~$\alpha_b^j = \alpha_b^1$.  
Based on these arcs that are contained in~$A_Q$ and thus in~$D_\pa$, and the fact that~$D_\pa$ is realized by~$\lambda_\pa'$, there is a temporal~$(a,b)$-path under~$\lambda_\pa$ that traverses the edge~$\{a,a'\}$ from~$a$ to~$a'$ at time~$\alpha_a^i$ and traverses the edge~$\{b,b'\}$ from~$b'$ to~$b$ at time~$\alpha_b^j$.

Note that only \tnew{the arcs described} are realized by~$\lambda^* \ltimes \lambda_\pa$.
Hence, $\lambda^* \ltimes \lambda_\pa$ realizes~$D$.
\end{claimproof}
Consequently, the set~$L_P$ of labelings of the edges of~$C$ fulfills the desired properties and can be computed in polynomial time. 
\end{proof}

\subsection{The Final Algorithm}
Our algorithm works as follows by using the main results of the last subsections:
Theoretically, these results correctly detect that~$D$ is not realizable, or perform some preprocessing.
In the description of the algorithm, we will implicitly assume that they perform the desired preprocessing, as in the case where we detect that~$D$ is not realizable, our algorithm can simply output that~$D$ is not realizable. 

A pseudo code of our algorithm is shown in~\Cref{alg:algorithmFes}.

\iflong
        \begin{algorithm}[t]
            \SetAlgoNoEnd
            \DontPrintSemicolon
            \caption{FPT Algorithm for parameter~$\fes$.}
            \label{alg:algorithmFes}
            \SetKwInOut{Input}{Input}\SetKwInOut{Output}{Output}
            \Input{An input-instance~$D=(V,A)$ of our undirected problem.}
            \Output{\texttt{Yes} if and only if~$D$ is a yes-instance of the respective version of the problem.}
            
            Compute the solid graph~$G=(V,E)$ of~$D$.\;
            Compute an arbitrary feedback edge set~$F$ of minimum size, as well as the set~$X^*$.\;
            Apply all splitting operations and reduction rules.\; 
            
            Compute $W^* \subseteq V^*$ according to~\Cref{compute w star} and the set~$\mc := \{P_1, \dots, P_{|\mc|}\}$ of~$W^*$-connectors.\;
            
            $E^*$ $\gets$ $\{e\in E\mid e \cap W^* \neq \emptyset\} \cup \bigcup_{w\in W^*} E(T_w)$\;
          
            $\mathcal{L}$ $\gets$ $\{i\cdot 2 n \mid i \in \Oh(\fes^2)\}$\;
          
            \ForEach{$\lambda \colon E^* \to 2^{\mathcal{L}}$ with~$|\lambda(e)| \in \Oh(\fes)$ for each edge~$e\in E^*$}{

            \ForEach{$P \in \mc$}{
                Compute the set~$L_P$ of labelings for the edges of the extension of the connector~$P$ that agree with~$\lambda$ on the labels of the first and the last edge of~$P$ according to~\Cref{connector labelings}.\;            
            }

            \ForEach{$(\lambda_{P_1}, \dots, \lambda_{P_{|\mc|}}) \in L_{P_1}\times  \dots \times L_{P_{|\mc|}}$}{
                       
              $\lambda^*$ $\gets$ $\lambda \ltimes \lambda_{P_1}\ltimes \dots\ltimes \lambda_{P_{|\mc|}}      $

            \lIf{$\lambda^*$ realizes~$D$}{
                \Return{Yes} 
            }
            }

            }
              \Return{No} 
             
        \end{algorithm}
        \fi

First, we apply the operation behind~\Cref{shrink pendant} to ensure that all pendant trees are small with respect to the degree of their root in~$G[V^*]$.
Afterwards, compute a set~$W^*$ according to~\Cref{compute w star} and the collection~$\mc$ of all~$W^*$-connectors.
Moreover, let~$E^*$ denote all edges incident with vertices of~$W^*$ and all edges that are part of a pendant tree having its root in~$W^*$.
Note that this includes the first and the last edge of each~$W^*$-connector.
Let~$\mathcal{L}:=\{i\cdot 2 n \mid i \in \Oh(\fes^2)\}$.
Iterate over all possible labelings~\tnew{$\lambda\colon E^* \to 2^{\mathcal{L}}$} that assign~$\Oh(\fes)$ labels to each edge of~$E^*$.
For each such labeling~$\lambda$, compute for each~$W^*$-connector~$P$ the set of labelings~$L_P$ for the edges of the extension of~$P$ according to~\Cref{connector labelings} for~$\{\alpha_a^1,\alpha_a^2\} = \lambda(\{a,a'\})$ and~$\{\alpha_b^1,\alpha_b^2\} = \lambda(\{b,b'\})$.
Iterate over all possible combinations of such connector labelings, that is, iterate over each~$(\lambda_1,\dots,\lambda_{|\mc|})\in L_{P_1} \times \dots \times L_{P_{|\mc|}}$.
For each such combination, consider the labeling~$\lambda^* := \lambda \ltimes \lambda_{P_1} \ltimes \dots \ltimes \lambda_{P_{|\mc|}}$.
If~$\lambda^*$ realizes~$D$, then correctly output that~$D$ is realizable.
Otherwise, if no considered labeling~$\lambda^*$ \tnew{realizes}~$D$, output that~$D$ is not realizable.

Recall that~$W^*$ has size~$\Oh(\fes)$, which implies that~$E^*$ has size~$\Oh(\fes)$ due to~\Cref{properties of small sets} and since we assumed that the operation behind~\Cref{shrink pendant} was applied.
Moreover, recall \tnew{that}~$L_{P_i}$ has constant size for each~$P_i\in \mc$.
Since (i)~all preprocessing steps run in polynomial time, (ii)~there are $|\mathcal{L}|^{|E^*| \cdot \Oh(\fes)} \subseteq \fes^{\Oh(\fes^2)}$~possible candidates for~$\lambda$, and (iii)~for each candidate~$\lambda$, we check~$\Oh(1)^{|\mc|} \subseteq 2^{\Oh(\fes)}$ labelings~$\lambda^*$, the whole algorithm runs in $\fes^{\Oh(\fes^2)}\cdot n^{\Oh(1)}$~time.  
  
  It thus remains to prove the correctness.
  Assume towards a contradiction that the algorithm is not correct.
  That is, $D$ is realizable but the algorithm \tnew{does} not find a labeling~$\lambda^*$ that realizes~$D$.
  We first observe the following.
  
  \begin{observation}
If~$D$ is realizable, then there is a~\minlab realization for~$D$, where the labels of the edges of~$E^*$ are from~$\mathcal{L}$. 
\end{observation}
This follows by two facts: 
First, in every~\minlab realization, each of the~$\Oh(\fes)$ edges in~$E^*$ receives at most $\Oh(\fes)$~labels.
Hence, in each~\minlab realization, the set of all labels assigned to edges of~$E^*$ has size at most~$\Oh(\fes^2)$.
Second, since each edge of~$E\setminus E^*$ receives at most two labels in each minimal realization and there are at most~$n$ edges in~$E\setminus E^*$ (since the feedback edge set~$F$ is a subset of~$E^*$ and~$G-F$ is a tree), by first removing all empty time steps from a~\minlab realization, and afterwards introducing sufficiently many empty time steps, we can ensure that the labels assigned to~$E^*$ are multiples of~$2n$ and more precisely, are from~$\mathcal{L}$. 

Since we assumed that~$D$ is realizable,  there is, thus, a~\minlab realization~$\lambda'$ for~$D$ that only assigns labels of~$\mathcal{L}$ to the edges of~$E^*$.
We thus consider the labeling~$\lambda$ which is the restriction of~$\lambda'$ to the edges of~$E^*$.
Now, for each~$W^*$-connector~$P$, let~$L_P$ be the corresponding labeling for the edges of the extension of~$P$ that agrees with~$\lambda'$ on the labels of the first and the last edge of~$P$ according to~\Cref{connector labelings}.
Since~$\lambda'$ realizes~$D$, \Cref{connector labelings} implies that there is some~$\lambda_P$ such that~$\lambda' \ltimes \lambda_P$ still realizes~$D$.
Since all connectors are pairwise edge-disjoint, there thus exists~$(\lambda_1,\dots,\lambda_{|\mc|})\in L_{P_1} \times \dots \times L_{P_{|\mc|}}$ for which~$\lambda^* := \lambda' \ltimes \lambda_1 \ltimes  \dots \ltimes \lambda_{|\mc|}$ realizes~$D$.
By the fact that~$\lambda$ is the restriction of~$\lambda'$ to the edges of~$E^*$, $\lambda^*$ equals~$\lambda \ltimes \lambda_1 \ltimes  \dots \ltimes \lambda_{|\mc|}$.
This contradicts the assumption that our algorithm \tnew{does} not find a realization~$\lambda^*$ for~$D$.  

This thus proves~\Cref{fpt simplified} and~\Cref{fes algo}.

\fi
  
\iflong
\subparagraph{Remarks on kernelization.}
All our preprocessing steps run in polynomial time and lead to the set~$W^*$ for which each connector is nice.
Moreover, there are only $\Oh(\fes)$~edges incident with vertices of~$W^*$ or vertices of pendant trees that have their root in~$W^*$.
Our preprocessing thus nearly achieves a polynomial problem kernel for our problem.
Only the size of the connectors prevent this kernelization algorithm so far.

We conjecture that (by an even deeper analysis) it is possible to, in polynomial time, replace each connector individually by a connector of constant size.
This would then, together with our previous reduction rules imply a kernel for our problem with $\Oh(\fes)$~vertices and solid edges and $\Oh(\fes^2)$~dashed arcs.
Based on our exact algorithm (see~\Cref{exact algos}), this would then also imply a running time of~$2^{\Oh(\fes^2)} \cdot n^{\Oh(1)}$ by simply solving the kernel via the exact algorithm.
        \fi

        \section{The Complexity of Directed Reachability Graph Realizability}\label{sec:harddir}
Finally, we consider the complexity of versions of~\DRGDlong.
\iflong

\fi
As already discussed by Döring~\cite{D25}, each directed graph~$D$ can be the strict reachability graph of a simple directed temporal graph\iflong{} (by assigning label 1 to each arc of~$D$)\fi.
Hence, \any\str\DRGD and \simp\str\DRGD are always yes-instances and can thus be solved in polynomial time.

\iflong
We now show that all DAGs and all transitive graphs are trivial yes-instances for each version of the problem.
\fi

\iflong
\begin{lemma}
\else
\begin{lemma}
\fi\label{easy parts}
An instance~$D=(V,A)$ of~\DRGD is a yes-instance, if \iflong
\begin{itemize}
\item we consider \any\str\DRGD or \simp\str\DRGD,
\item $D$ is a DAG, or
\item $D$ is a transitive graph. 
\end{itemize}
\fi
\end{lemma}

\iflong
Here, a directed graph~$D=(V,A)$ is \emph{transitive}, if for every distinct  vertices~$u,v,w\in V$ with~$(u,v)\in A$ and~$(v,w)\in A$, the arc~$(u,w)$ is also in~$A$.
\fi

\iflong
\begin{proof}
The first point was already shown by Döring~\cite{D25}: By assigning label 1 to every arc of~$D$, each arc is realized and each strict temporal path has length at most one.
Hence, this trivial labeling realizes~$D$.

For the second point consider an arbitrary topological ordering~$\pi$ of~$D$.
One can realize~$D$ by a happy labeling~$\lambda$ as follows: 
Iterate over the topological ordering~$\pi$ from the sinks towards the sources.
For the~$i$th vertex~$v$ of~$\pi$, label the outgoing arcs~$B$ of~$v$ via an arbitrary bijection between~$B$ and~$[i \cdot n + 1, i \cdot n + |B|]$.
Since each vertex has less than~$n$ outgoing arcs, this labeling assigns a unique label to each arc.
Hence, $\lambda$ is a happy labeling.
Moreover, no temporal path under~$\lambda$ has length more than one, since each path of~$D$ of length at least two, the assigned labels are strictly decreasing.
This implies that~$\lambda$ realizes~$D$.

For the last point consider an arbitrary bijection~$\lambda$ between~$A$ and~$[1,|A|]$.
Clearly, each arc receives one label and~$\lambda$ is a happy labeling.
It remains to show that there is no temporal path between any two vertices~$u$ and~$v$ for which~$D$ does not contain the arc~$(u,v)$.
Suppose that there is such a temporal path~$P$.
Since this path can only use arcs that received a label under~$\lambda$ (that is, since~$P$ only uses arcs of~$D$), $P$ is a path in~$D$, which implies that~$(u,v)$ is an arc of~$D$.
This is due to the fact that~$D$ is a transitive graph.
Hence, $\lambda$ realizes~$D$.
\end{proof}

Next, we show that on directed graphs that are close to being a DAG, \pro\str\DRGD, \happy\str\DRGD, and each version of~\nstr\DRGD under consideration is~\NP-hard. 
To this end, we first make the following observations about the possibility of labeling arcs in these settings.
Our first observation is a generalization of~\cite[Lemma~4.1]{D25}.

\begin{lemma}\label{directed cycle}
Let~$D$ be an instance of~\pro\DRGD, \happy\DRGD, or some version of~\nstr\DRGD.
Moreover, let~$C$ be an induces directed cycle of length at least~$3$ in~$D$.
Then, for each realization of~$D$, there is at least one arc of~$C$ that receives no label.
\end{lemma}
\begin{proof}
Let~$C= (v_1, \dots, v_\ell, v_1)$.
Assume towards a contradiction that this is not the case and let~$\lambda$ be a proper or non-strict directed realization where each arc of~$C$ receives at least one label.
Since the cycle is induced, for each~$i$, $(v_i,v_{i+2})$ is not an arc of~$D$.
Hence, $\max \lambda((v_{i+1},v_{i+2})) \leq \min \lambda((v_{i},v_{i+1}))$ for each~$i \in [1,\ell]$.\footnote{For convenience, $\ell+1 = 1$ and~$\ell+2 = 2$.}
Since~$\lambda$ is a proper or non-strict realization, we even get that~$\max \lambda((v_{i+1},v_{i+2})) < \min \lambda((v_{i},v_{i+1}))$.
This leads to a contradiction, since $\max \lambda((v_{2},v_{3})) < \min \lambda((v_{1},v_{2})) \leq \max \lambda((v_{1},v_{2})) < \min \lambda((v_{\ell},v_{1})) \leq \dots \leq \max \lambda((v_{2},v_{3}))$.
\end{proof}

Let~$D =(V,A)$ be a directed graph.
We say that an arc~$(u,v)\in A$ is~\emph{triangulated} in~$D$, if there is a vertex~$x\in V$ such that~$(u,x)\in A$ and~$(x,v)\in A$.

\begin{lemma}\label{all nontriangulated}
Let~$D$ be an instance of any version of~\DRGD.
In each realization of~$D$, each arc~$(u,v)$ of~$D$ that is not triangulated receives at least one label.
\end{lemma}
\begin{proof}
Assume towards a contradiction that there is a directed temporal graph~$(G,\lambda)$ that realizes~$D$ such that there is an arc~$(u,v)$ of~$D$ which is not triangulated in~$D$ and receives no label.
Since~$(G,\lambda)$ realizes~$D$ and~$D$ contains the arc~$(u,v)$, there is a temporal path~$P$ from~$u$ to~$v$ in~$(G,\lambda)$.
Let~$x$ be an arbitrary vertex of~$P$ distinct from~$u$ and~$v$.
The existence of the temporal path~$P$ implies that there is a temporal path~$P_1$ from~$u$ to~$x$ in~$(G,\lambda)$ and that there is a temporal path~$P_2$ from~$x$ to~$v$ in~$(G,\lambda)$.
Since the reachability graph of~$(G,\lambda)$ is exactly~$D$, this implies that~$D$ contains the arcs~$(u,x)$ and~$(x,v)$.
This contradicts the assumption that~$(u,v)$ is not triangulated.
\end{proof}

Based on these two lemma, we can now show the stated hardness results. 
\fi

\iflong
\begin{theorem}
\else
\begin{theorem}
\fi\label{hardness directed fes}
\pro\DRGD, \happy\DRGD, and all considered versions of~\nstr\DRGD are \NP-hard on graphs with a constant size feedback arc set.
Moreover, none of these versions of~\DRGD can be solved in $2^{o(|V| + |A|)} \cdot n^{\Oh(1)}$~time, unless the ETH fails.
\end{theorem}
\iflong
\begin{proof}
We again reduce from~\SAT.
Let~$\phi$ be an instance of~\SAT with variables~$X$ and clauses~$Y$.
We construct a directed graph~$D=(V,A)$ as follows:
For each variable~$x\in X$, $V$ contain the vertices~$x, x', a_x, b_x, \ol, {\ol}', a_{\ol}$, and~$b_{\ol}$.
For each clause~$y_i\in Y$, $V$ contains the vertex~$s_i$.
Additionally, $V$ contains the vertices~$v^*,w^*,\top,\top',v_S,v_L,v_{L'}$.
Next, we describe the arcs of~$D$.
For each variable~$x\in X$, $A$ contains the arcs
\begin{itemize}
\item $(v^*,x),(x,x'),(x',\ol),({\ol},{\ol}'),({\ol}',w^*)$,
\item $(x,a_x),(a_x,\top),(x,b_x),(b_x,x'),(x,\top),(x',\top),(x',v_{L'}),(v_{L},x)$, and
\item $({\ol},a_{\ol}),(a_{\ol},\top),({\ol},b_{\ol}),(b_{\ol},{\ol}'),({\ol},\top),({\ol}',\top),({\ol}',v_{L'}),(v_{L},{\ol})$.
\end{itemize}
Let~$y_i$ be a clause of~$Y$.
For each literal~$\ell$ of~$y_i$, $A$ contains the arcs~$(s_i,\ell),(s_i,\ell')$, and~$(s_i,b_\ell)$.
Additionally, $A$ contains the arcs~$(s_i,\top)$ and~$(v_S,s_i)$.
Finally, $A$ contains the arcs~$(w^*,v^*), (\top,\top'), (\top',v_S), (\top',v_{L})$, and~$(v_{L'},v_{S})$.
This completes the construction of~$D$.

\begin{claim}
The set of arcs~$\{(w^*,v^*), (\top,\top'), (v_{L'},v_S)\}$ is a feedback arc set of~$D$.
\end{claim}
\begin{claimproof}
Consider the graph~$D'$ obtained by removing these three arcs from~$D$.
By construction, for each~$(u,v)\in \{(w^*,v^*), (\top,\top'), (v_{L'},v_S)\}$, $(u,v)$ is the only arc going out of~$u$ in~$D$ and $(u,v)$ is the only arc going into~$v$ in~$D$.
This implies that~$u$ is a sink in~$D'$ and~$v$ is a source in~$D'$.
To show that~$D'$ is a DAG it thus remains to show that~$D'':=D'-\{w^*,v^*,\top,\top',v_{L'},v_S\}$ is a DAG.
Note that for each clause~$y_i\in Y$, $(v_S,s_i)$ is the only arc going into~$s_i$ in~$D$.
Hence, each vertex of~$Q_S:= \{s_i\mid y_i\in Y\}$ is a source in~$D''$.
To show that~$D''$ is a DAG it thus remains to show that~$D''':=D''-Q_S$ is a DAG.
By construction, $D'''$ consists of one connected component~$V_x$ for each variable~$x\in X$, each of which is a DAG itself. 
Consequently, $D'$ is a DAG.
\end{claimproof}

$(\Rightarrow)$
Let~$\pi$ be a satisfying assignment for~$\phi$.
We show that there is a directed happy temporal graph~$\mg:= (D',\lambda)$ with reachability graph~$D$.

We set~$\lambda((w^*,v^*)) = 8, \lambda((\top,\top')) = 13, \lambda((\top',v_L)) = \lambda((\top',v_S)) = 12$, and~$\lambda((v_{L'},v_{S})) = 4$.
For each variable~$x\in X$, we set~$\lambda((v^*,x)) = 7$ and~$\lambda((\ol',w^*)) = 9$.
Moreover, for each~$\ell \in \{x,\ol\}$, we set~
$\lambda((\ell,a_\ell)) = 1, \lambda((a_\ell,\top)) =   \lambda((\ell',\top)) = 14, \lambda((\ell,b_\ell)) = 3, \lambda((b_\ell,\ell')) = 15, \lambda((\ell',v_{L'})) = 5$, and~$\lambda((v_{L},\ell)) = 11$.
If~$\pi(x) = \texttt{true}$, we set~$\lambda((x,x')) = 6$ and~$\lambda((x',\ol)) = 5$.
Otherwise, that is, if~$\pi(x) = \texttt{false}$, we set~$\lambda((x',\ol)) = 11$ and~$\lambda((\ol,\ol')) = 10$.
Finally, for each clause~$y_i\in Y$, we set~$\lambda((v_S,s_i)) = 3$.
Moreover, for each literal~$\ell$ of~$y_i$, we set~$\lambda((s_i,\ell)) = 2$.

All other arcs receive no label.
This completes the construction of~$\lambda$.

Note that~$\lambda$ is happy. We now show that the reachability graph of~$\mg$ equals~$D$.
To show this, it suffices to show that (i)~for each arc~$(u,v)\in A$ that receives no label under~$\lambda$, there is a temporal~$(u,v)$-path under~$\lambda$, and (ii)~for each non-arc~$(u,v)\notin A$, there is no temporal~$(u,v)$-path under~$\lambda$.

First, we show that all arcs of~$D$ are realized.
Let~$(u,v)$ be an arc of~$D$ that receives no label under~$\lambda$.
By construction of~$\lambda$, $(u,v)$ has one of four types:
\begin{itemize}
\item there is some variable~$x\in X$ and some~$\ell\in \{x,\ol\}$, such that~$u= \ell$ and~$v= \ell'$,
\item there is some variable~$x\in X$ and some~$\ell\in \{x,\ol\}$, such that~$u= \ell$ and~$v= \top$,
\item there is some clause~$y_i\in Y$ and some literal~$\ell$ contained in~$y_i$, such that~$u= s_i$ and~$v \in \{\ell', b_\ell\}$, or
\item there is some clause~$y_i\in Y$, such that~$u= s_i$ and~$v =  \top$.
\end{itemize}

If~$(u,v)$ is of the first type, the arc is realized by the temporal path~$(\ell,b_\ell,\ell')$ with labels~$(3,15)$.
If~$(u,v)$ is of the second type, the arc is realized by the temporal path~$(\ell,a_\ell,\top)$ with labels~$(1,14)$.
If~$(u,v)$ is of the third type, the arc is realized by the temporal path~$(s_i,\ell,b_\ell,\ell')$ with labels~$(2,3,15)$ (or its subpath).
Hence, consider the fourth type and let~$y_i\in Y$.
Since~$\pi$ satisfies~$\phi$, there is some literal~$\ell$ of~$y_i$, such that~$\ell$ is satisfied by~$\pi$.
This implies that the arc~$(\ell,\ell')$ received a label~$\alpha$ under~$\lambda$.
By definition, $\alpha \in \{6,10\}$.
Hence, the arc~$(s_i,\top)$ is realized by the path~$(s_i,\ell,\ell',\top)$ with labels~$(2,\alpha,14)$.
Consequently, each arc of~$D$ is realized.

Next, we show that no non-arc of~$D$ is realized.
To this end, we analyze the structure of temporal paths of length more than one under~$\lambda$.
Let~$P$ be a temporal path of length at least two under~$\lambda$ and let~$u$ and~$v$ be the endpoints of~$P$.
We show that~$(u,v)$ is an arc of~$D$.
By definition of~$\lambda$, for every vertex~$u \in \{\top,\top',v_S,v_L,v_{L'},v^*,w^*\} \cup \{s_i\mid y_i\in Y\}$, each arc going into~$u$ has a strictly larger label than each arc going out of~$u$.
Thus, $P$ contains none of these vertices as an intermediate vertex.
In other words, the set of intermediate vertices of~$P$ is a subset of~$\{\ell,\ell',a_\ell,b_\ell \mid x\in X, \ell \in \{x,\ol\}\}$.

Moreover, $P$ cannot traverse the arc~$(v^*,x)$ for any variable~$x\in X$, since~$v^*$ is not an intermediate vertex of~$P$ as each outgoing arc of~$x$ receives a strictly smaller label. 
Similarly, $P$ traverses neither the arc~$(x',\ol)$ nor the arc~$(\ol',w^*)$.
This implies that there is some variable~$x\in X$ and some literal~$\ell \in \{x,\ol\}$, such that~$P$ contains only vertices of~$Q_\ell = \{\ell,\ell',a_\ell,b_\ell,\top\} \cup S_\ell$, where~$S_\ell$ denotes the vertices of~$\{s_i\mid y_i\in Y, \ell \in y_i\}$.
By construction of~$D$, $D[Q_\ell]$ is a DAG with sources~$S_\ell$ and unique sink~$\top$.
Moreover, $b_\ell$ is the only vertex of~$Q_\ell$ that has no arc towards~$\top$.
Clearly, $P$ cannot be a temporal path from~$b_\ell$ to~$\top$, since each arc going into~$\top$ has label~$13$ and the unique arc~$(b_\ell,\ell')$ going out of~$b_\ell$ receives label~$14$.
Hence, if~$P$ is an~$(u,v)$-path with~$v = \top$, then~$(u,v)$ is an arc of~$D$.
Similarly, for each vertex~$s_i \in S_\ell$, $a_\ell$ is the only vertex of~$Q_\ell \setminus S_\ell$ for which~$(s_i,a_\ell)$ is not an arc of~$D$.
Clearly, $P$ cannot be a temporal path from~$s_i$ to~$a_\ell$, since each arc going out of~$s_i$ has label~$2$ and the unique arc~$(\ell,a_\ell)$ going into~$a_\ell$ receives label~$1$.
Hence, if~$P$ is a~$(u,v)$-path with~$u\in S_i$, then~$(u,v)$ is an arc of~$D$.

Thus, in the following, we only have to consider the case where~$P$ contains only vertices of~$\{\ell,\ell',a_\ell,b_\ell\}$.
The only path of length at least two in~$D[\{\ell,\ell',a_\ell,b_\ell\}]$ is the path~$(\ell,b_\ell,\ell')$.
Thus, the only remaining option for~$P$ is to be this unique path.
Since~$(\ell,\ell')$ is an arc of~$D$, there is an arc between the endpoints of~$P$ in~$D$.
Consequently, no non-arc of~$D$ is realized, which implies that~$\lambda$ realizes~$D$.

$(\Leftarrow)$
Let~$\mg:= (D,\lambda\colon A \to 2^{\mathbb{N}})$ be a directed temporal graph with (i)~proper labeling~$\lambda$ and strict reachability graph equal to~$D$ or (ii)~non-strict reachability graph equal to~$D$.
We show that there is a satisfying assignment~$\pi$ for~$\phi$.
For each variable~$x\in X$, we set \[\pi(x) := \begin{cases}\texttt{true} & \lambda((x,x')) \neq \emptyset \text{ and}\\
\texttt{false} & \text{otherwise.} \end{cases}\]

In the following, we show that~$\pi$ satisfies~$\phi$.
To this end, we first show that~$\lambda((\ol,\ol')) = \emptyset$ if~$\pi(x) = \texttt{true}$, that is, if~$\lambda((x,x')) \neq \emptyset$.
This statement follows form~\Cref{directed cycle,all nontriangulated} and the fact that~$C:=(v^*,x,x',\ol,\ol',w^*,v^*)$ is an induced directed cycle in~$D$ for which~$(x,x')$ and~$(\ol,\ol')$ are the only triangulated arcs.
Hence, if~$\lambda((\ol,\ol')) \neq \emptyset$, then~$\pi(x) = \texttt{false}$.

Based on this observation, we now show that~$\pi$ satisfies~$\phi$.
Let~$y_i$ be a clause of~$Y$.
Since~$\mg$ realizes~$D$ and~$(s_i,\top)$ is an arc of~$D$, there is a temporal path~$P$ from~$s_i$ to~$\top$ in~$\mg$.
Recall that~$s_i$ has only the out-neighbor~$\top$ and the out-neighbors~$\{\ell, \ell', b_\ell\}$ for each literal~$\ell$ of~$y_i$.
Hence, the temporal path~$P$ only visits these vertices, since~$P$ is a dense path, that is, since there is an arc from~$s_i$ to each other vertex of~$P$. 
Moreover, $P$ cannot visit any vertex~$b_\ell$, since no such vertex has an arc towards~$\top$.
This implies that~$P$ is either the path~$(s_i, \top)$, or for some literal~$\ell$ of~$y_i$, $P \in \{(s_i, \ell, \top),(s_i, \ell', \top),(s_i, \ell,\ell',\top)\}$.
 
First, note that the direct arc~$(s_i,\top)$ does not receive a label under~$\lambda$.
This is due to~\Cref{directed cycle,all nontriangulated} and the fact that~$C:=(s_i,\top,\top',v_S,s_i)$ is an induced directed cycle in~$D$ for which~$(s_i,\top)$ is the only triangulated arc.
This implies that~$P \neq (s_i,\top)$.

Hence, there is a literal~$\ell$ of~$y_i$ such that~$P \in \{(s_i, \ell, \top),(s_i, \ell', \top),(s_i, \ell,\ell',\top)\}$.
Similarly to the above, $(\ell,\top)$ does not receive a label under~$\lambda$, since~$C:=(\ell,\top,\top',v_L,\ell)$ is an induced directed cycle in~$D$ for which~$(\ell,\top)$ is the only triangulated arc.
This implies that~$P \neq (s_i,\ell,\top)$.
Analogously, $(s_i,\ell')$ does not receive a label under~$\lambda$, since~$C:=(s_i,\ell',v_{L'},v_S,s_i)$ is an induced directed cycle in~$D$ for which~$(s_i,\ell')$ is the only triangulated arc.
This implies that~$P \neq (s_i,\ell',\top)$. 
As a consequence, $P = (s_i, \ell, \ell', \top)$, which implies that~$\lambda((\ell, \ell')) \neq \emptyset$.
If~$\ell = x$ for some variable~$x\in X$, then~$\pi(x) = \texttt{true}$ by definition of~$\pi$, which implies that clause~$y_i$ is satisfied by~$\pi$.
Otherwise, that is, if~$\ell = \ol$ for some variable~$x\in X$, then~$\lambda((x, x')) = \emptyset$, since~$\lambda((\ol, \ol')) \neq \emptyset$.
Hence~$\pi(x) = \texttt{false}$ by definition of~$\pi$, which implies that clause~$y_i$ is satisfied by~$\pi$.
In both cases, clause~$y_i$ is satisfied by~$\pi$.
Consequently, $\pi$ satisfies~$\phi$.

\subparagraph{ETH lower-bound.}
Note that~$D$ has~$\Oh(|\phi|)$  vertices and arcs.
Since~\SAT cannot be solved in $2^{o(|\phi|)} \cdot |\phi|^{\Oh(1)}$~time, unless the ETH fails~\cite{T84}, this implies the stated running-time lower-bound for the stated versions of~\DRGD.
\end{proof}
\fi 

This thus settles the complexity of each version of \DRGD under consideration.
\iflong Moreover, an FPT algorithm for any of these versions when parameterized by the size of a smallest feedback arc set is impossible, unless P = NP. \fi

\section{Conclusion}  
We studied~\RGDlong and gave for both directed and undirected temporal graphs the complete picture for the classical complexity of all settings, answering this open problem posed by~\cite{DBLP:journals/tcs/CasteigtsCS24} and~\cite{D25}.
For~\URGD, we additionally showed that the problem can be solved in FPT-time for the feedback edge set number~$\fes$ of the solid graph.
As we showed, this parameter cannot be replaced by smaller parameters like feedback vertex set number or treedepth of the solid graph, unless FPT~=~W[2].

There are several directions for future work:
First, it would be interesting to see whether (some) versions of~\URGD admit a polynomial kernel for~$\fes$.
Another interesting task is to determine whether~\pro\URGD, \happy\URGD, or~\nstr\URGD admits an FPT algorithm for feedback vertex set or treedepth.
This is not excluded by our W[2]-hardness result, since that reduction only worked for~\any\str\URGD and~\simp\str\URGD. \iflong
It would also be interesting to see whether~\pro\URGD, \happy\URGD, or~\nstr\URGD can be solved in polynomial time on solid graphs that are triangle-free.\fi

Finally, one could analyze the parameterized complexity of~\DRGDlong with respect to directed graph parameters.

\bibliographystyle{plainurl}
\bibliography{temporal}

\end{document}